\RequirePackage{rotating}
\documentclass[acmsmall,screen]{acmart}

\setcopyright{rightsretained}
\acmDOI{10.1145/3689760}
\acmYear{2024}
\acmJournal{PACMPL}
\acmVolume{8}
\acmNumber{OOPSLA2}
\acmArticle{320}
\acmMonth{10}
\acmSubmissionID{oopslab24main-p486-p}
\received{2024-04-01}
\received[accepted]{2024-08-18}

\usepackage{graphicx} 
\usepackage{comment}
\usepackage{gdefs}
\usepackage{refs}
\usepackage{amsmath}

\usepackage{kbordermatrix}
\usepackage{cancel}
\DeclareMathOperator{\tensor}{\otimes}     
\DeclareMathOperator{\bigO}{\mathcal{O}}   
\DeclareMathOperator{\bigOmega}{\Omega}    
\usepackage{multirow}
\usepackage{float}
\usepackage[justification=centering]{caption}
\usepackage[justification=centering]{subcaption}
\usepackage{amsfonts}
\usepackage{amsthm}
\usepackage{stmaryrd}
\usepackage{autobreak}
\usepackage{enumitem}
\setlist{noitemsep,topsep=0.7mm,leftmargin=*}
\usepackage{setspace}
\usepackage{hyperref}
\usepackage{graphbox}
\usepackage[framemethod=default]{mdframed}
\usepackage[ruled,linesnumbered]{algorithm2e}
\usepackage{nicematrix}
\SetKwInput{Input}{Input}
\SetKwInput{Output}{Output}
\SetKwBlock{Begin}{begin}{end}

\newcommand{\MPP}{\mbox{Matched-Path Principle\/}}
\newcommand{\CIP}{\mbox{Contextual-Interpretation Principle\/}}

\newcommand{\Matched}{{\it matched\/}}
\newcommand{\NumConnections}{{\rm NumConnections\/}}
\newcommand{\DontCareGrouping}{{\it DontCareGrouping\/}}
\newcommand{\ForkGrouping}{{\it ForkGrouping\/}}
\newcommand{\BReturnTuples}{{\it BReturnTuples\/}}
\newcommand{\AConnection}{{\it AConnection\/}}
\newcommand{\BConnection}{{\it BConnections\/}}
\newcommand{\BConnections}{{\it BConnections\/}}
\newtheorem{Construction}{\indent\,Construction}
\usepackage{xspace}
\usepackage{csquotes}
\usepackage{wrapfig}
\usepackage{rotating}
\usepackage{tablefootnote}
\usepackage{nccmath}

\makeatletter
\def\thmhead@plain#1#2#3{%
  \thmname{#1}\thmnumber{\@ifnotempty{#1}{ }\@upn{#2}}%
  \thmnote{ {\the\thm@notefont#3}}}
\let\thmhead\thmhead@plain
\makeatother
\theoremstyle{definition}
\newtheorem*{theorem*}{Theorem}

\newenvironment{Constr}{\par %
                \begin{Construction}}{$\QED$\end{Construction}}

\newcounter{LineNumber}

\newenvironment{ProofSketch}{%
  \renewcommand{\proofname}{Proof Sketch}\proof}{\endproof}

\usepackage[color]{changebar}
\cbcolor{red}

\newcommand{\sem}[1]{\llbracket {#1} \rrbracket}

\newcommand{\EXP}{{\textit{EXP}}}
\newcommand{\POP}{{\textit{POP}}}

\DeclareMathOperator{\op}{{\textit{op}}}
\newcommand{\QFT}{{\textit{QFT}}\xspace}
\newcommand{\GHZ}{{\textit{GHZ}}\xspace}
\newcommand{\BV}{{\textit{BV}}\xspace}

\newcommand{\CNOT}{{\textit{CNOT}}}
\newcommand{\ev}{{\textit{ev}}}
\newcommand{\EV}{{\textit{EV}}}
\newcommand{\bp}{{\textit{bp}}}
\newcommand{\BP}{{\textit{BP}}}
\newcommand{\ZeroBP}{{\bar{0}_{\BP}}}

\newcommand{\protoWCFLOBDD}{\textrm{proto-WCFLOBDD}\xspace}

\newcommand{\Omit}[1]{}

\newcommand{\twrchanged}[1]{#1}
\newcommand{\swarat}[1]{#1}

\definecolor{lightgray}{rgb}{0.55,0.52,0.54}

\newcommand{\OnlySupplemental}[1]{}    

\renewcommand{\subsubsection}[1]{\medskip\noindent{\textbf{#1}}}

\date{}

\ccsdesc[500]{Computing methodologies~Representation of mathematical functions}
\ccsdesc[500]{Computing methodologies~Quantum mechanic simulation}
\ccsdesc[300]{Software and its engineering~Software verification}

\keywords{Weighted decision diagram, matched paths, best-case double-exponential compression, quantum simulation}

\begin{document}
\title{Weighted Context-Free-Language Ordered Binary Decision Diagrams}

\author{Meghana Sistla}
\orcid{0000-0002-4215-0651}
\affiliation{%
  \institution{University of Texas at Austin}
  \city{Austin}
  \country{USA}
}
\email{mesistla@utexas.edu}

\author{Swarat Chaudhuri}
\orcid{0000-0002-6859-1391}
\affiliation{%
  \institution{University of Texas at Austin}
  \city{Austin}
  \country{USA}
}
\email{swarat@cs.utexas.edu}

\author{Thomas Reps}
\orcid{0000-0002-5676-9949}
\affiliation{%
  \institution{University of Wisconsin}
  \city{Madison}
  \country{USA}
}
\email{reps@cs.wisc.edu}

\begin{abstract}
This paper presents a new data structure, called \emph{Weighted Context-Free-Language Ordered BDDs} (WCFLOBDDs), which are a hierarchically structured decision diagram, akin to Weighted BDDs (WBDDs) enhanced with a procedure-call mechanism.
For some functions, WCFLOBDDs are exponentially more succinct than WBDDs.
They are potentially beneficial for representing functions of type $\mathbb{B}^n \rightarrow D$,
when a function's image $V \subseteq D$ has many different values.
We apply WCFLOBDDs in quantum-circuit simulation, and find that they perform better than WBDDs on certain benchmarks.
With a 15-minute timeout, the number of qubits that can be handled by WCFLOBDDs is 1-64$\times$ that of WBDDs
(and 1-128$\times$ that of CFLOBDDs, which are an unweighted version of WCFLOBDDs).
These results support the conclusion that for this
\twrchanged{
application---from the standpoint of problem size, measured as the number of qubits---WCFLOBDDs provide the best of both worlds: performance roughly matches
}
whichever of WBDDs and CFLOBDDs is better.
\twrchanged{
(From the standpoint of running time, the results are more nuanced.)
}

\end{abstract}

\maketitle

\section{Introduction}
\label{Se:intro}

Many BDD variants \cite{DBLP:books/siam/Wegener00} have been introduced to address different failings of vanilla BDDs \cite{toc:Bryant86}, often motivated by new applications.
For instance, multi-terminal BDDs (MTBDDs) \cite{DBLP:journals/fmsd/FujitaMY97} (or ADDs \cite{DBLP:journals/fmsd/BaharFGHMPS97}) were introduced to represent non-Boolean functions $\{0,1\}^n \rightarrow D$, where $D$ is a set.
MTBDDs have their own drawbacks: for a function $h$ whose image has a large number of different values $V \subseteq D$, the size of an MTBDD for $h$ cannot be any smaller than $|V|$.
To date, the most successful solution to this issue has been \emph{Weighted BDDs} (WBDDs)---BDD-like structures with weights on edges \cite{DBLP:journals/tcad/NiemannWMTD16,DBLP:conf/date/ViamontesMH04}.
For a given Boolean assignment $a$, $h(a)$ is computed as the product of the weights on the path for $a$.
This approach yields succinct representations when factors of the values in $V$ can be reused in the WBDD.
Recently, because of the efficient compression capabilities of WBDDs, they, and other related weighted variants
\twrchanged{
have been applied to
}
quantum-circuit simulation \cite{Book:ZW2020}.

\swarat{In this paper, we explore a new way to further compress the WBDD representation: \emph{the introduction of hierarchical structure}.
A recent paper \cite{TOPLAS:SCR24} has developed this idea for unweighted BDDs. 
Whereas WBDDs (and BDDs) can be seen as acyclic finite-state machines (FSMs) that read assignments in bit-serial order, Context-Free-Language Ordered BDDs (CFLOBDDs), the data structure proposed in that work, can be viewed as hierarchical FSMs \cite{TOPLAS:ABEGRY05}. Just as hierarchical FSMs can be exponentially more succinct than FSMs, CFLOBDDs can be exponentially more succinct than classical BDDs.
}
Our contribution here is to combine the complementary benefits of WBDDs and CFLOBDDs through a data structure called \textit{Weighted CFLOBDDs}
\twrchanged{
(WCFLOBDDs).
}
In the best case, WCFLOBDDs are exponentially more succinct than both WBDDs and CFLOBDDs.
\figref{DesignSpace} shows the design space of BDDs, WBDDs, CFLOBDDs, and WCFLOBDDs:

\begin{wrapfigure}{R}{0.54\textwidth}
  \centering
  \vspace{-3.0ex}
  {\small
  $\begin{array}{@{\hspace{0ex}}r@{\hspace{0.95ex}}|@{\hspace{0.65ex}}r@{\hspace{0.65ex}}||c@{\hspace{0.65ex}}|@{\hspace{0.65ex}}c@{\hspace{0ex}}}
    \multicolumn{2}{c||}{}                              & \multicolumn{2}{c}{\textrm{Hierarchical}} \\
    \cline{3-4}
    \multicolumn{2}{c||}{}                              & \textrm{no}   & \textrm{yes} \\
    \hline\hline
    \multirow{2}{*}{\textrm{Weights}} &  \textrm{no} & \textrm{BDD~\cite{toc:Bryant86,DBLP:journals/fmsd/FujitaMY97,DBLP:journals/fmsd/BaharFGHMPS97}}  & \textrm{CFLOBDD~\cite{TOPLAS:SCR24}}  \\
    \cline{2-4}
                                         & \textrm{yes} & \textrm{WBDD~\cite{DBLP:journals/tcad/NiemannWMTD16,DBLP:conf/date/ViamontesMH04}} & \textrm{WCFLOBDD (this paper)} \\
  \end{array}$
  }
  \caption{The design space of BDDs, WBDDs, CFLOBDDs, and WCFLOBDDs.}
  \label{Fi:DesignSpace}
  \vspace{-4ex}
\end{wrapfigure}

We evaluated WCFLOBDDs, WBDDs, and CFLOBDDs on both synthetic and quantum-simulation benchmarks, and found
\twrchanged{
that---from the standpoint of problem size, measured as number of qubits---WCFLOBDDs
}
perform better than WBDDs and CFLOBDDs on most benchmarks.
On the quantum-simulation benchmarks, with a 15-minute timeout,
the number of qubits that can be handled by WCFLOBDDs is
2,097,152 for Greenberger-Horne-Zeilinger (GHZ) (1$\times$ over CFLOBDDs, 64$\times$ over WBDDs); 
262,144 for Bernstein-Vazirani (BV) (1$\times$ over CFLOBDDs, 16$\times$ over WBDDs);
524,288 for Deutsch-Jozsa (DJ) (2$\times$ over CFLOBDDs,
32$\times$ over WBDDs);
8,192 for Simon's algorithm
(1$\times$ over CFLOBDDs, 1$\times$ over WBDDs)
2,048 for Quantum Fourier Transform (QFT) (128$\times$ over CFLOBDDs, 1$\times$ over WBDDs);
and 16 for Grover's algorithm (2$\times$ over CFLOBDDs, 1$\times$ over WBDDs).
Such results support the conclusion that, for at least some applications, WCFLOBDDs provide the best of both worlds:
\twrchanged{
performance roughly matches whichever of WBDDs and CFLOBDDs is better (and beats both on DJ).
}
\twrchanged{
(From the standpoint of running time, the results are more nuanced.)
}

Our work makes the following contributions:
\begin{itemize}
  \item
    We define WCFLOBDDs, a data structure to represent functions, relations, matrices, etc.\ (\sectref{wcflobdds}).
  \item
    We show that WCFLOBDDs can be exponentially more succinct than both WBDDs and CFLOBDDs (\sectref{separation}).
  \item
    We give algorithms for WCFLOBDD operations (\sectref{algos}).
  \item
    We find that WCFLOBDDs perform as well as or better than WBDDs and CFLOBDDs on synthetic and most quantum-simulation benchmarks (\sectref{eval}).
\end{itemize}
\sectref{back} reviews WBDDs \twrchanged{and CFLOBDDs}.
\sectref{related} discusses related work.
\sectref{conclusion} concludes.
\section{Background}
\label{Se:back}

\subsection{Weighted Binary Decision Diagrams (WBDDs)}
Consider the family of Hadamard matrices $\mathcal{H} = \{ H_{2^i} \mid i \geq 1 \}$ defined recursively as
\newline
\fbox{$
H_{2^i} =
\begin{cases}
    \frac{1}{\sqrt{2}}\left[\begin{smallmatrix}
    1 & 1\\
    1 & -1
    \end{smallmatrix}\right] & \text{$i$ = $1$}\\
    H_{2^{i-1}} \otimes H_{2^{i-1}} & \text{otherwise}
\end{cases}
$}
where $\otimes$ denotes Kronecker product.
  The \emph{Kronecker product} $A \tensor B$ equals
  $
    \left[\begin{smallmatrix}
                  A_{1,1} & \cdots & A_{1,m} \\
                  \vdots  & \ddots & \vdots  \\
                  A_{n,1} & \cdots & A_{n,m}
                \end{smallmatrix}\right] \otimes B
             \eqdef \left[\begin{smallmatrix}
                              A_{1,1}B & \cdots & A_{1,m}B \\
                               \vdots  & \ddots & \vdots   \\
                              A_{n,1}B & \cdots & A_{n,m}B
                    \end{smallmatrix}\right].
  $
As $i$ increases, matrices in $\mathcal{H}$ increase in size exponentially.
$H_{2^i}$ contains
\twrchanged{
$2^{i}$ rows and $2^{i}$ columns, and thus $2^{i+1}$ elements.
}
To index a row (column), ${i-1}$ row (column) variables are required.

\begin{figure}[tb!]
    \centering
    \begin{subfigure}{0.3\linewidth}
        \centering
        \includegraphics[width=0.60\linewidth]{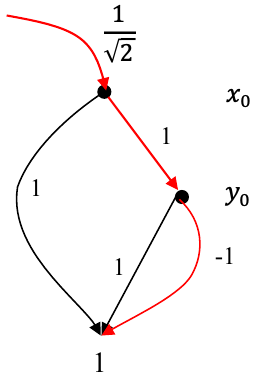}
        \caption{WBDD for $H_2$}
        \label{Fi:wbdds-h2}
    \end{subfigure}
    \begin{subfigure}{0.3\linewidth}
        \centering
        \includegraphics[width=0.73\linewidth]{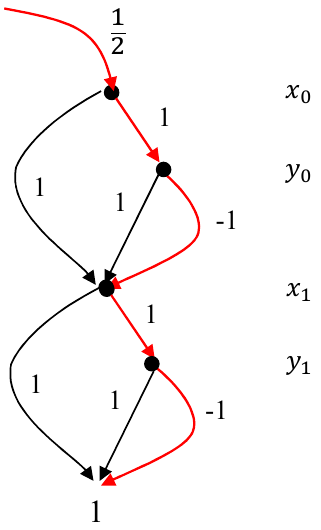}
        \caption{WBDD for $H_4$}
        \label{Fi:wbdds-h4}
    \end{subfigure}
    \begin{subfigure}{0.3\linewidth}
        \centering
        \includegraphics[width=0.55\linewidth]{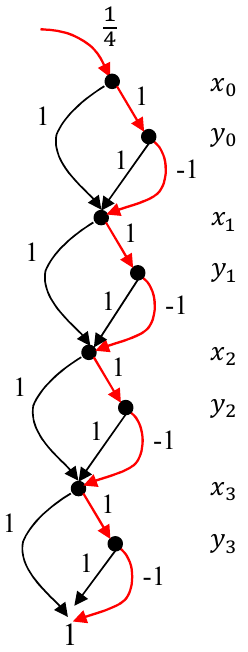}
        \caption{WBDD for $H_8$}
        \label{Fi:wbdds-h8}
    \end{subfigure}
    \caption{WBDDs for the matrices $H_2$, $H_4$, and $H_8$, with $x$ variables for rows and $y$ variables for columns. }
    \label{Fi:wbdds-hadamard}
\end{figure}

\figref{wbdds-hadamard} shows the WBDD representations for the first three matrices in $\mathcal{H}$: $H_2$, $H_4$, and $H_8$.
Every node is associated with a decision variable, and the outgoing edges correspond to the $0$ or $1$ decision.\footnote{
  In all diagrams, left branches are for $0$ ($F$); right branches are for $1$ ($T$).
}
The $x$ variables represent the row index;
the $y$ variables represent the column index.
The numbers associated with edges indicate the weight of the edge.
The red highlighted paths indicate the all-$1$ variable-assignment $a = \forall i \in \{1..|\text{vars}|\} : \{ x_i \mapsto 1, y_i \mapsto 1 \}$, which corresponds to the element in the lower right-hand corner of the matrix.
The product of the weights along the path, together with the common weight (at the top of the WBDD) produces the value of the represented function at $a$.
The size of the WBDDs $\{ H_{2^i} \mid i \geq 1 \}$ increases linearly in the number of variables $i$.

\paragraph{Semi-Field Weights}
\twrchanged{
In WBDDs (and WCFLOBDDs), weights are drawn from a semi-field.
}

\begin{definition}
\twrchanged{
Let $\mathcal{D} = \langle A, +, \cdot, \bar{0}, \bar{1} \rangle$ be a set, where $\bar{0}, \bar{1} \in A$, that supports two binary operations: $+$ and $\cdot$.
}
$\mathcal{D}$ is a \textbf{\emph{semi-field}} if, for all $a, b, c \in A$, the following properties are satisfied:

\noindent
{\small
\begin{equation*}
  \begin{array}{@{\hspace{0ex}}r@{\hspace{1.0ex}}r@{\hspace{0.75ex}}c@{\hspace{0.75ex}}l@{\hspace{1.0ex}}|r@{\hspace{1.0ex}}r@{\hspace{0.75ex}}c@{\hspace{0.75ex}}l@{\hspace{0ex}}}
    \textrm{Associativity:}  & a + (b + c) & = & (a + b) + c                 & \textrm{Annihilation:}   & a \cdot \bar{0} & = & \bar{0} = \twrchanged{\bar{0} \cdot a} \\
                             & a \cdot (b \cdot c) & = & (a \cdot b) \cdot c & \textrm{Distributivity:} & a \cdot (b + c) & = & (a \cdot b) + (a \cdot c) \\
    \textrm{Commutativity:}  & a + b & = & b + a                             &                          & (b + c) \cdot a & = & (b \cdot a) + (c \cdot a) \\
    \textrm{Identities:}       &  a + \bar{0} & = & a = \twrchanged{\bar{0} + a}
                                             & \textrm{Mult.\ Inverse:} & \multicolumn{3}{c}{a \neq 0 \Rightarrow \exists a^{-1} \in \mathcal{D} : \twrchanged{~ a \cdot a^{-1} = \bar{1} = a^{-1} \cdot a }} \\
                               & a \cdot \bar{1} & = & a = \twrchanged{\bar{1} \cdot a}                                            &  &  &  &      
  \end{array}
\end{equation*}}
\end{definition}

\noindent
\twrchanged{
The non-$\bar{0}$ elements are a group under multiplication,
$\langle A - \{\bar{0}\}, \cdot, \bar{1} \rangle$, whereas
$\langle A, +, \bar{0} \rangle$ is only a semi-group.
The following properties are consequences of the semi-field axioms:
\begin{itemize}
    \item
\twrchanged{
      The additive identity $\bar{0}$ and multiplicative identity $\bar{1}$ are each uniquely defined.
      \begin{itemize}
        \item
          If there is an element $d \in A$ that satisfies $\forall a : a + d = a = d + a$, then $\bar{0} = \bar{0} + d = d$.
        \item
          If there is an element $d \in A$ that satisfies $\forall a : a \cdot d = a = d \cdot a$, then $\bar{1} = \bar{1} \cdot d = d$. 
      \end{itemize}
      Consequently, we can speak of $\bar{0}$ as \emph{the} additive identity element and $\bar{1}$ as \emph{the} multiplicative identity element.
}
    \item
\twrchanged{
      For each element $a \in A$, $a \neq \bar{0}$, the multiplicative inverse is unique.
      \begin{itemize}
        \item
          Suppose that $b$ and $c$ are both inverses of $a$.
          Then $b = b \cdot \bar{1}$ $= b \cdot (a \cdot c)$ $= (b \cdot a) \cdot c$ $= \bar{1} \cdot c = c$.
      \end{itemize}
}
    \item 
\twrchanged{
      A semi-field has no zero divisors:
      \begin{itemize}
        \item
          $a \cdot b = \bar{0} \land a \neq \bar{0}$ $\Rightarrow$ $b = \bar{1} \cdot b$ $= (a^{-1} \cdot a) \cdot b$ $= a^{-1} \cdot (a \cdot b)$ $= a^{-1} \cdot \bar{0} = \bar{0}$.
          Similarly, $b \cdot a = \bar{0} \land a \neq \bar{0}$ $\Rightarrow$ $b = b \cdot \bar{1}$ $= b \cdot (a \cdot a^{-1})$ $= (b \cdot a) \cdot a^{-1}$ $= \bar{0} \cdot a^{-1} = \bar{0}$.
      \end{itemize}
      Equivalently, because $\langle A - \{\bar{0}\}, \cdot, \bar{1} \rangle$ is a group, it is closed under multiplication;
      thus, for $a,b \in A - \{\bar{0}\}$, $a \cdot b$ can never equal $\bar{0}$.
}
\end{itemize}
}

\twrchanged{
Because $\mathbb{R}$ and $\mathbb{C}$ are fields, they are also semi-fields.
An example of a semi-field that is not a field is the set of invertible $n \times n$ matrices (together with the all-0 matrix of size $n \times n$ as $\bar{0}$), with matrix addition and matrix multiplication.
}

\twrchanged{
We need inverse elements with respect to ``$\cdot$'' to be able to canonicalize WBDDs, as well as weighted decision trees (see \figref{Complicated}(b)) and WCFLOBDDs (\sectref{canonicity}).
Because one multiplies weights as one follows the path for an assignment, the basic idea is to label left branches (for a Boolean variable bound to 0) with $\bar{1}$, and label right branches (for 1) with some value.
Suppose that the branches at a node are originally labeled with $a$ and $b$, respectively.
During canonicalization, the left branch would be labeled with $\bar{1}$, the right branch with $a^{-1} \cdot b$, and the value $a$ would be propagated to every incoming edge. 
This relabeling maintains the product over all paths---e.g., $\ldots \cdot a \cdot \bar{1} \cdot \ldots$ along the left branch, and  $\ldots \cdot a \cdot a^{-1} \cdot b  \cdot \ldots$  along the right branch. (There are some special cases for canonicalization when the left edge is originally labeled with $\bar{0}$.)
To carry out such weight-propagation steps, ``$\cdot$''  must be associative, but need not be commutative.
}

\subsection{CFLOBDDs as Hierarchical Finite-State Machines}
\twrchanged{
A CFLOBDD over $2^n$ variables can be considered to be a form of non-recursive, single-entry, multi-exit, hierarchical finite-state machine \cite{DBLP:journals/scp/Harel87,TOPLAS:ABEGRY05} that
(i) has a fixed maximum nesting-depth of $n$, and
(ii) accepts a subset of $\{0, 1\}^{2^n}$.
(If necessary, the client pads the variables to the next higher power of $2$.)
The hierarchy in a CFLOBDD is captured by the \emph{levels} of its (sub-)machines.
Level $l$ indicates the remaining number of nested calls that are allowed;
thus, the topmost level of a CFLOBDD with $2^n$ variables is level $n$, with nested calls counting down to level $0$.
}

\twrchanged{
At each level $k$, a sub-machine of a CFLOBDD has a fixed pattern of calls:
it always performs an ``$A$ call'' to a level-$(k\text{-}1)$ sub-machine, followed by a ``$B$ call'' to a level-$(k\text{-}1)$ sub-machine.
A given level-$k$ sub-machine always performs the same $A$ call, but depending on the return state, it can perform different $B$ calls.
The level-$k$ sub-machine returns to level $k+1$ using an exit state determined from the return state of the $B$ call. 
}

\twrchanged{
This view of a CFLOBDD as a hierarchical finite-state machine helps to explain the variable-decomposition principle used:
at each level $k$, the variables are---in effect---divided in half:
the $A$ call interprets $x_0, \dots, x_{2^{k-1}-1}$, and a $B$ call interprets $x_{2^{k-1}},\dots, x_{{2^k}-1}$.
At level $0$, one is left with only a single variable, so level $0$ is where an individual variable is interpreted.
}

\section{Weighted CFLOBDDs (WCFLOBDDs)}
\label{Se:wcflobdds}

\begin{figure}[tb!]
\centering
  \begin{subfigure}{0.22\linewidth}
    \centering
    \includegraphics[width=.98\linewidth]{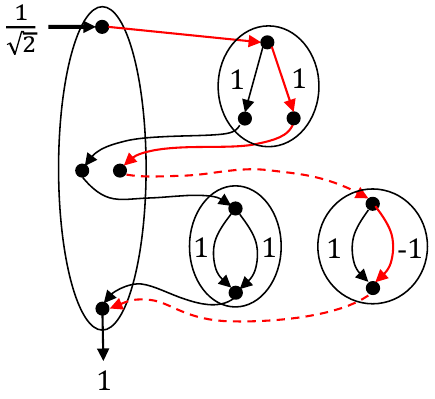}
    \caption{WCFLOBDD for $H_2$}
    \label{Fi:wcflobdd_H2}
  \end{subfigure}
  \begin{subfigure}{0.32\linewidth}
    \centering
    \includegraphics[width=0.98\linewidth]{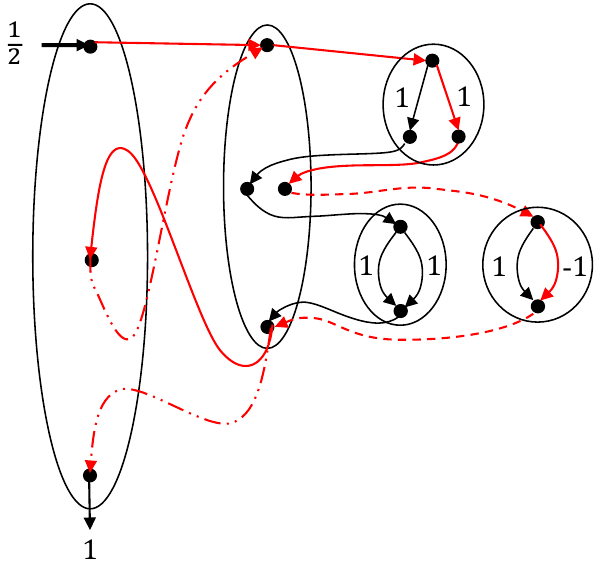}
    \caption{WCFLOBDD for $H_4$}
    \label{Fi:wcflobdd_H4}
  \end{subfigure}
  \begin{subfigure}{0.42\linewidth}
    \centering
    \includegraphics[width=0.98\linewidth]{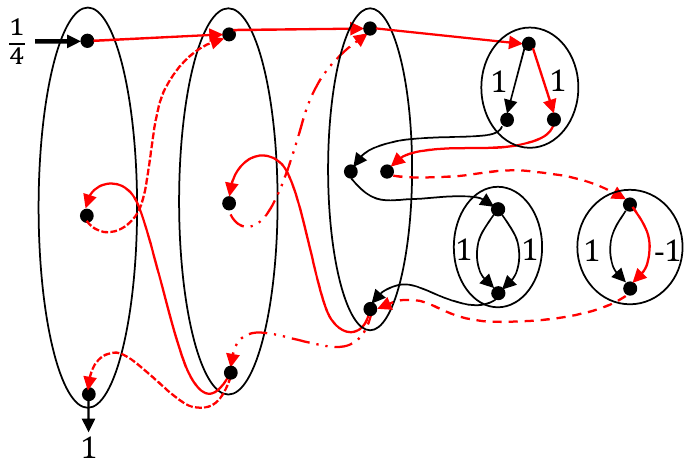}
    \caption{WCFLOBDD for $H_8$}
    \label{Fi:wcflobdd_H8}
  \end{subfigure}
  \begin{subfigure}{0.23\linewidth}
    \centering
    \includegraphics[width=.98\linewidth]{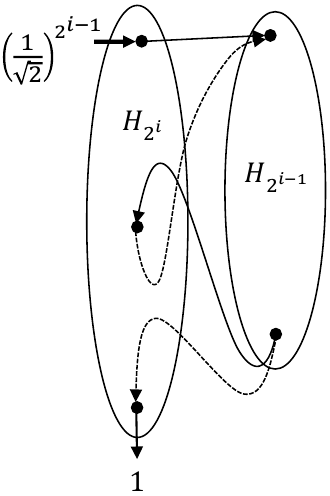}
    \caption{General structure of WCFLOBDD for $H_{2^i}$}
    \label{Fi:wcflobdd_Hn}
  \end{subfigure}
  \begin{subfigure}{0.72\linewidth}
    \centering
    \includegraphics[width=0.98\linewidth]{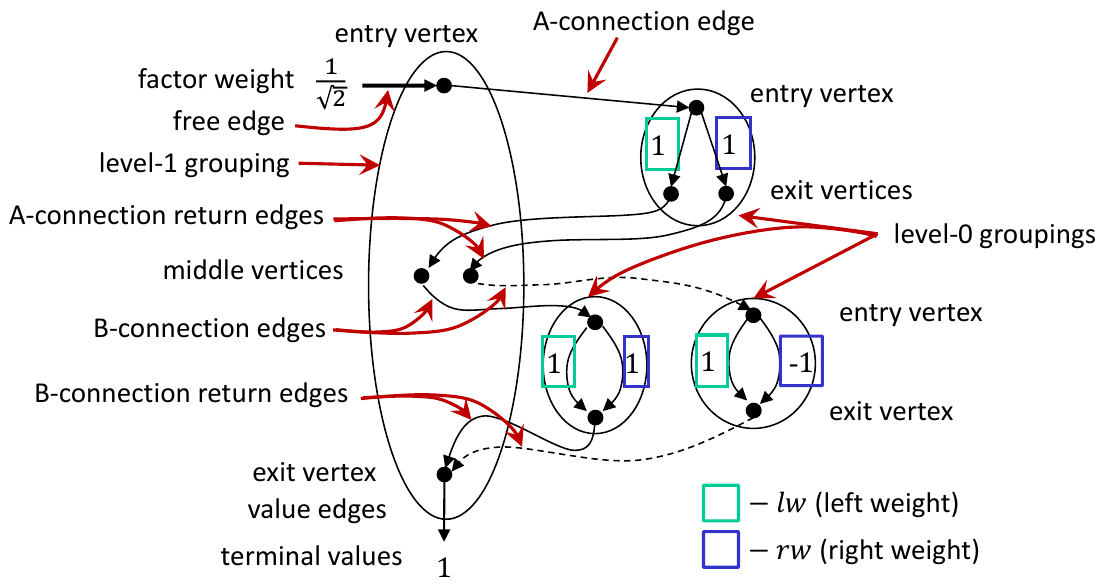}
    \caption{Terminology for WCFLOBDD components}
    \label{Fi:terminology}
  \end{subfigure}
  \caption{
    (a), (b), and (c) show WCFLOBDDs for the first three matrices in $\mathcal{H}$---$H_2$, $H_4$, and $H_8$---with the variable ordering $\langle x_0, y_0, x_1, y_1, \ldots \rangle$ ($\vec{x}$: row; $\vec{y}$: column).
    (d) shows the general structure of a WCFLOBDD that represents $H_{2^i}$. 
    (e) illustrates the constituents of the WCFLOBDD for $H_2$.
}
  \label{Fi:wcflobdd_hadamard}
\end{figure}

\subsection{Basic Structure of WCFLOBDDs}
\label{Se:BasicStructureOfWCFLOBDDs}

\twrchanged{
The hierarchy in a WCFLOBDD is similar to that in a CFLOBDD.
The number of variables at each level is always a power of 2, and at each level $k$, the variables are divided in half---the $A$ call interpreting $x_0, \dots, x_{2^{k-1}-1}$, and the $B$ calls interpreting $x_{2^{k-1}},\dots, x_{{2^k}-1}$.
The main difference from a CFLOBDD is that at level $0$, each transition has a \emph{weight}.
This approach produces a structure that is akin to a non-recursive, single-entry, multi-exit, \emph{weighted} hierarchical finite-state machine.
}

\twrchanged{
We define WCFLOBDDs in two parts.
\defref{wcflobdd-def} defines the basic structure of WCFLOBDDs, whose various elements are depicted in \figref{wcflobdd_hadamard}(e).
\defref{StructuralInvariants} imposes some additional structural invariants to ensure that WCFLOBDDs provide a canonical representation of Boolean functions.
Where necessary, we distinguish between \emph{mock-WCFLOBDDs} (\defref{wcflobdd-def}) and \emph{WCFLOBDDs} (\defref{StructuralInvariants}), although we typically drop the qualifier ``mock-'' when there is little danger of confusion.
}

\begin{definition}[\twrchanged{Mock-WCFLOBDD; see \figref{wcflobdd_hadamard}(e)}]
\label{De:wcflobdd-def}
\twrchanged{
A \emph{mock-WCFLOBDD} at level $k$ is a hierarchical structure made up of some number of \emph{groupings}, of which there is one grouping at level $k$, called the \emph{head-grouping}, and at least one grouping at each level $0, 1, \ldots, k-1$.
A grouping is a collection of vertices and edges (to other groupings), and weights (for level-0 groupings), with the structure described below.
}

\twrchanged{
A mock-WCFLOBDD is a triple $\langle f, G, V \rangle$, where
$f \in \mathcal{D}$ is the \emph{factor weight} of the \emph{free-edge} (defined below),
$G$ is the head-grouping
}
\twrchanged{
(conceptually ``topmost;''
portrayed leftmost),
}
and $V \in \{ [\bar{0}], [\bar{1}], [\bar{0}, \bar{1}], [\bar{1}, \bar{0}] \}$---where ``$[\cdot]$'' denotes a tuple---is the sequence of \emph{terminal values}, whose cardinality must match the number of exit vertices of $G$ (defined below).

\begin{itemize}
  \item \textit{Groupings:}
    Each oval-shaped object is called a \emph{grouping}, and every grouping is associated with a level $l$ ($\geq 0$).
    Each grouping at level $i$ is always connected to groupings at level $i\text{-}1$. 
  \item \textit{Vertices:}
    Each grouping $g$ at a level $\geq 1$ has
\twrchanged{
    a unique
}
    \textit{entry vertex} (the dot at the top of $g$), and some \textit{middle vertices} and \textit{exit vertices}
    (dots in the middle and at the bottom of $g$).
    The three sets of vertices are disjoint.
\twrchanged{
    (It is useful to think of the collections of middle vertices and exit vertices as two sequences, each numbered from left-to-right.)
}

    \hspace*{1.5ex}
\twrchanged{
    A level-$0$ grouping 
\twrchanged{
    (conceptually ``bottommost;''
    portrayed to the right)
}
    has a unique entry vertex, no middle vertices, and two edges (the $0$-edge and $1$-edge), which are directed to either one or two exit vertices.
    There are two kinds of level-$0$ groupings:
}
    \textit{fork groupings} (with two exit vertices) and \textit{don't-care groupings} (with one exit vertex).

    \item \textit{A-connection:} The grouping between $g$'s entry and middle vertices is called the A-connection.
    The edge from $g$'s entry vertex is called an A-connection edge, and the edges from the A-connection's exit vertices to $g$'s middle vertices are called \emph{A-connection return edges}.
    \item \textit{B-connections:} The groupings between $g$'s middle and exit vertices are called B-connections.
    The edges from $g$'s middle vertices are called B-connection edges and the edges from a B-connection's exit vertices to $g$'s exit vertices are called \emph{B-connection return edges}.

    \hspace{1.5ex}
    There is one A-connection, but there can be more than one B-connection (each with a set of B-connection return edges back to $g$'s exit vertices).
    \item \textit{Value Edges:}
    The edges that connect the exit vertices of the head-grouping to the terminal values $V$ are called \emph{value edges}.
\end{itemize}

\twrchanged{
The only edges that assigned weights are the free-edge and the edges of level-0 groupings.
}
\begin{itemize}
  \item
    The incoming edge to the
\twrchanged{
    entry vertex of the outermost grouping is called the \emph{free-edge}.
}
    The weight $f \in \mathcal{D}$ associated with the free-edge is the \emph{factor weight}.
  \item
    A level-0 grouping has a pair of semi-field weights $(lw, rw)$:
    $lw \in \mathcal{D}$ is the weight for the decision-edge for $0$ (false);
    $rw \in \mathcal{D}$ is the weight for the decision-edge for $1$ (true).
\end{itemize}

\twrchanged{
The terminal values connected to the exit vertices of the head-grouping can only be $\bar{0}, \bar{1} \in \mathcal{D}$.
}
For a given function $h$, the appropriate sequence of terminal values $V$ can be understood by considering the decision tree for $h$:
\begin{itemize}
  \item
    $[\bar{0}]$ only occurs for the function
    $\lambda \vec{x}. \bar{0}$.
    All leaves of the decision tree are $\bar{0}$.
  \item
    $[\bar{1}]$ occurs for a function for which the value is never $\bar{0}$ (for any assignment to the Boolean variables). No leaf of the decision tree is $\bar{0}$.
  \item
    $[\bar{1}, \bar{0}]$ occurs when a function has the value $\bar{0}$ for
    at least one
    assignment, but the leftmost leaf of the decision tree is non-$\bar{0}$.
  \item
    $[\bar{0}, \bar{1}]$ occurs when the function has a non-$\bar{0}$ value for
    at least one assignment, and
    the leftmost leaf of the decision tree is $\bar{0}$.
\end{itemize}
\end{definition}

\figref{wcflobdd_hadamard}(a) shows the WCFLOBDD for $H_2$. 
It consists of 4 groupings, of which 3 are at level 0 and one is at level 1; the upper-right level-$0$ grouping is a fork grouping, and the bottom-right level-$0$ groupings are don't-care groupings.

A grouping at level-$k$ encodes $2^k$ variables and $2^{2^k}$ paths.
As seen in~\figref{wcflobdd_hadamard}(a), the grouping at level-$1$ encodes two variables, and therefore 4 paths. 
In \figref{wcflobdd_hadamard}(b), the largest oval is a level-$2$ grouping, which encodes 4 variables: 
2 variables through the A-connection grouping and 2 variables through the B-connection grouping, and therefore $2^4$ paths.
In general, at level $k$ the $2^k$ variables are divided into halves:
$2^{k-1}$ in the A-connection and $2^{k-1}$ in the B-connection.

\begin{figure}[tb!]
    \centering
    \includegraphics[scale=0.44]{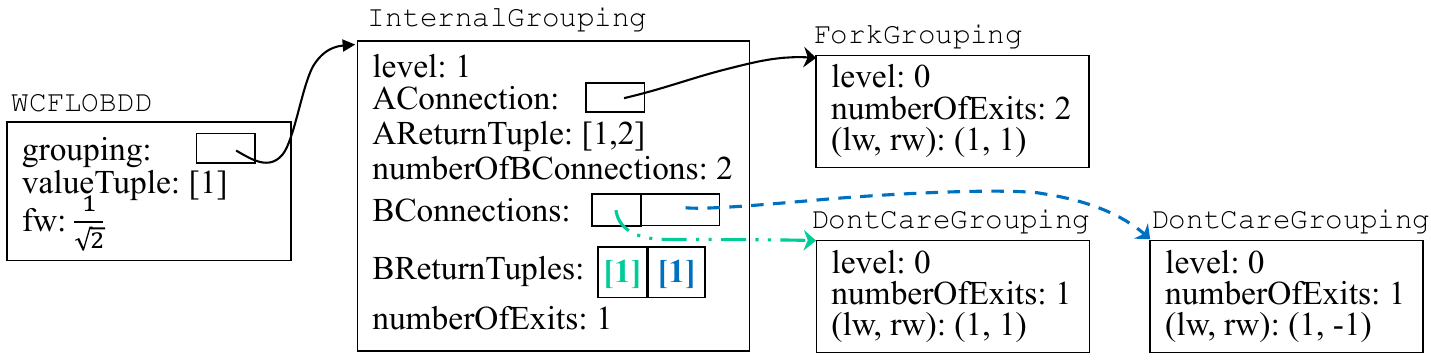}
    \caption{
    Object diagram of the WCFLOBDD for matrix $H_2$ (\figref{wcflobdd_hadamard}(a)).
    }
    \label{Fi:oops_wcflobdd}
    \vspace{-3ex}
\end{figure}


\twrchanged{
\figref{oops_wcflobdd} shows an object diagram of the WCFLOBDD for matrix $H_2$ from \figref{wcflobdd_hadamard}(a).
In this representation, which will be used in all algorithms stated in the paper, there are no explicit entry, middle, and exit vertices. The pointer to the object serves as the entry vertex; numbers in [1..\textit{numberOfBConnections}] stand for the middle vertices;
and those in [1..\textit{numberOfExits}] stand for the exit vertices.
(The context determines whether an index is being used as a middle vertex or exit vertex.)
As can be seen by comparing \figref{oops_wcflobdd} to \figref{wcflobdd_hadamard}(a), the various edges that appear in \figref{wcflobdd_hadamard}(a) are stored as pointers to Groupings and elements of ReturnTuples.
For a given A-connection edge or B-connection edge $c$ from grouping $g_i$ to $g_{i-1}$, the associated ReturnTuple $rt_c$ consists of the sequence of targets of return edges from exit vertices of $g_{i-1}$ to the middle or exit vertices of $g_i$ that correspond to $c$.
There are three grouping classes:
$\textit{InternalGrouping}$, $\textit{ForkGrouping}$, and $\textit{DontCareGrouping}$.
The latter two are level-$0$ groupings, which store left and right edge weights $(lw, rw)$;
thus, operations that construct a level-$0$ grouping take two inputs, $lw$ and $rw$.
All groupings at levels $\geq 1$ are $\textit{InternalGrouping}$s.
}

\twrchanged{
To be able to make inductive arguments about WCFLOBDDs, we define
}

\twrchanged{
\begin{definition}[Mock-proto-WCFLOBDD]\label{De:MockProtoWCFLOBDD}
  A \emph{mock-proto-WCFLOBDD} at level $i$ is a grouping at level $i$, together with the lower-level groupings to which it is connected (and the connecting edges).
  In other words, a mock-proto-WCFLOBDD has the following recursive structure:
  \begin{itemize}
    \item
      a mock-proto-WCFLOBDD at level 0 is either a fork grouping or a don't-care grouping
    \item 
      a mock-proto-WCFLOBDD at level $i$ is headed by a grouping at level $i$ whose
      \begin{itemize}
        \item
          A-connection edge and associated return edges ``call'' a level-($i$-1) mock-proto-WCFLOBDD
        \item
          B-connection edges and their associated return edges ``call'' some number of level-($i$-1) mock-proto-WCFLOBDDs.
      \end{itemize}
  \end{itemize}
  A mock-WCFLOBDD is a mock-proto-WCFLOBDD in which (i) the exit vertices of the mock-proto-CFLOBDD have been associated with specific terminal values $V$, and (ii) a factor weight has been supplied.
  (In other words, in the tuple $\langle f, G, V \rangle$ in \defref{wcflobdd-def}, $G$ serves double-duty: both as the head-grouping, and as the mock-proto-WCFLOBDD of which it is the head-grouping.)
\end{definition}
}

When interpreting a Boolean assignment, the path taken must abide by the following principle:
\begin{mdframed}[innerleftmargin = 3pt, innerrightmargin = 3pt, skipbelow=-0.0em]
  $\textbf{\MPP}$.  {\em When a path follows an edge that returns to level $i$ from level $i-1$, it must follow an edge that matches the closest preceding edge from level $i$ to level $i-1$.\/}
\end{mdframed}

\twrchanged{
Formally, the matched-path principle can be expressed as a condition that---for a path to be \emph{matched}---the word spelled out by the labels on the edges of the path must be a word in a certain context-free language \cite{kn:Yann90}.
This idea is the origin of ``CFL'' in ``CFLOBDD'' and ``WCFLOBDD.''
One way to formalize the condition is to label each edge from level $i$ to level $i-1$ with an open-parenthesis symbol of the form ``$(_b$'', where $b$ is an index that distinguishes the edge from all other edges to any entry vertex of any grouping of the WCFLOBDD.
(In particular, suppose that there are  $\NumConnections$ such edges, and that the value of $b$ runs from $1$ to $\NumConnections$.)
Each return edge that runs from an exit vertex of the level-$(i\text{-}1)$ grouping back to level $i$, and corresponds to the edge labeled ``$(_b$'', is labeled ``$)_b$''.
Each path in a WCFLOBDD then generates a string of parenthesis symbols formed by concatenating, in order, the labels of the edges on the path.
(For purposes of this discussion, edges in level-$0$ groupings are treated as $\epsilon$.)
A path 
is called a \emph{$\Matched$-path} iff the path's word is in the language $L(\Matched)$ of balanced-parenthesis strings generated by
\begin{equation}
  \label{Eq:MatchedPathGrammar}
  \Matched
  \rightarrow
  \epsilon \mid \Matched~\Matched \mid {(_b}~\Matched~{)_b} \qquad 1 \leq b \leq \NumConnections 
\end{equation}
Only $\Matched$-paths that start at the entry vertex of the WCFLOBDD's highest-level grouping and end at a terminal value are considered in interpreting a WCFLOBDD.
}

Decisions are taken only at the level-$0$ groupings:
the edges in a level-$0$ grouping correspond to the two choices for a Boolean variable. 
A grouping can be reused multiple times at different
\twrchanged{
``calls,'' as seen in \figref{wcflobdd_hadamard}(b)--(d).
}
Because of sharing at all levels, a given level-$0$ grouping (decision grouping) does not represent a specific variable,
according to the following principle:

\begin{mdframed}[innerleftmargin = 3pt, innerrightmargin = 3pt, skipbelow=-0.0em]
  $\textbf{\CIP}$.
  {\em A level-$0$ grouping is not associated with a specific Boolean variable.
  Instead, the variable that a level-$0$ grouping refers to is determined by the context:
  the
\twrchanged{
  $j^{\textit{th}}$
}
  level-$0$ grouping visited along a matched path is used to interpret the
\twrchanged{
  $j^{\textit{th}}$
}
  Boolean variable.}
\end{mdframed}

\emph{Function evaluation.}
\twrchanged{
As with BDDs, it is often useful to think of an assignment $a$ as a sequence of Boolean values.
We use ``$||$'' to denote concatenation of two half-size sequences:
$a = a_1 || a_2$, where $|a_1|$ $= |a_2|$ $= |a|/2$ (and $|a|$ is a power of 2).
}

\begin{definition}[Exit vertex reached via assignment $a$]
  \label{De:ExitVertexReached}
\twrchanged{
  Suppose that $g$ is a grouping.
  Let $g.\BReturnTuples[j]$ be the sequence of return edges for the $j^{\textit{th}}$ B-connection of $g$;
  i.e., $g.\BReturnTuples[j][k]$ represents the exit vertex of $g$ that is connected to the $k^{\textit{th}}$ exit vertex of the $j^{\textit{th}}$ B-connection of $g$.
  For a grouping $g$ and an assignment $a$, $\langle g \rangle(a)$ denotes the exit vertex of $g$ reached via assignment $a$, defined as follows:
  \[
    \begin{array}{@{\hspace{0ex}}r@{\hspace{0.75ex}}c@{\hspace{0.75ex}}l@{\hspace{18ex}}r@{\hspace{0.75ex}}c@{\hspace{0.75ex}}l@{\hspace{0ex}}}
      \langle \DontCareGrouping \rangle([0]) & \eqdef & 1
      &
      \langle \DontCareGrouping \rangle([1]) & \eqdef & 1 \\
      \langle \ForkGrouping \rangle([0]) & \eqdef & 1
      &
      \langle \ForkGrouping \rangle([1]) &\eqdef & 2 \\
      \langle g \rangle(a) & \eqdef & \multicolumn{4}{@{\hspace{0ex}}l}{g.\BReturnTuples[j][k], \textrm{where}~g.\textit{level} \geq 1, a = a_1 || a_2,} \\
           & & \multicolumn{4}{@{\hspace{0ex}}l}{j = \langle g.\AConnection\rangle(a_1),~\textrm{and}~k = \langle g.\BConnections[j]\rangle(a_2).}
    \end{array}
  \]
}
\end{definition}

\begin{definition}[WCFLOBDD evaluation]
  \label{De:WCFLOBDDEvaluation}
\twrchanged{
  Given WCFLOBDD $C = \langle f, G, V \rangle$ for $h: \{0,1\}^n \rightarrow \mathcal{D}$, and an assignment $a$ for $h$'s arguments, the value $h(a)$ is obtained from $C$, denoted by $\sem{C}(a)$, as the product of the edge weights on the path for $a$ in $C$:
  \[
    \begin{array}{@{\hspace{0ex}}r@{\hspace{0.75ex}}c@{\hspace{0.75ex}}l@{\hspace{0ex}}}
      \sem{C}(a) & \eqdef & f \cdot \sem{G}(a) \cdot V[\langle G \rangle(a)] \\
      \sem{G}(a) & \eqdef & \sem{G.\AConnection}(a_1) \cdot \sem{G.\BConnection[j]}(a_2),
    \end{array}
  \]
  where $a = a_1 || a_2$ and $j = \langle G.\AConnection \rangle (a_1)$.
}
\end{definition}

In the case of~\figref{wcflobdd_hadamard}(a),
$\mathcal{D} = \langle \mathbb{R}, +, \times, 0, 1 \rangle$.
For $a = \{ x_0 \mapsto 1, y_0 \mapsto 1 \}$, the path highlighted in bold evaluates to $H[1][1] = \frac{1}{\sqrt{2}} \times 1 \times -1 \times 1 = \frac{-1}{\sqrt{2}}$,
as desired.
\figref{wcflobdd_hadamard}(b) and \figref{wcflobdd_hadamard}(c) show the WCFLOBDD representations of $H_4$ and $H_8$.
\figref{wcflobdd_hadamard}(d) shows the generic structure for
how $H_{2^i}$ is
formed from $H_{2^{i-1}}$.
As one can see, the grouping at level $i\text{-}1$ is shared twice at 
\twrchanged{
each
}
level $i \ge 2$, giving a representation that is exponentially more succinct than a WBDD.

\textit{Intuition for succinctness:} Consider the WBDD and WCFLOBDD for $H_4$ in~\figref{wbdds-hadamard}(b) and~\figref{wcflobdd_hadamard}(b), respectively.
The A-connection of the level-$2$ grouping in \figref{wcflobdd_hadamard}(b) represents the top half of the WBDD structure (for variables $\langle x_0, y_0 \rangle$), and the B-connection of the level-$2$ grouping represents the bottom half of the WBDD structure (for variables $\langle x_1, y_1 \rangle$).
However, because the A-connection grouping and B-connection grouping of the WCFLOBDD have the identical structure, they are shared, producing a more succinct representation.

\textit{Implementation.}
Our WCFLOBDD implementation will be made available in open-source form if the paper is accepted.
The implementation is parameterized---via C++ templatization---to work with any user-defined semi-field $\mathcal{D}$ (and its corresponding operations $+$ and $\times$).

\subsection{Canonicity}
\label{Se:canonicity}

Like other decision diagrams, we ensure \emph{canonicity}---every function has a unique representation---by imposing certain structural invariants \cite[p.\ 48]{DBLP:books/siam/Wegener00}, \cite[\S4]{TOPLAS:SCR24}.

\textit{Intuition for canonicity:}
The structural invariants are designed to ensure that---for a given order on the Boolean variables---each Boolean function has a unique, canonical representation as a WCFLOBDD.
In reading \defref{StructuralInvariants} below, it will help to keep in mind that the goal of the invariants is to force there to be a \emph{unique} way to fold a given decision tree into a WCFLOBDD that represents the same Boolean function.
The decision-tree folding method is discussed later in this section and in
\sectref{canonicalness},
and is illustrated in \figref{Complicated}, but the main characteristic of the folding method is that it works \emph{greedily, left to right}.
This directional bias shows up in structural invariants~\ref{Inv:1}, \ref{Inv:2a}, and \ref{Inv:2b}.

\begin{definition}\label{De:StructuralInvariants}
\twrchanged{
    A \emph{(proto-)WCFLOBDD} is a mock-(proto-)WCFLOBDD in which every grouping/proto-WCFLOBDD  satisfies the \emph{structural invariants} given below.
}
    (See \cite[Fig.\ 8]{TOPLAS:SCR24}.)   
\end{definition}

\paragraph{\textbf{Structural Invariants.}}
The structural invariants mainly deal with (i) the organization of return tuples, and (ii) weights in level-0 groupings.
The following invariants deal with constraints on return tuples.
Let $c$ be an A-connection or B-connection edge with return tuple $rt_c$ from $g_{i-1}$ to $g_{i}$.
\begin{enumerate}
  \item
    \label{Inv:1}
     If $c$ is an $A$-connection edge, then $rt_c$ must map $g_{i-1}$'s exit vertices 1-to-1, in order, onto $g_i$'s middle vertices.
  \item
    If $c$ is the $B$-connection edge whose source is middle vertex $j+1$ of $g_i$ and whose target is $g_{i-1}$, then $rt_c$ must meet two conditions:
        \begin{enumerate}
          \item
            \label{Inv:2a}
            It must map $g_{i-1}$'s exit vertices 1-to-1 (but not necessarily onto) $g_i$'s exit vertices.
          \item
            \label{Inv:2b}
            It must ``compactly extend'' the set of exit vertices in $g_i$ defined by the return tuples for the previous $j$ $B$-connections (similar to condition (\ref{Inv:1}) on A-connections):
            Let $rt_{c_1}$, $rt_{c_2}$, $\ldots$, $rt_{c_j}$ be the return tuples for the first $j$ $B$-connection edges out of $g_i$.
            Let $S$ be the set of indices of exit vertices of $g_i$ that occur in return tuples $rt_{c_1}$, $rt_{c_2}$, $\ldots$, $rt_{c_j}$, and let $n$ be the largest value in $S$.
            (That is, $n$ is the index of the rightmost exit vertex of $g_i$ that is a target of any of the return tuples $rt_{c_1}$, $rt_{c_2}$, $\ldots$, $rt_{c_j}$.)
            If $S$ is empty, then let $n$ be $0$.

            \hspace*{1.5ex}
            Now consider $rt_c$ ($= rt_{c_{j+1}}$).
            Let $R$ be the (not necessarily contiguous) sub-sequence of $rt_c$ whose values are strictly greater than $n$.
            Let $m$ be the size of $R$.
            Then $R$ must be exactly the sequence $[n+1, n+2, \ldots, n+m]$.
        \end{enumerate}
  \item
    A \protoWCFLOBDD can be used (i.e., pointed to by other higher-level groupings) multiple times, but a \protoWCFLOBDD never contains two separate instances of equal {\protoWCFLOBDD}s.
      (Equality is defined inductively on the hierarchical structure of groupings.)
    
  \item
    \label{Inv:Structural:End}
    For every pair of $B$-connection edges $c$ and $c'$ of grouping $g_i$, with associated return tuples $rt_c$ and $rt_{c'}$, if $c$ and $c'$ lead to two level $i\text{-}1$ \protoWCFLOBDD{s}, say $p_{i-1}$ and $p'_{i-1}$, such that $p_{i-1} = p'_{i-1}$, then the associated return tuples must be different (i.e., $rt_c \neq rt_{c'}$).
    This condition means that two $B$-connection edges $c$ and $c'$ can ``call'' the same \protoWCFLOBDD, but the two sets of return edges $rt_c$ and $rt_{c'}$ have to be different.
\end{enumerate}

\noindent
The following invariants pertain to the weights on edges in level-0 groupings:
\begin{enumerate}[resume]
  \item
   \label{Inv:Weights:Start}
    The left-edge and right-edge weights $(lw, rw)$ must obey:
    \begin{equation*}
        (lw, rw) = \begin{cases}
        (\bar{1}, rw) & \text{when $lw \neq \bar{0}$}\\
        (\bar{0}, \bar{1}) &\text{ otherwise }
        \end{cases}
    \end{equation*}
  \item
    If a middle vertex $m$ of $g_i$ has a path with aggregated weight $\bar{0}$ from the entry vertex of $g_i$ to $m$, then there must exist an exit-vertex $e$ such that every level-0 grouping edge on every path from $m$ to $e$ has weight $\bar{0}$.
  \item
    If $e$ is an exit vertex of the outermost grouping such that all paths of weight $\bar{0}$ lead to $e$, then $e$ must map to the terminal value $\bar{0}$.
  \item
    \label{Inv:Weights:End}
    If all paths lead to the terminal value $\bar{0}$,
    then the factor-weight must be $\bar{0}$.
\end{enumerate}

\noindent
Finally, in a level-$k$ \emph{WCFLOBDD} $\langle f, G, V \rangle$,
\begin{enumerate}[resume]
  \item
    The head-grouping $G$ of $\langle f, G, V \rangle$
    heads a level-$k$ \protoWCFLOBDD.
  \item
    \label{It:OutermostLevel:ValueTuple}
    Value-tuple $V$ is one of $\{ [\bar{0}], [\bar{1}], [\bar{0}, \bar{1}], [\bar{1}, \bar{0}] \}$, and thus maps each exit vertex of $G$ to a distinct terminal value.
    If there exists an exit vertex $e$ of $G$ such that all paths leading to $e$ have weight $\bar{0}$, then $e$ is mapped to $\bar{0}$.
\end{enumerate}

\paragraph{\textbf{Canonicity Theorem}}
\twrchanged{
Now we state the canonicity theorem for WCFLOBDDs:
}

\begin{theorem}\label{The:Canonicity}
If $C_1$ and $C_2$ are two level-$k$ WCFLOBDDs for the same Boolean function, and use the same variable ordering, then $C_1$ and $C_2$ are isomorphic.
\end{theorem}
\begin{ProofSketch}
    We prove the theorem by establishing three properties:
\begin{enumerate}
    \item Every level-$k$ WCFLOBDD represents a decision tree (for the same Boolean function with the same variable ordering) with $2^{2^k}$ leaves.
    \item Every decision tree with $2^{2^k}$ leaves is represented by some level-$k$ WCFLOBDD.
    \item No decision tree with $2^{2^k}$ leaves is represented by more than one level-$k$ WCFLOBDD (up to isomorphism).
\end{enumerate}


\emph{Obligation 1.}
\twrchanged{
This property is established
}
by induction. Let $P(l)$ 
\twrchanged{
be
}
the number of leaves of a decision tree for the function represented by a level-$l$ \protoWCFLOBDD.
We want to prove that $P(l) = 2^{2^l}$.
\begin{itemize}
    \item Base Case: $l = 0$, \#variables = 1, \#paths = \#leaves in a decision tree = 2.
    \item Induction Step: Let $P(m) = 2^{2^m}$.
    We wish to prove that $P(m+1) = 2^{2^{m+1}}$.
    Given a level-$l$ \protoWCFLOBDD with outermost grouping $g_l$, by construction the number of variables of $g_l$ is equally divided between the A-connection and B-connections of $g_l$. 
    Hence, for a level-$m{+}1$ \protoWCFLOBDD, we have
    $
      P(m+1) = \sum_{j}A_j(m) \cdot P(m) = (\sum_j A_j(m)) \cdot P(m) =  P(m) \cdot P(m) = 2^{2^m} \cdot 2^{2^m}  = 2^{2^{m+1}}
    $,
    where $A_j(m)$ is the number of matched paths through the level-$m$ A-connection \protoWCFLOBDD to its $j^{th}$ exit vertex. However, we are just summing over all A-connection exit vertices, so $\sum_j A_j(m) = P(m)$.
\end{itemize}

\emph{Obligation 2.}
This property is established by giving a (deterministic) construction to convert a decision tree with $2^{2^k}$ leaves into a level-$k$ WCFLOBDD (satisfying all structural invariants). The construction is given in Appendix
~\sectref{canonicalness}.
(The construction is a recursive procedure that works greedily over the decision tree, from left to right.)

\emph{Obligation 3.}
This property is established by (i) unfolding WCFLOBDD $C$ into a decision tree;
(ii) folding the decision tree back to a WCFLOBDD $C'$; and
(iii) showing that $C'$ is isomorphic to $C$.
(The folding-back step uses the same deterministic construction from Obligation 2.)
\end{ProofSketch}

\noindent
The folding/unfolding constructions of Obligations 2 and 3 are used solely
for the purpose of proving \theoref{Canonicity}; 
they are not explicit operations on WCFLOBDDs.

\textit{Intuition.}
Structural invariants (\ref{Inv:1})-(\ref{Inv:Structural:End}), which deal with return tuples,
ensure that the WCFLOBDD for a Boolean function $f$ has a \textit{``left-to-right''} folding of the decision tree for $f$ (with the same variable ordering).
Structural invariants (\ref{Inv:Weights:Start})-(\ref{Inv:Weights:End}), which deal with weights,
ensure that there cannot be two {\protoWCFLOBDD}s for the same decision-tree structure that have different sets of level-$0$ groupings, i.e., with different $(lw, rw)$ weights.
The deterministic construction algorithm that folds a decision tree into a WCFLOBDD 
involves a two-step process. The decision tree is first converted into a weighted decision tree (a decision tree with weights on the edges), and then the weighted decision tree is folded into a WCFLOBDD. The first step enforces the structural invariants of the weights and the second step of \textit{``left-to-right''} folding handles the structural invariants that deal with the return tuples.

\begin{figure}
    \centering
    \begin{subfigure}[t]{0.48\linewidth}
    \includegraphics[width=\linewidth]{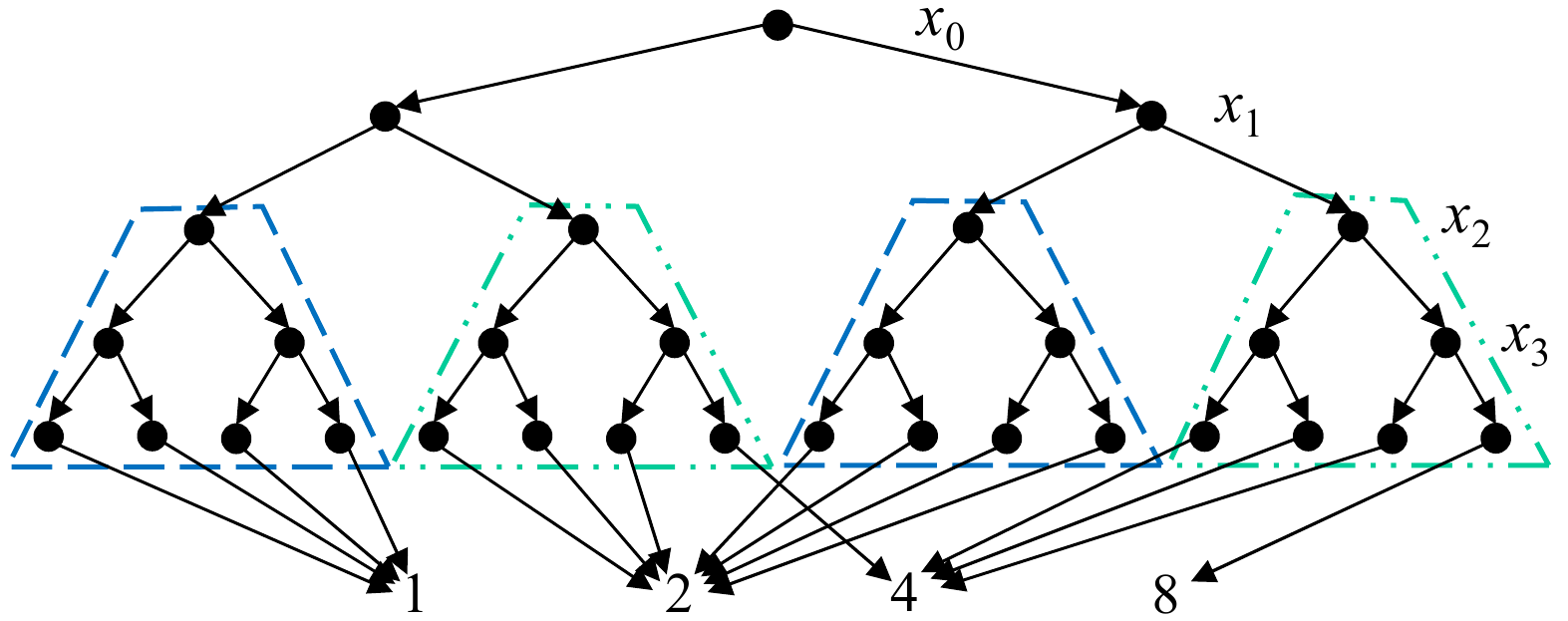}
    \caption{Decision tree}
    \vspace{4ex}
    \end{subfigure}
    \begin{subfigure}[t]{0.48\linewidth}
    \includegraphics[width=\linewidth]{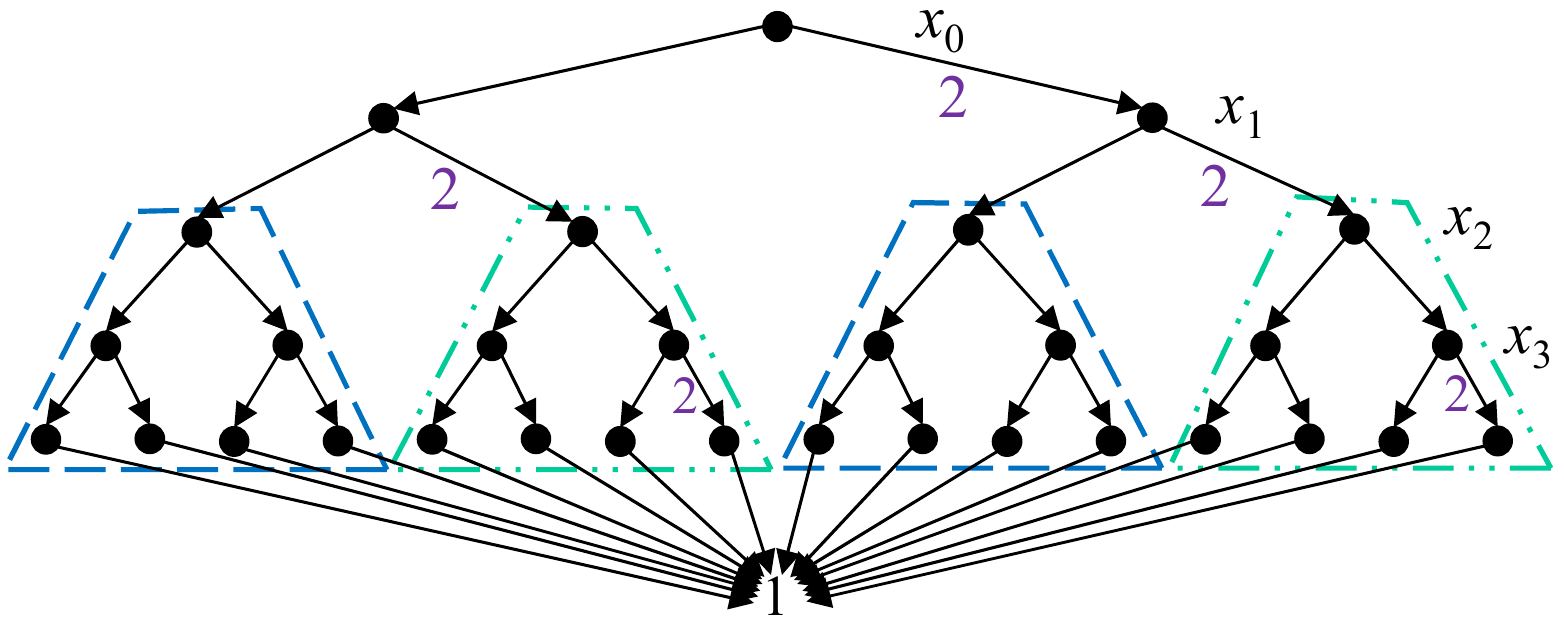}
    \caption{
      Weighted decision tree for (a).
      The edges labelled with numbers in purple are the edge weights.
\twrchanged{
      To reduce clutter,
}
      labels for edges with weight $1$ are omitted.
    }
    \vspace{4ex}
    \end{subfigure}
    \begin{subfigure}[t]{0.45\linewidth}
    \centering
    \includegraphics[width=\linewidth]{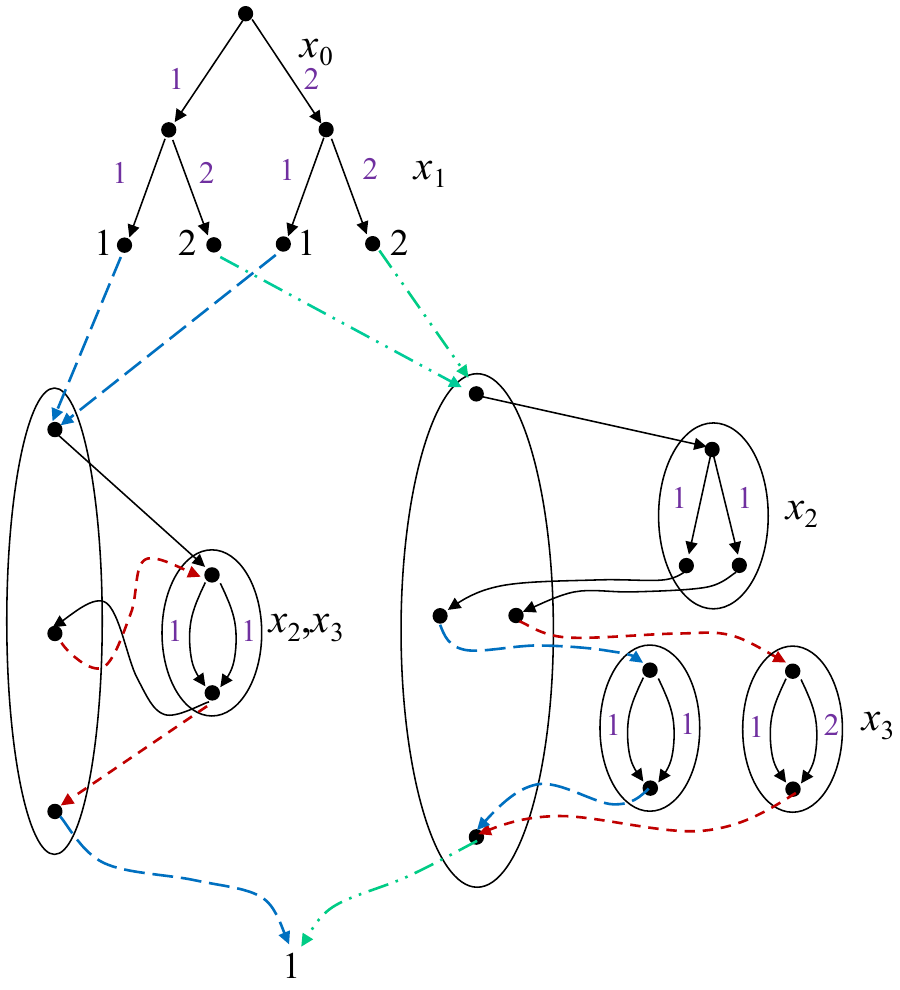}
    \caption{
    \protect \raggedright
      Hybrid of weighted decision tree for $x_0$ and $x_1$, and WCFLOBDDs for $x_2$ and $x_3$.
      The dashed, and dashed-double-dotted edges from the four vertices labeled 1, 2, 1, and 2, respectively, correspond to the dashed, and dashed-double-dotted trapezoids in (a) and (b). The numbers highlighted in purple indicate the edge-weights.
    }
    \end{subfigure}
    \begin{subfigure}[t]{0.5\linewidth}
    \centering
    \includegraphics[width=\linewidth]{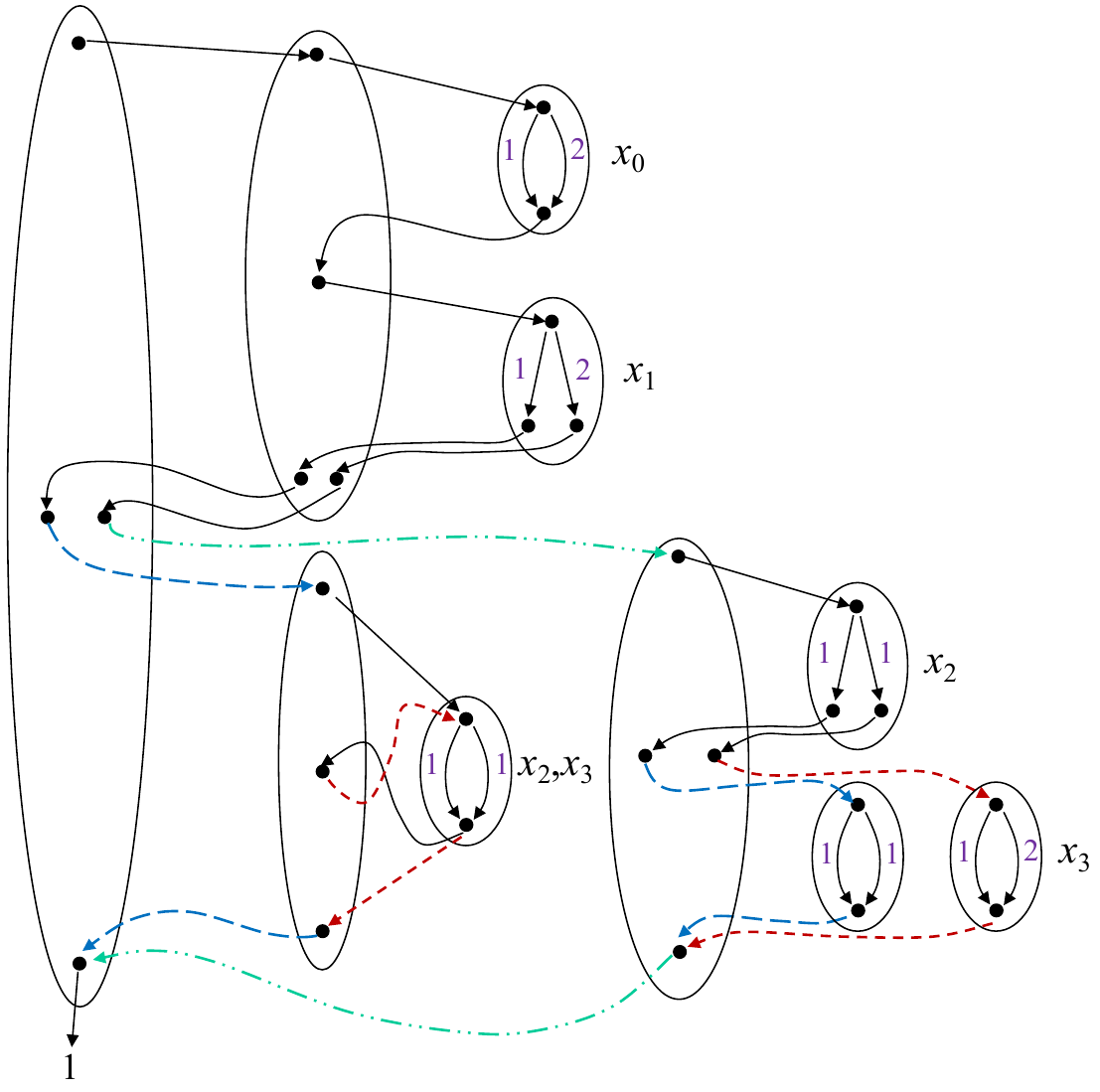}
    \caption{WCFLOBDD representation for the weighted decision tree in (a) (Some of the groupings have been duplicated to avoid clutter). The numbers highlighted in purple indicate the edge-weights.}
    \end{subfigure}
    \caption{Representations of the vector $V$ = $[1,1,1,1,2,2,2,4,2,2,2,2,4,4,4,8]^T$ with variables $\langle x_0, x_1, x_2, x_3 \rangle$ representing
\twrchanged{
    an index into $V$.
}
    }
    \label{Fi:Complicated}
\end{figure}

\textit{Example:} \figref{Complicated} shows the representation of a vector
of length 16, and thus 4 variables are used to represent the index of an element
of the vector. 
The four diagrams in \figref{Complicated} show some of the stages of
converting a decision tree to a WCFLOBDD.
\figref{Complicated}(a) shows the decision-tree representation of the vector.
The weighted decision tree for \figref{Complicated}(a) is shown in \figref{Complicated}(b).
In this step, the weights are assigned to the edges, ensuring that the structural invariants on the weights hold.
\figref{Complicated}(c) shows the process of converting a weighted decision tree to a WCFLOBDD.
The vertices labeled 1,2,1,2 at ``half-height'' (at the end of the weighted decision tree for $x_0$ and $x_1$) indicate the common sub-structures found during the left-to-right folding of the bottom half of the weighted decision tree.
These common ``values'' are used to drive the left-to-right folding of the upper-half weighted decision tree, and to ensure that the structural invariants hold on the return tuples of the A-connection of the outermost grouping in \figref{Complicated}(d).
Finally, \figref{Complicated}(d) shows the full WCFLOBDD representation for \figref{Complicated}(a).

\paragraph{\textbf{Comparison with CFLOBDDs.}}
WCFLOBDDs draw inspiration from CFLOBDDs:
the structural elements of levels and A-connection/B-connection edges, along with structural invariants that enforce a left-to-right folding of the decision tree are similar to CFLOBDDs.
However, in WCFLOBDDs the edges in level-$0$ groupings have weights, which differs from CFLOBDDs.
To ensure that WCFLOBDDs are canonical, we introduced structural invariants (\ref{Inv:Weights:Start})-(\ref{Inv:Weights:End}) and (\ref{It:OutermostLevel:ValueTuple}).

\subsection{Exponential Succinctness Gaps}
\label{Se:separation}

\twrchanged{
In this section, we establish
exponential gaps between (i) WCFLOBDDs and WBDDs, and (ii) WCFLOBDDs and CFLOBDDs using an example function.
}

\begin{definition}
\label{De:PopRelation}
\twrchanged{
The $n$-bit population-count function $\POP_n{:}~\{0,1\}^{n} \rightarrow \{2^{j} \mid j \in \{0 \ldots 2^n\}\}$ on variables $\{ x_0\cdots x_{n-1} \}$ is
$
  \POP_n(X) \eqdef 2^{\textit{pop}(X)}
                = 2^{\sum_{i=0}^{n-1}{\delta(x_i)}},
$
where
$
\delta(a) = \begin{cases}
    1 & \text{a = 1}\\
    0 & \text{otherwise}
\end{cases}
$
and $pop(X)$ represents the number of $1$s in the binary representation of $X$.
}
\end{definition}

\begin{theorem}\label{The:PopSeparation}
\twrchanged{
\twrchanged{
For $n = 2^k$, where $k \geq 0$,
}
$\POP_n$,
defined over the variables $\{x_0\cdots x_{n-1}\}$, 
can be represented by a WCFLOBDD with
\twrchanged{
$\Theta(\log n)$
}
vertices and edges.
A CFLOBDD
that uses only the variables $\{x_0\cdots x_{n-1}\}$
for $\POP_n$ requires $\Omega(n)$ vertices and edges;
and a WBDD that uses the same variable set
\twrchanged{
requires
}
$\Theta(n)$ nodes.
}
\end{theorem}

\begin{wrapfigure}{R}{0.50\textwidth}
  \centering
  \vspace{-3.0ex}
  \begin{tabular}{c@{\hspace{3.0ex}}c}
    \begin{tabular}{@{\hspace{0ex}}c@{\hspace{0ex}}}
      \includegraphics[width=0.14\textwidth]{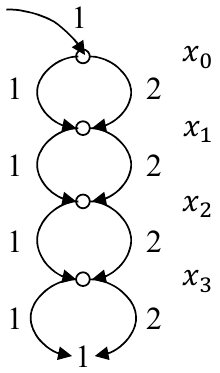}
    \end{tabular}
    &
    \begin{tabular}{@{\hspace{0ex}}c@{\hspace{0ex}}}
      \includegraphics[width=0.30\textwidth]{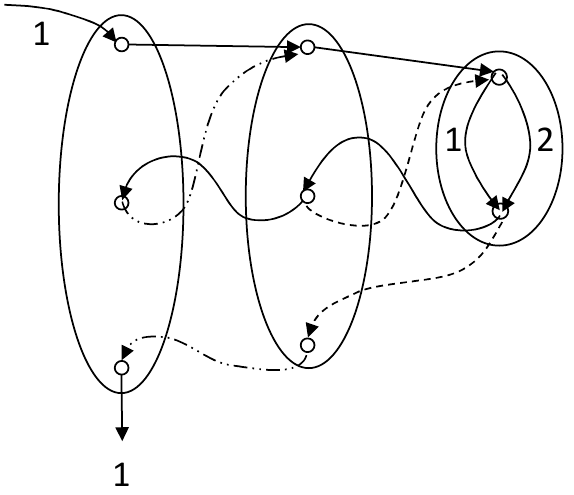}
    \end{tabular}
    \\
    {\small (a)} & {\small (b)} 
  \end{tabular}
  \caption{
\twrchanged{
  Two representations of $\POP_4$ using (a) WBDDs and (b) WCFLOBDDs.
}
  }
    \label{Fi:SeparationFigures}
  \vspace{-4ex}
\end{wrapfigure}


\begin{Proof}
\twrchanged{
\textit{Weighted Decision Tree.}
\twrchanged{
 Each $x_i$ node in
}
 a weighted decision tree for $\POP_n$ has a 1 on the left branch ($x_i$ = 0) and a 2 on the right branch ($x_i$ = 1). The paths of the weighted decision tree evaluate to $2^k + 1$ different values: the leftmost path (all variables set to 0)  evaluates to 1 = $2^0$; the rightmost path (all variables set to 1) evaluates to $2^{2^k}$.  There are repeated values that the function represented by the weighted decision tree evaluates to, but we have every power of 2 from $2^0$ to $2^{2^k}$.
 }

\twrchanged{
\textit{WBDD Claim.}
A WBDD for $\POP_n$ would 
}

\noindent
\twrchanged{
consist of $2^k$ don't-care nodes (one for every 
}

\noindent
\twrchanged{
variable)---with left-edge weight of $1$ and right-edge weight of $2$---stacked on top of each other.
Hence, the size of the WBDD is linear in the number of variables: $\Theta(2^k) = \Theta(n)$.
\twrchanged{
(See \figref{SeparationFigures}(i).)
}
}

\twrchanged{
\textit{WCFLOBDD Claim.}
\twrchanged{
Consider a WCFLOBDD for $\POP_n$ $W$ with $n = 2^k$ variables and $k$ levels.
There is only a single grouping in $W$ for each level.
At level $0$, the unique grouping is a don't-care grouping with $(lw, rw) = (1, 2)$.
At each level $l \geq 1$, the level-$l$ grouping's A-connection and (one) B-connection both ``call'' the unique grouping at level $l-1$.
Because every level of $W$ has only a single grouping,
}
the WCFLOBDD has $\Theta(k)$ = $\Theta(\log n)$ vertices and edges.
\twrchanged{
(See \figref{SeparationFigures}(ii).)
}
}

\twrchanged{
\textit{CFLOBDD Claim.}
The terminal values of the CFLOBDD $C$ representing $\POP_n$ would be the unique output values of the function. The terminal values of $C$ would contain all powers of 2 from $2^0$ to $2^{2^k}$; hence $C$ has $2^k + 1$ terminal values and the size of $C$ is $\Omega(2^k) = \Omega(n)$.
}

\twrchanged{
Moreover, weighted decision trees, WBDDs, CFLOBDDs, and WCFLOBDDs always have the structures described above \emph{regardless of the variable order one chooses to use}.
Hence, for $\POP_n$, WCFLOBDDs are inherently exponentially more succinct than both WBDDs and CFLOBDDs.
}
$~~\QED$
\end{Proof}

\section{Operations on WCFLOBDDs}
\label{Se:algos}

\begin{sidewaystable}
    \centering

    \resizebox{\textwidth}{!}{
    \begin{tabular}{|c|c|c|c|c|c|}
        \hline
         \multicolumn{1}{|c|}{\multirow{2}{*}{Operation}} & \multirow{2}{*}{Type Signature} & \multirow{2}{*}{Description} & \multicolumn{3}{|c|}{Time Complexity} \\
         \cline{4-6}
                                                          &                                 &                              &   WCFLOBDD & WBDD &
                                                          \twrchanged{CFLOBDD} \\
         \hline
         \multirow{2}{*}{Equal} & WCFLOBDD $\times$ WCFLOBDD  & \multirow{2}{*}{Checks if two WCFLOBDDs are equal} & \multirow{2}{*}{$\mathcal{O}(1)$} & \multirow{2}{*}{$\mathcal{O}(1)$} & \multirow{2}{*}{\twrchanged{$\mathcal{O}(1)$}} \\
         & $\rightarrow$ Boolean & & & & \\
         \hline
         \multirow{2}{*}{ConstantZero} & \multirow{2}{*}{Int ($k$) $\rightarrow$ WCFLOBDD}  & Creates a WCFLOBDD for & \multirow{2}{*}{$\mathcal{O}(k)$} & \multirow{2}{*}{$\mathcal{O}(1)$}  & \multirow{2}{*}{\twrchanged{$\mathcal{O}(k)$}} \\
          &  & the function $\lambda x_0 \dots x_{2^k-1}.F$ & & & \\
         \hline
         \multirow{2}{*}{ConstantOne} & \multirow{2}{*}{Int ($k$) $\rightarrow$ WCFLOBDD} & Creates a WCFLOBDD for & \multirow{2}{*}{$\mathcal{O}(k)$} & \multirow{2}{*}{$\mathcal{O}(1)$}  & \multirow{2}{*}{\twrchanged{$\mathcal{O}(k)$}} \\
          &  & the function $\lambda x_0 \dots x_{2^k-1}.T$ & & & \\
         \hline
         \multirow{2}{*}{Projection} & \multirow{2}{*}{Int($k$) $\times$ Int($i$) $\rightarrow$ WCFLOBDD} & Creates a WCFLOBDD for & \multirow{2}{*}{$\mathcal{O}(k)$} & \multirow{2}{*}{$\mathcal{O}(2^k)$}  & \multirow{2}{*}{\twrchanged{$\mathcal{O}(k)$}} \\
         &  & the function $\lambda x_0 \dots x_{2^k-1}. x_i$  & & & \\
         \hline
         \multirow{2}{*}{ScalarMultiply} & WCFLOBDD ($c$) $\times$ Value ($v$)  & \multirow{2}{*}{Performs $c' = c \ast v$} & \multirow{2}{*}{$\mathcal{O}(1)$} & \multirow{2}{*}{$\mathcal{O}(1)$}  & \multirow{2}{*}{\twrchanged{$\mathcal{O}(|c| \times |c'|)$}} \\
         & $\rightarrow$ WCFLOBDD & & & & \\
         \hline
         Pointwise Multiplication & WCFLOBDD ($c_1$) $\times$ WCFLOBDD ($c_2$) & \multirow{2}{*}{Performs \twrchanged{$c' = c_1 \odot c_2$}} & \multirow{2}{*}{$\mathcal{O}(|c_1| \times |c_2| \times |c'|)$} & \multirow{2}{*}{$\mathcal{O}(|c_1|_B \times |c_2|_B)$}  & \multirow{2}{*}{\twrchanged{$\mathcal{O}(|c_1| \times |c_2| \times |c'|)$}} \\
           \twrchanged{(Hadamard product $\odot$)} &  \twrchanged{$\rightarrow$ WCFLOBDD} &  &  & & \\
          \hline
         \multirow{2}{*}{Pointwise Addition} & WCFLOBDD ($c_1$) $\times$ WCFLOBDD ($c_2$) & \multirow{2}{*}{Performs \twrchanged{$c' = c_1 + c_2$}} & \multirow{2}{*}{$\mathcal{O}(2^{\textit{vars}})$} & \multirow{2}{*}{$\mathcal{O}(2^{\textit{vars}})$}  & \multirow{2}{*}{\twrchanged{$\mathcal{O}(|c_1| \times |c_2| \times |c'|)$}} \\
          &  \twrchanged{$\rightarrow$ WCFLOBDD} &  &  & & \\
         \hline
         \multirow{2}{*}{PathWeight} & \multirow{2}{*}{WCFLOBDD($c$) $\rightarrow$ WCFLOBDD} & Computes the weight of paths to  & \multirow{2}{*}{$\mathcal{O}(|c|)$}  & \multirow{2}{*}{$\mathcal{O}(|c|_B)$}  & \multirow{2}{*}{\twrchanged{$\mathcal{O}(|c|)$}} \\
          &  & every exit vertex of every grouping &  &  & \\
         \hline
        \multirow{2}{*}{Sampling} & \multirow{2}{*}{WCFLOBDD($c$) $\rightarrow$ String} & \multirow{2}{*}{Samples a path from $c$} & \multirow{2}{*}{$\mathcal{O}(\max (\textit{vars}, |c|))$} & \multirow{2}{*}{$\mathcal{O}(\max (\textit{vars}, |c|_B))$}  & \multirow{2}{*}{\twrchanged{$\mathcal{O}(\max (\textit{vars}, |c|))$}} \\
         &  &  &  &  & \\
         \hline
         \multirow{2}{*}{KroneckerProduct} & WCFLOBDD($c_1$) $\times$ WCFLOBDD($c_2$)  & \multirow{2}{*}{Performs $c' = c_1 \tensor c_2$} & \multirow{2}{*}{$\mathcal{O}(1)$}  & \multirow{2}{*}{$\mathcal{O}(|c_1|_B)$}
         & \multirow{2}{*}{\twrchanged{$\mathcal{O}\left(\left(
                            \begin{array}{@{\hspace{0ex}}c@{\hspace{0.75ex}}l@{\hspace{0ex}}}
                                & |c_1| + |c_2| \\
                              + & c_1.\text{\#exits} \times c_2.\text{\#exits}
                            \end{array}
                            \right) \times |c'|\right)$}} \\
         & $\rightarrow$ WCFLOBDD & & & & \\
         \hline
         \multirow{2}{*}{MatrixMultiply} & WCFLOBDD($c_1$) $\times$ WCFLOBDD($c_2$)  & Performs $c' = c_1 \times c_2$ & $\mathcal{O}(N^3)$, plus the time for & \multirow{2}{*}{$\mathcal{O}(N^3)$}  & \twrchanged{$\mathcal{O}(N^3)$, plus the time for} \\
         & $\rightarrow$ WCFLOBDD & for matrices of size $N \times N$ & a final call to \texttt{Reduce} & & \twrchanged{a final call to \texttt{Reduce}} \\
         \hline
    \end{tabular}
    }
    \caption{
      List of operations on 
\twrchanged{
      WCFLOBDDs, WBDDs and CFLOBDDs;
}
      $\textit{vars}$ denotes the number of Boolean variables ($= 2^k$, where $k$ is the number of levels of the WCFLOBDD).
      The size measure $|\cdot|$ counts the number of
      vertices and edges---with no double-counting of vertices and edges in groupings that are shared due to hash-consing.
      In the column for the time complexities of WBDD operations, an occurrence of $c$ refers to a WBDD argument of the operation, and $|c|_B$ denotes the size of WBDD $c$ (the number of nodes and edges).
      Note that the complexity of MatrixMultiply is in terms of the sizes of the matrices represented by $c_1$ and $c_2$,
\twrchanged{
      and not the sizes of the data structures $c_1$ and $c_2$, plus the cost of  an added ``final call to \texttt{Reduce}.''
      The latter cost depends on the sizes of the structures that are input to and output from \texttt{Reduce}.
      That is, the complexity of $c’ = \mathtt{Reduce}(c)$ is
      $\bigO(|c'| \times |c|)$.
      Because it is difficult to combine the complexity of the symbolic matrix-multiplication computation and that of \texttt{Reduce}, we broke out the cost of \texttt{Reduce} as a separate item.
    }
    }
    \label{Ta:algo-list-cflobdds}
\end{sidewaystable}

In this section, we discuss the algorithms for operations on WCFLOBDDs.
\tableref{algo-list-cflobdds} shows a table with the operations on WCFLOBDDs, along with their complexities, as well as the complexities of the corresponding WBDDs operations.

Because WCFLOBDDs are hierarchically structured, algorithms are
divide-and-conquer algorithms that, for a given \texttt{InternalGrouping}, recurse on the A-connection and then the B-connections (or vice versa), thereby splitting the variables in half for each subproblem, along the following lines:

\begin{center}
\smallskip
\noindent
{\small{{\tt
\begin{minipage}{\columnwidth}
\begin{tabbing}
Op\=(Grouping g) \{ \+ \\
    $\ldots$       // Base case: Construct appropriate level-0 grouping \\
    Internal\=Grouping g' = new InternalGrouping(k);  // k = g.level \\
    // Recursive calls on Op(g.AConnection) and Op(g.BConnections[i]), \\
    // for 1 $\leq$ i $\leq$ numberOfBConnections, to fill in g'.AConnection and the \\
    // elements of array g'.BConnections with level-(k-1) Groupings \\
    $\ldots$ \\
    return RepresentativeGrouping(g'); \- \\
\}
\end{tabbing}
\end{minipage}
}}}
\end{center}

\smallskip
\noindent
For each \texttt{InternalGrouping} constructed, the algorithms enforce the structural invariants (\sectref{canonicity}) on return tuples and weights (at level $0$) so that the return value is a (proto-)WCFLOBDD.
The algorithms use two standard techniques (mostly elided in the pseudo-code)
to ensure that the amount of memory used to represent a WCFLOBDD does not blow up, and to reduce the cost of WCFLOBDD operations.
\emph{Hash consing}~\cite{goto1974monocopy} is used to ensure that only one representative of a value exists in memory.
The operation ${\tt RepresentativeGrouping }$ checks whether a grouping is a duplicate, and if so, discards it and returns the grouping's representative.
${\tt RepresentativeGrouping }$ consults a global hash-table that maintains pointers to the unique representation of each grouping.
Consequently, one can test in unit time if two proto-WCFLOBDDs are equal by comparing their pointers.
(As will be seen in \lineseqref{PPConstantOneStart}{PPConstantOneEnd} of
\twrchanged{
\algref{PairProduct},
}
this ability allows some special cases to be identified quickly and helps speed up computations.)

In discussions of space costs later in this section, the ``size of a WCFLOBDD'' is the total number of vertices and edges---with no double-counting of vertices and edges in groupings that are shared due to hash-consing.

A \emph{cache} for a function $F$ is an associative-lookup table with pairs of the form $[x, F(x)]$, used to eliminate the cost of re-doing a previously performed computation \cite{michie1967memo}.
All the algorithms on WCFLOBDDs perform function caching using the following idiom:
\begin{center}
\small{{\tt
\begin{minipage}{\columnwidth}
\begin{tabbing}
F(x)\=\ \{ \+ \\
    $\textbf{if}~\texttt{cache}_F(x) \neq \texttt{NULL}~\textbf{return}~\texttt{cache}_F(x)$; \\
    $\ldots$ \\
    $\texttt{cache}_F(x) = \texttt{retVal}$;  // Update the cache with the return value\\
    return retVal; \- \\
\}
\end{tabbing}
\end{minipage}
}}
\end{center}
At the end of every algorithm, the cache is updated with the computed result, thereby allowing a redundant computation to be avoided if the algorithm is called later with the same arguments.

\begin{figure}[tb!]
    \centering
    \begin{subfigure}{0.45\linewidth}
    \centering
    \includegraphics[width=0.38\linewidth]{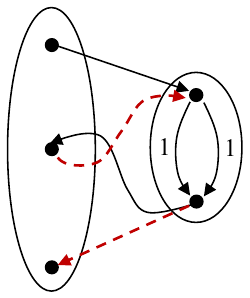}
    \caption{Level-$1$ ConstantOneProtoWCFLOBDD}
    \end{subfigure}
    \begin{subfigure}{0.45\linewidth}
    \centering
    \includegraphics[width=0.38\linewidth]{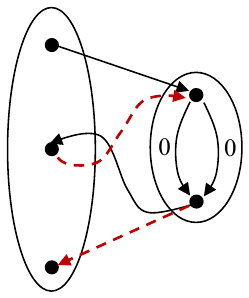}
    \caption{Level-$1$ ConstantZeroProtoWCFLOBDD}
    \end{subfigure}
\Omit{
    \begin{subfigure}{0.22\linewidth}
    \includegraphics[width=.8\linewidth]{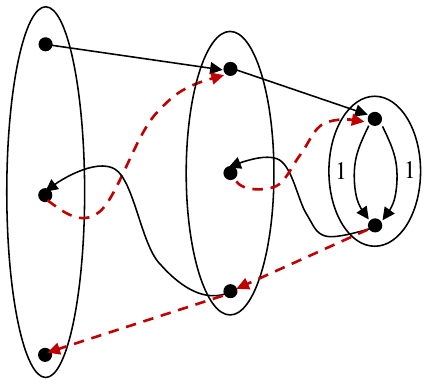}
    \caption{Level-$2$ ConstantOneProtoWCFLOBDD}
    \end{subfigure}
    \begin{subfigure}{0.45\linewidth}
    \includegraphics[width=.8\linewidth]{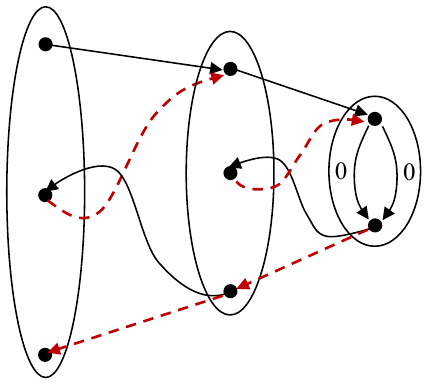}
    \caption{Level-$2$ ConstantZeroProtoWCFLOBDD}
    \end{subfigure}
}
    \begin{subfigure}{0.45\linewidth}
    \centering
    \includegraphics[width=0.76\linewidth]{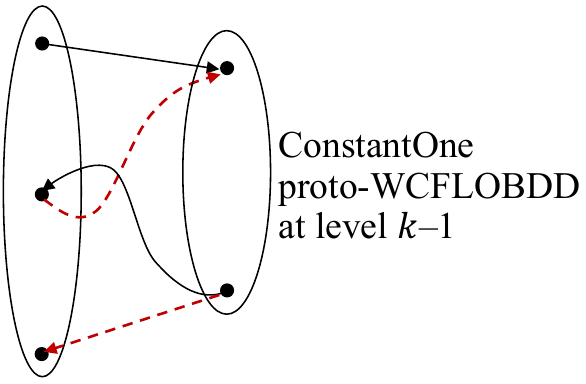}
    \caption{Level-$k$ ConstantOneProtoWCFLOBDD}
    \end{subfigure}
    \begin{subfigure}{0.45\linewidth}
    \centering
    \includegraphics[width=0.76\linewidth]{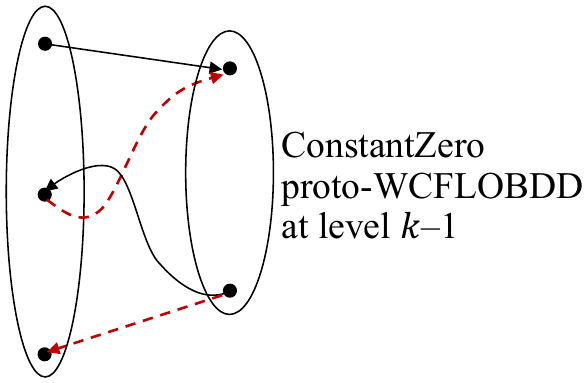}
    \caption{Level-$k$ ConstantZeroProtoWCFLOBDD}
    \end{subfigure}
\caption{${\tt ConstantOneProtoWCFLOBDD }$ and ${\tt ConstantZeroProtoWCFLOBDD}$ for levels $1$\Omit{, $2$,} and $k > 1$.}
\label{Fi:nodistinctfigs}
\end{figure}

\Omit{
\begin{algorithm}[tb!]
\caption{ConstantZeroProtoWCFLOBDD\label{Fi:ConstantZeroProtoCFLOBDDAlgorithm}}
\small{
\Input{int k -- level of the topmost grouping}
\Output{Proto-WCFLOBDD for $\lambda x. \bar{0}$ with $2^k$ variables}
\Begin{
\lIf{k == 0}{\Return RepresentativeDontCareGrouping($\bar{0},\bar{0}$)}
InternalGrouping g = new InternalGrouping(k)\;
g.AConnection = ConstantZeroProtoWCFLOBDD(k-1)\;
g.AReturnTuple = [1]\;
g.numberOfBConnections = 1\;
g.BConnections[1] = g.AConnection\;
g.BReturnTuples[1] = [1]\;
g.numberOfExits = 1\;
\Return RepresentativeGrouping(g)\;
}
}
\end{algorithm}
}

\Omit{
\begin{algorithm}
\caption{ConstantOneProtoCFLOBDD\label{Fi:ConstantOneProtoCFLOBDDAlgorithm}}
\small{
\Input{int k -- level of the topmost grouping}
\Output{Proto-CFLOBDD for $\lambda x. \bar{1}$ with $2^k$ variables}
\Begin{
\If{k == 0}{\Return RepresentativeDontCareGrouping(1,1)\;}
InternalGrouping g = new InternalGrouping(k)\;
g.AConnection = ConstantOneProtoCFLOBDD(k-1)\;
g.AReturnTuple = [1]\;
g.numberOfBConnections = 1\;
g.BConnections[1] = g.AConnection\;
g.BReturnTuples[1] = [1]\;
g.numberOfExits = 1\;
\Return RepresentativeGrouping(g)\;
}
}
\end{algorithm}
}

\subsubsection{Constant Functions.}
\figref{nodistinctfigs} shows the $\protoWCFLOBDD$s for the constant functions $f_{\bar{0}}(x) = \lambda x. \bar{0}$ and $ f_{\bar{1}}(x) = \lambda x. \bar{1}$ ($\texttt{ConstantZero}$ and $\texttt{ConstantOne}$, respectively).
These $\protoWCFLOBDD$s each have just one kind of level-$0$ grouping.
Their value tuples are $[\bar{0}]$ and $[\bar{1}]$, respectively.
\Omit{Many WCFLOBDDs algorithms perform equality tests on groupings with $\texttt{ConstantZero}$ and $\texttt{ConstantOne}$ to identify arguments on which to perform special-case handling.}

\subsubsection{Projection Function.}
Create a level-$k$ \protoWCFLOBDD $g$ for the function $f(k,i) \eqdef \lambda x_0 \dots x_{2^k -1}. x_i$.
If $i < 2^{k-1}$, $g$'s A-connection is a \protoWCFLOBDD for $f(k-1, i)$, and $g$ has two middle vertices. If $i \neq 0$, then the $1^{\textit{st}}$ B-connection of $g$ is {\tt ConstantOne} and the $2^{\textit{nd}}$ B-connection is {\tt ConstantZero};
if $i = 0$, the B-connections are the opposite.
If $i \geq 2^{k-1}$, $g$'s A-connection is {\tt ConstantOne} and $g$ has one middle vertex; $g$'s only B-connection is a \protoWCFLOBDD for $f(k-1, i - 2^{k-1})$. In all cases, $g$ has two exit vertices;
they lead to the value tuple $[\bar{1}, \bar{0}]$ in all cases except for $i = 0$, when they lead to $[\bar{0}, \bar{1}]$.

\subsubsection{Unary Operations.}
\textit{Scalar Multiplication.} Given a scalar $v \in \mathcal{D}$ and WCFLOBDD $C = \langle \textit{fw}, g, \textit{vt} \rangle$ for function $f$, the WCFLOBDD for $\lambda x. (v \cdot f(x))$ is
$
  C' = \begin{cases}
         \langle v \cdot \textit{fw}, g, \textit{vt} \rangle & v \neq \bar{0} \\
         \texttt{ConstantZero} & v = \bar{0}
       \end{cases}
$


\subsubsection{Pointwise Binary Operations.}
A binary operation $\op$ works \emph{pointwise} if, for two functions $f$ and $g$, $f \op g \eqdef \lambda x . f(x) \op g(x)$.
We discuss $\op \in \{ \cdot, + \}$.
The algorithms operate on WCFLOBDDs that are parametrized on the semi-field $\mathcal{D}$, and hence, the operations $\cdot$ and $+$ are specific to $\mathcal{D}$.

\begin{algorithm}[tb!]
\caption{Pointwise Multiplication\label{Fi:MBO}}
\small{
\Input{WCFLOBDDs n1 = $\langle \textit{fw1}, g1, \textit{vt1} \rangle$, n2 = $\langle \textit{fw2}, g2, \textit{vt2} \rangle$}
\Output{\twrchanged{WCFLOBDD} n = n1 $\cdot$ n2}
\Begin{
\tcp{Perform cross product}
    Grouping$\times$PairTuple [g,pt] = PairProduct(g1,g2)\;  \label{Li:BAAR:CallPairProduct}
    ValueTuple deducedValueTuple = [ vt1[i1] $\cdot$ vt2[i2]~:~[i1,i2] $\in$ pt ]\;  \label{Li:BAAR:LeafValues}
     \tcp{Collapse duplicate leaf values, folding to the left}
    Tuple$\times$Tuple [inducedValueTuple,inducedReductionTuple] = CollapseClassesLeftmost(deducedValueTuple)\;  \label{Li:BAAR:CollapseLeafValues}
    Grouping$\times$Weight [g', $\textit{fw}$] = Reduce(g, inducedReductionTuple, deducedValueTuple)\;  \label{Li:BAAR:CallReduce}
    WCFLOBDD n = RepresentativeCFLOBDD($\textit{fw} \cdot \textit{fw1} \cdot \textit{fw2}$, g', inducedValueTuple)\;
    \Return n\;
}
}
\end{algorithm}

\smallskip
\noindent
\textit{Pointwise Multiplication (\algref{MBO}).}
Given the WCFLOBDDs for functions $f$ and $g$, the goal is to create the WCFLOBDD for their pointwise product, $f \cdot g$.
Let $c_1 = \langle \textit{fw}_1, g_1, v_1\rangle$ and $c_2 = \langle \textit{fw}_2, g_2, v_2 \rangle$ be the WCFLOBDDs that represent $f$ and $g$, respectively.
As with BDDs, such operations on WCFLOBDDs can be implemented via a two-step process:
(i) create a cross-product of $c_1$ and $c_2$, and
(ii) perform a reduction step on the result of step (i).
The cross-product
(called PairProduct)
is performed recursively on the A-connection groupings, followed by the B-connection groupings
(see \figref{PairProductIllustration} and \algref{PairProduct}).
The cross-product of two groupings yields a tuple of the form ($g$, $[pt]$), where $g$ is the resultant grouping and $pt$ is a sequence of index-pairs.
The index-pairs in $pt$ indicate the B-connection groupings on which cross-product operations need to be performed.
\figref{PairProductIllustration}(c) shows the cross-product of the WCFLOBDDs that represent
$H_2 = \frac{1}{\sqrt{2}}\left[\begin{smallmatrix}
        1 & 1\\
        1 & -1
    \end{smallmatrix}\right]$
and $I_2 = \left[\begin{smallmatrix}
        1 & 0\\
        0 & 1\\
        \end{smallmatrix}\right]$.

\begin{figure}[tb!]
    \centering
    \begin{subfigure}{0.32\linewidth}
        \includegraphics[width=\linewidth]{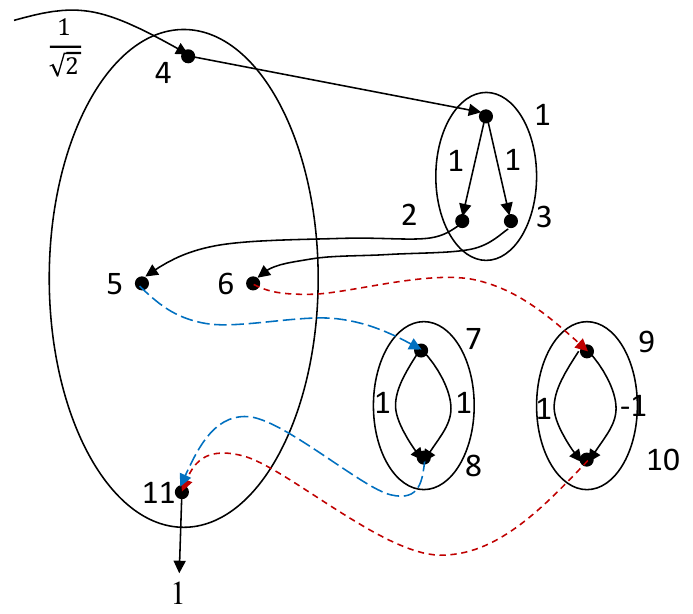}
        \caption{$H_2 = \frac{1}{\sqrt{2}}\left[\begin{smallmatrix}
        1 & 1\\
        1 & -1
    \end{smallmatrix}\right]$}
    \end{subfigure}
    \begin{subfigure}{0.32\linewidth}
        \includegraphics[width=\linewidth]{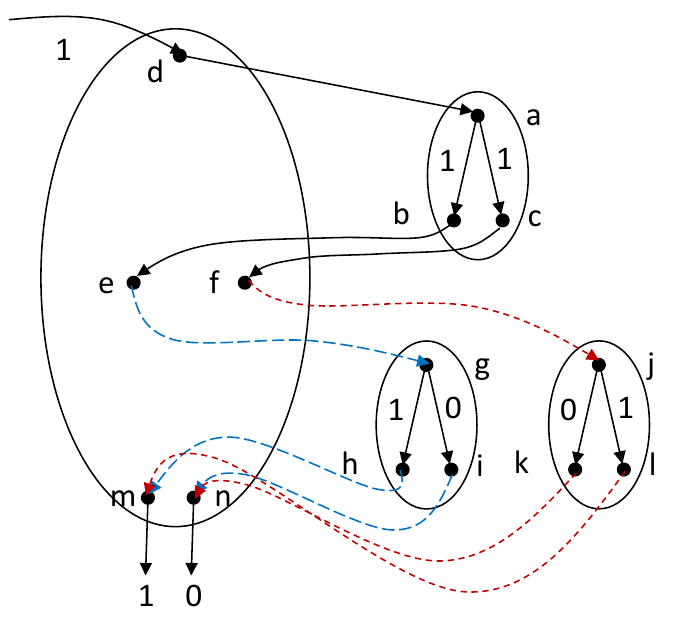}
        \caption{$I_2 = \left[\begin{smallmatrix}
        1 & 0\\
        0 & 1\\
        \end{smallmatrix}\right]$}
    \end{subfigure}
    \begin{subfigure}{0.32\linewidth}
        \includegraphics[width=\linewidth]{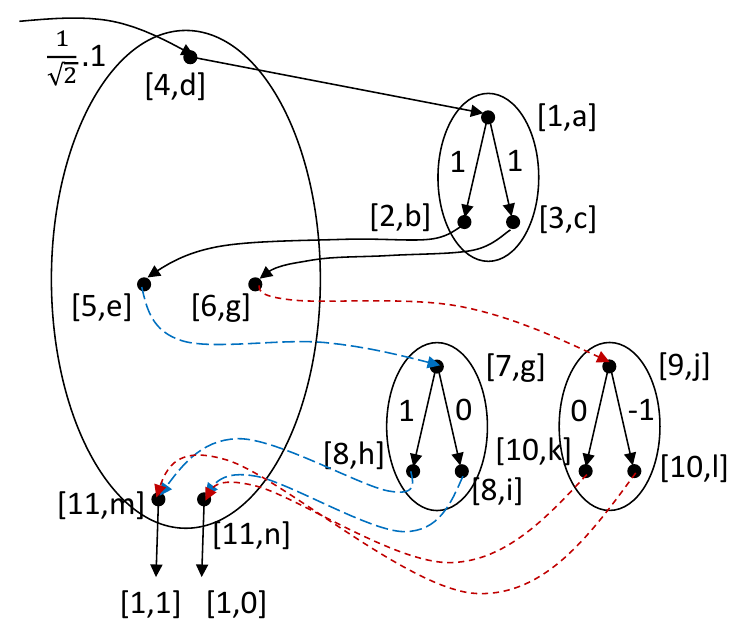}
        \caption{Result of ${\tt PairProduct}$ on (a) and (b)}
    \end{subfigure}
    \caption{Illustration of ${\tt PairProduct}$ 
      of the WCFLOBDDs for the matrices $H_2$ and $I_2$.
    }
    \label{Fi:PairProductIllustration}
    \vspace{-3ex}
\end{figure}

\begin{algorithm}[tb!]
\caption{PairProduct\label{Fi:PairProduct}}
\small{
\Input{Groupings g1, g2}
\Output{Grouping g: product of g1 and g2; PairTuple ptAns: tuple of pairs of exit vertices}
\Begin{
    \lIf {g1 and g2 are both ConstantOneProtoCFLOBDD \label{Li:PPConstantOneStart}}  
    { \Return [ g1, [[1,1]] ]}   
    \lIf {g1 or g2 is ConstantZeroProtoCFLOBDD }  
    { \Return [ ConstantZeroProtoCFLOBDD(g1.level), [[1,1]] ]}   
    \lIf {g1 is ConstantOneProtoCFLOBDD}
    {\Return [ g2, [[1,k]~:~k $\in$ [1..g2.numberOfExits]] ]}
    \lIf {g2 is ConstantOneProtoCFLOBDD}
    {\Return [ g1, [[k,1]~:~k $\in$ [1..g1.numberOfExits]] ]}      \label{Li:PPConstantOneEnd}
\tcp{Elided: similar cases for other base cases, with appropriate pairings of exit vertices}
\If {g1 and g2 are fork groupings  \label{Li:PPBothForkGroupings}} {
    ForkGrouping g = new ForkGrouping(g1.lw $\cdot$ g2.lw, g1.rw $\cdot$ g2.rw)\;
\Return [ g, [[1,1],[2,2]] ]\;
} \label{Li:PPBothForkGroupingsEnd}
\tcp{ Pair the A-connections}
    Grouping$\times$PairTuple [gA,ptA] = PairProduct(g1.AConnection, g2.AConnection)\;
    InternalGrouping g = new InternalGrouping(g1.level)\;
    g.AConnection = gA\;
    g.AReturnTuple = [1..$|$ptA$|$]\tcp*[r]{Represents the middle vertices}
    g.numberOfBConnections = $|$ptA$|$\;  \label{Li:PPNumberOfBConnections}
    \tcp{Pair the B-connections, but only for pairs in ptA }
    \tcp{Descriptor of pairings of exit vertices}
    Tuple ptAns = []\;  \label{Li:ptAnsDeclaration}

    \tcp{Create a B-connection for each middle vertex}
    \For {$j \leftarrow 1$ \KwTo $|ptA|$}{
        Grouping$\times$PairTuple [gB,ptB] = PairProduct(g1.BConnections[ptA(j)(1)], g2.BConnections[ptA(j)(2)])\;
        g.BConnections[j] = gB \;
        \tcp{Create g.BReturnTuples[j], and augment ptAns as necessary}
        g.BReturnTuples[j] = [] \;        \label{Li:PPExitVertexLoopStart}
        \For{$i \leftarrow 1$ \KwTo $|ptB|$} {
            c1 = g1.BReturnTuples[ptA(j)(1)](ptB(i)(1))\tcp*[r]{a g1 exit}
            c2 = g2.BReturnTuples[ptA(j)(2)](ptB(i)(2)) \tcp*[r]{a g2 exit}
            \eIf(\tcp*[f]{Not a new exit vertex of g}){[c1,c2] $\in$ ptAns}
                {
                index = the k such that ptAns(k) == [c1,c2] \;
                g.BReturnTuples[j] = g.BReturnTuples[j] $||$ index \;
                }
            (\tcp*[f]{Identified a new exit vertex of g}){
                g.numberOfExits = g.numberOfExits + 1 \;
                g.BReturnTuples[j] = g.BReturnTuples[j] $||$ g.numberOfExits \;
                ptAns = ptAns $||$ [c1,c2] \; \label{Li:PPConcatenateExitPair}
            }
        }  \label{Li:PPExitVertexLoopEnd}
    }
    \Return [RepresentativeGrouping(g), ptAns]\; \label{Li:PPTabulateAnswer}
}
}
\end{algorithm}

At level-$0$ and the topmost level, the
PairProduct algorithm performs the following steps:
\begin{itemize}
  \item
    Level-$0$:
    Let the two level-$0$ groupings be $g_1^0$ with weights $(lw_1, rw_1)$ and $g_2^0$ with weights $(lw_2, rw_2)$.
    The return value is a new level-$0$ grouping $g$, with weights $(lw_1 \cdot lw_2, rw_1 \cdot rw_2)$, along with a sequence of index-pairs on which the cross-product is performed next.
    For example, if $g_1^0$ is a fork-grouping and $g_2^0$ is a don't-care grouping, then $pt = [(1,1), (2,1)]$.
  \item
    Topmost level:
    If $(g, pt)$ is returned after performing the cross-product, the indices in $pt$ are indices into the value tuples of $c_1$ and $c_2$.
    A new value tuple $v$ is constructed accordingly.
    For instance, if $v_1 = [\bar{1}, \bar{0}]$, $v_2 = [\bar{0}, \bar{1}]$, and $pt = [(1,1), (1, 2), (2, 1)]$, then $v$ is
    $[ (v_1[1] \cdot v_2[1]), (v_1[1] \cdot v_2[2]), (v_1[2] \cdot v_2[1])] = [\bar{1} \cdot \bar{0}, \bar{1} \cdot \bar{1}, \bar{0} \cdot \bar{0}] = [\bar{0}, \bar{1}, \bar{0}]$.
\end{itemize}

\begin{algorithm}[tb!]
\caption{CollapseClassesLeftmost\label{Fi:InducedTuples}}
\small{
\Input{Tuple equivClasses}
\Output{Tuple$\times$Tuple [projectedClasses, renumberedClasses]}
\Begin{
\tcp{Project the tuple equivClasses, preserving left-to-right order, retaining the leftmost instance of each class}
\label{Li:IT:TupleCollapseStart}
Tuple projectedClasses = [equivClasses(i) : i $\in$ [1..$|$equivClasses$|$] $|$ i = min\{j $\in$ [1..$|$equivClasses$|$] $|$ equivClasses(j) = equivClasses(i)\}]\;
\label{Li:IT:TupleCollapseEnd}
\tcp{Create tuple in which classes in equivClasses are renumbered according to their ordinal position in projectedClasses}
Map orderOfProjectedClasses = \{[x,i]: i $\in$ [1..$|$projectedClasses$|$] $|$ x = projectedClasses(i)\}\;
Tuple renumberedClasses = [orderOfProjectedClasses(v) : v $\in$ equivClasses]\;
\Return [projectedClasses, renumberedClasses]\;
}
}
\end{algorithm}

When the resulting value tuple has duplicate entries, such as $v = [\bar{0}, \bar{1}, \bar{0}]$ in the example above, it is necessary to perform a reduction step, \texttt{Reduce}
(\algref{Reduce} in \sectref{canonicalness}),
to maintain the WCFLOBDD structural invariants.
In this case, \texttt{Reduce} folds together the first and third exit vertices of $g$, which are both mapped to $\bar{0}$ by $v$.
In \algref{MBO}, the information that directs the reduction step is obtained by calling CollapseClassesLeftmost (\algref{InducedTuples}), which returns two tuples:
inducedValueTuple (here, $[\bar{0}, \bar{1}]$) and inducedReductionTuple (here, $[1, 2, 1]$):
inducedValueTuple consists of the leftmost occurrences of $\bar{0}$ and $\bar{1}$ in $v$, in the same left-to-right order in which they occur in $v$;
inducedReductionTuple indicates where each occurrence of $\bar{0}$ and $\bar{1}$ in $v$ is mapped to in inducedValueTuple.
\texttt{Reduce} traverses $g$ backwards to create $g'$, a reduced version of $g$ that, together with inducedValueTuple as its terminal values, satisfies the WCFLOBDD structural invariants.
The result from $c_1 \cdot c_2$ is the WCFLOBDD $\langle \textit{fw}_1 \cdot \textit{fw}_2, g', \textrm{inducedValueTuple} \rangle$.

Just as there can be multiple occurrences of a given node in a BDD, there can be multiple occurrences of a given grouping in a WCFLOBDD.
To avoid a blow-up in costs, binary operations need to avoid making repeated calls on a given pair of groupings $h_1 \in c_1$ and $h_2 \in c_2$.
Assuming that the hashing methods used for hash-consing and function caching run in expected unit-cost time, the cost of
PairProduct
is bounded by the product of the sizes of the two-argument WCFLOBDDs, and the cost of \texttt{Reduce} is bounded by the product of the sizes of its input and output WCFLOBDDs.


\smallskip
\noindent
\textit{Pointwise Addition 
(\algrefs{ABO}{WeightedPairProduct} in \sectref{PointwiseAddition})
.}
Given the WCFLOBDDs $f = \langle \textit{fw}_1, h_1, v_1\rangle$ and $g = \langle \textit{fw}_2, h_2, v_2 \rangle$, the goal is to create the WCFLOBDD for their pointwise sum, $f + g$.
The algorithm performs a kind of ``weighted'' cross-product of $f$ and $g$.
The algorithm is similar to the one for pointwise multiplication, but with a few 
embellishments---in particular, weights are aggregated at each level and used in the recursive computation at the next level.
The weighted cross-product construction is performed recursively on the A-connections, followed by additional recursive calls on the B-connections according to sets of B-connection indices included in the answer returned from the call on the A-connections.
In addition to two groupings, $g_1$ and $g_2$, the cross-product also takes as input two weights, $p_1$ and $p_2$.
To get things started, the top-level call is WeightedPairProduct($h_1$,$h_2$, $\textit{fw}_1$, $\textit{fw}_2$).

The return value takes the form $[g, pt]$, where $g$ is a grouping that is temporary (and later will be updated with actual weights) and $pt$ is a sequence of tuples of the form $\langle (q_1, i_1), (q_2, i_2) \rangle$.
Values such as $i_1$ and $i_2$ are exit-vertex indices---used to determine what further recursive calls to invoke---and $q_1$ and $q_2$ are weights that play the role of $p_1$ and $p_2$ in those recursive
calls---i.e., WeightedPairProduct($\_$, $\_$, $q_1$, $q_2$), where the groupings used as the first and second arguments are based on $i_1$ and $i_2$, respectively.

\emph{Level $0$}:
Base-case processing happens at level $0$.
Let the two level-$0$ groupings be $g_1^0$ with weights $(lw_1, rw_1)$ and $g_2^0$ with weights $(lw_2, rw_2)$.
Suppose that the input weights are $p_1$ and $p_2$.
A new level-$0$ grouping $g$ is created with weights $(x_1, x_2)$, where $x_1 = \bar{1}$ in all cases except when $p_1 \cdot lw_1 = \bar{0}$, in which case $x_1 = \bar{0}$, and
$x_2 = \bar{1}$ in all cases except when $p_2 \cdot lw_2 = \bar{0}$, in which case $x_2 = \bar{0}$.
The return values are $g$ and $pt$, where $pt$ contains two tuples indicating which cross-products are to be performed subsequently.
The following table shows how $pt$ is constructed:
{\small
\begin{equation*}
  \begin{array}{ll|l}
    \multicolumn{1}{c}{g_1^0} & \multicolumn{1}{c}{g_2^0} & \multicolumn{1}{c}{pt} \\
    \hline
    \texttt{DontCareGrouping} & \texttt{DontCareGrouping}
         & [\langle (p_1 \cdot lw_1, 1), (p_2 \cdot lw_2, 1) \rangle, \langle (p_1 \cdot rw_1, 1), (p_2 \cdot rw_2, 1) \rangle] \\
    \hline
    \texttt{ForkGrouping} & \texttt{DontCareGrouping}
         & [\langle (p_1 \cdot lw_1, 1), (p_2 \cdot lw_2, 1) \rangle, \langle (p_1 \cdot rw_1, 2), (p_2 \cdot rw_2, 1) \rangle] \\
    \hline
    \texttt{DontCareGrouping} & \texttt{ForkGrouping}
         & [\langle (p_1 \cdot lw_1, 1), (p_2 \cdot lw_2, 1) \rangle, \langle (p_1 \cdot rw_1, 1), (p_2 \cdot rw_2, 2) \rangle] \\
    \hline
    \texttt{ForkGrouping} & \texttt{ForkGrouping}
         & [\langle (p_1 \cdot lw_1, 1), (p_2 \cdot lw_2, 1) \rangle, \langle (p_1 \cdot rw_1, 2), (p_2 \cdot rw_2, 2) \rangle] \\
    \hline
  \end{array}
\end{equation*}}

\emph{Topmost level}:
Suppose that $(\textit{fg}, \textit{pt})$ is returned by the top-level call on WeightedPairProduct.
Some post-processing takes place to set up an appropriate call on \texttt{Reduce}.
In each tuple $\langle (q_1, i_1), (q_2, i_2) \rangle$ in $\textit{pt}$, $i_1$ and $i_2$ are indices into the value tuples $v_1$ and $v_2$, respectively, of $f$ and $g$.
For instance, if $v_1 = [\bar{1}, \bar{0}]$, $v_2 = [\bar{1}]$, and
$pt = [\langle (d_1, 1), (d_2, 1) \rangle, \langle (d_3, 2), (d_4, 1) \rangle]$, where $d_1, d_2, d_3, d_4 \in \mathcal{D}$,
then a tuple $v$ is created: $v = [ (d_1 \cdot v_1[1] + d_2 \cdot v_2[1]), (d_3 \cdot v_1[2] + d_4 \cdot v_2[1])] = [d_1 + d_2, d_4]$.
$\textit{fg}$ is then reduced with respect to $v$, i.e., $v$ is propagated through $\textit{fg}$
to create $\textit{fg}'$ so that
(i) the edges of level-0 groupings in $\textit{fg}'$ hold appropriate weights, and
(ii) $\textit{fg}'$ satisfies the structural invariants of \defref{StructuralInvariants}.
The reduction operation also returns the new factor weight $\textit{fw}'$.
For our example, assuming that neither $d_1 + d_2$ nor $d_4$ is $\bar{0}$, the overall top-level tuple of terminal values would be $v' = [\bar{1}]$:
$v'$ is obtained from $v$ by replacing every non-zero element in $v$ with $\bar{1}$---to obtain $v'' = [\bar{1}, \bar{1}]$---and
then $v'$ retains only the first occurrences of the different values in $v''$.
The resultant WCFLOBDD is
$h = f + g = \langle \textit{fw}', \textit{fg}', v'\rangle$.

The pseudo-code is
given as
~\algrefs{ABO}{WeightedPairProduct} in \sectref{PointwiseAddition}.
Note that several optimizations can be performed for special cases, such as $\texttt{ConstantZero}$ or cases when both groupings are equal, 
to avoid performing a traversal of all the levels of groupings.

\subsubsection{Representing Matrices and Vectors.}
When matrices are represented with WCFLOBDDs, the variables correspond to the bits of the matrix's row and column indices.
For a matrix $M$ of size $2^{n} \times 2^{n}$, the WCFLOBDD representation has $2n$ variables $(x_0, \ldots, x_{n-1})$ and $(y_0, \ldots, y_{n-1})$, where the $x$-variables are the
row-index bits and the $y$-variables are the column-index
bits.
Typically, we use a variable ordering in which the $x$ and $y$ variables are interleaved,
$x_0$ and $y_0$ are the most-significant bits, and $x_{n-1}$ and $y_{n-1}$ are the least-significant bits.
The nice property of this ordering is that, as we work through each pair of variables in an assignment, the matrix elements that remain ``in play'' represent a sub-block of $M$. 

When vectors of size $2^n \times 1$ are represented using WCFLOBDDs, the variables $(x_0, \ldots, x_{n-1})$ correspond to the bits of the vector's row index.
Typically, $x_0$ is the most-significant bit and $x_{n-1}$ is the least-significant bit.

\subsubsection{Kronecker Product.}
Consider two matrices $M_1$ and $M_2$ represented by WCFLOBDDs $c_1 = \langle \textit{fw}_1, g_1, v_1 \rangle$ and $c_2 = \langle \textit{fw}_2, g_2, v_2 \rangle$.
Their Kronecker product $M = M_1 \otimes M_2$ is performed by using the 
``procedure-call'' mechanism of WCFLOBDDs to ``stack'' the Boolean variables of $c_1$ on top of those of $c_2$.
If $c_1$ or $c_2$ is $\texttt{ConstantZero}$, then $g = \texttt{ConstantZero}$.
Otherwise, a new grouping $g$ is created with $g_1$ as the A-connection of $g$, and $g_2$ as the B-connection(s) of $g$. If there exists an exit vertex $e$ of $g_1$ that has a terminal value of $\bar{0}$,
then the middle vertex of $g$ that is connected to $e$ has a B-connection to a \texttt{ConstantZeroProtoWCFLOBDD}.
The terminal values of $g$ are as follows:
\begin{equation*}
  v = \begin{cases}
        v_2 &
             \text{if $v_1$ == $[\bar{1}]$ (i.e., no such $e$ exists)}
        \\
        [\bar{0},\bar{1}] & 
            \text{if $e$ is the $1^{\textit{st}}$ exit vertex of $g_1$}
             \\
        [\bar{1},\bar{0}]
        &
              \text{if $e$ is the $2^{\textit{nd}}$ exit vertex of $g_1$ and $v_2 = [\bar{1}]$}
        \\
        v_2 & \text{otherwise}
      \end{cases}
\end{equation*}

The resultant WCFLOBDD $c = c_1 \otimes c_2$ is $\langle \textit{fw}_1 \cdot \textit{fw}_2, g, v \rangle$.

\subsubsection{Matrix Multiplication.}
\twrchanged{
Matrix multiplication is performed by a recursive divide-and-conquer algorithm, where the divide step performs a block decomposition that reduces the problem size to $\sqrt{N} \times \sqrt{N}$, and the (recursive) conquer step is similar to the standard cubic-time algorithm.
In such a divide-and-conquer algorithm, one solves subproblems on A-connections and then B-connections.\footnote{
\twrchanged{
  Some other divide-and conquer algorithms on WCFLOBDDs solve subproblems on B-connections and then A-connections.
}
}
In essence, one splits on half of the Boolean variables, which leads to $O(\sqrt{N})$ problems, each of size $O(\sqrt{N})$ \cite[Fig.\ 12]{TOPLAS:SCR24}.
}

\twrchanged{
Consider two $N \times N$ matrices $P$ and $P'$, represented by WCFLOBDDs $C = \langle \textit{fw}, g, v \rangle$ and $C' = \langle \textit{fw}', g', v' \rangle$
using the interleaved-variable order.
The A-connections of $g$ and $g'$ represent the commonalities in sub-blocks of $P$ and $P'$, respectively, of size $\sqrt{N} \times \sqrt{N}$ and the B-connections of $g$ and $g'$ represent the sub-matrices of $P$ and $P'$, respectively, of size $\sqrt{N}\times \sqrt{N}$.
The matrix-multiplication algorithm is recursively called on the A-connections of $g$ and $g'$, followed by the B-connections, based on information returned by the call on the A-connections, and then possibly some matrix-addition operations.
}

The challenge that we face is that at all levels below top-level,
\twrchanged{
various sub-matrices of the left argument need to be multiplied by various sub-matrices of the right argument, and added together.
However, for the computation performed at a given level of the WCFLOBDD, the algorithm has neither values nor sub-matrices at hand.
Those (unknown) sub-matrices correspond to the exit vertices of the groupings that the current invocation of the algorithm was passed as the left and right arguments.
Call these sets of exit vertices $\EV$ and $\EV’$, respectively.
We can use the elements of $\EV$ and $\EV’$ as variables, and compute non-top-level matrix multiplications symbolically.
It turns out that the symbolic information we need to keep takes the form of bilinear polynomials over $\EV$ and $\EV’$, consisting of summands of the form $c \cdot \ev_i \cdot \ev’_j$, where $\ev_i \in \EV$ and $\ev’_j \in \EV’$.
}

\begin{figure}[tb!]
\centering
\twrchanged{
\begin{tabular}{@{\hspace{0ex}}ccc@{\hspace{0ex}}}
  \begin{tabular}{@{\hspace{0ex}}c@{\hspace{0ex}}}
    \includegraphics[width=.32\linewidth]{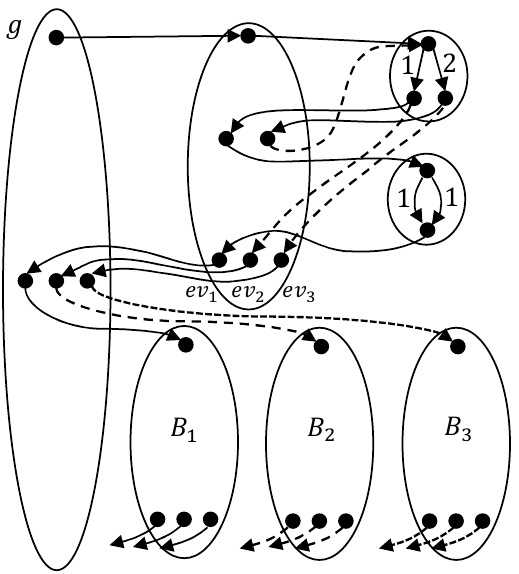}
  \end{tabular}
  &
  \begin{tabular}{@{\hspace{0ex}}c@{\hspace{0ex}}}
    \includegraphics[width=.35\linewidth]{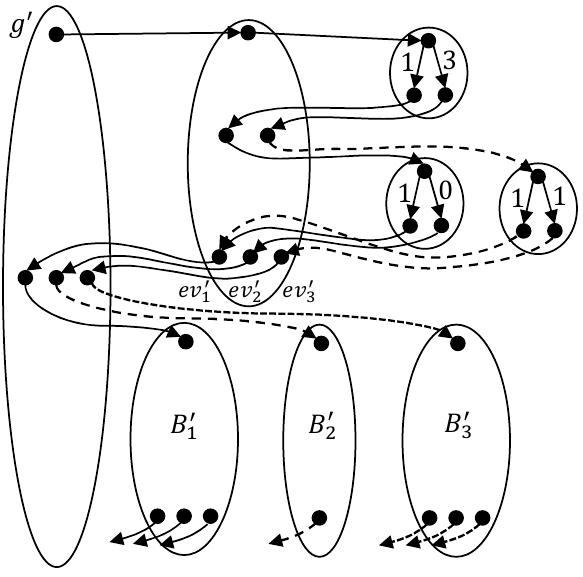}
  \end{tabular}
  &
  \begin{tabular}{@{\hspace{0ex}}c@{\hspace{0ex}}}
    \includegraphics[width=.29\linewidth]{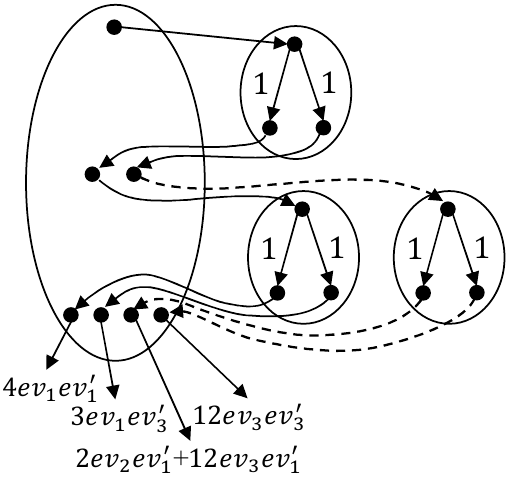}
  \end{tabular}
  \\
  {\small (a)} & {\small (b)} & {\small (c)}
\end{tabular}
}
\vspace{-1.5ex}
\caption{
\twrchanged{
  Illustration of how bilinear polynomials over exit vertices of lower-level groupings arise in matrix multiplication.
  (a) Left argument of $g \times g'$; (b) right argument of $g \times g'$; (c) the level-1 structure that is constructed in the level-1 subproblem $g$.AConnection $\times$ $g'$.AConnection.
}
}
\label{Fi:MatrixMultiplicationIllustration}
\end{figure}

\twrchanged{
\begin{example}\label{Exa:MatrixMultiplicationIllustration}
\figref{MatrixMultiplicationIllustration} illustrates of how such bilinear polynomials over exit vertices arise in matrix multiplication.
\figref{MatrixMultiplicationIllustration}(a) and (b) show level-2 groupings, $g$ and $g'$, which are the left-hand and right-hand arguments of a matrix-multiplication (sub)problem.
The first step of this multiplication problem is to symbolically multiply the level-1 groupings $g$.AConnection and $g'$.AConnection.
In \figref{MatrixMultiplicationIllustration}(a) and (b), when considered as $2\times 2$ matrices over their respective exit vertices, $[\ev_1, \ev_2, \ev_3]$ and $[\ev_1', \ev_2', \ev_3']$,
$g$.AConnection and $g'$.AConnection are the matrices of bilinear functions shown on the left side of \eqref{TwoByTwoMatrixSymbolic}:
\begin{equation}
    \label{Eq:TwoByTwoMatrixSymbolic}
    \begin{bmatrix}
      \ev_1 & \ev_1 \\
      2\ev_2 & 4\ev_3
    \end{bmatrix}
    \times
    \begin{bmatrix}
      \ev'_1 & 0\ev'_2 \\
      3\ev'_1 & 3\ev'_3 
    \end{bmatrix}
    =
    \begin{bmatrix}
      4\ev_1 \ev'_1 & 3\ev_1 \ev'_3 \\
      2\ev_2 \ev'_1 + 12\ev_3 \ev'_1 & 12\ev_3 \ev'_3
    \end{bmatrix}\\
  \end{equation}
\figref{MatrixMultiplicationIllustration}(c) shows the level-1 preliminary structure that is constructed after multiplying
$g$.AConnection and $g'$.AConnection.
This structure represents the matrix on the right-hand side of \eqref{TwoByTwoMatrixSymbolic}.
Each exit vertex of \figref{MatrixMultiplicationIllustration}(c) is associated with a bilinear polynomial consisting of summands of the form $c \cdot ev_i \cdot ev’_j$, where $\ev_i \in \EV$ and $\ev’_j \in \EV’$.
(In essence, the structure is a CFLOBDD with a bilinear polynomial for each terminal value.)
The interpretation of the bilinear monomials from \figref{MatrixMultiplicationIllustration}(c) leads to five B-connection multiplications being performed: $B_1 \times B_1'$, $B_1 \times B_3'$, $B_2 \times B_1'$, $B_3 \times B_1'$, and $B_3 \times B_3'$.
Each multiplication produces a structure that is similar to the one shown in \figref{MatrixMultiplicationIllustration}(c).
Five scalar multiplications and one matrix addition are then performed.
These last steps again produce bilinear polynomials because bilinear polynomials are closed under linear arithmetic (see \eqref{BilinearPolynomialOperations}).
\end{example}
}

The matrix-multiplication algorithm works recursively level-by-level, creating WCFLOBDDs for which the entries in the value tuples are bilinear polynomials.
We call such value tuples \emph{MatMultTuple}s.
A \emph{MatMultTuple} is a sequence of bilinear polynomials over the exit vertices of two groupings $g_1$ and $g_2$.
Each bilinear polynomial $\bp$ is a map from a pair of exit-vertex indices to a coefficient (which is a value in $\mathcal{D}$):
$\bp \in \BP_{\EV,\EV'} \eqdef (\EV \times \EV') \rightarrow \mathcal{D}$, where $\EV, \EV'$ are the sets of exit vertices of $g_1$ and $g_2$.
To perform linear arithmetic on bilinear polynomials, we define
\begin{equation}
  \label{Eq:BilinearPolynomialOperations}
  \begin{array}{@{\hspace{0ex}}l@{\hspace{5.0ex}}r@{\hspace{0.5ex}}c@{\hspace{0.5ex}}l@{\hspace{0ex}}}
      \ZeroBP : \BP & \ZeroBP & \eqdef & \lambda (\ev,\ev') \, . \,\bar{0} \\
      + : \BP \times \BP \rightarrow \BP & \bp_1 + \bp_2 & \eqdef & \lambda (\ev,\ev') \, . \,\bp_1(\ev,\ev') + \bp_2(\ev,\ev')\\
      \cdot :\mathcal{D} \times \BP \rightarrow \BP & n \cdot \bp & \eqdef & \lambda (\ev,\ev') \, . \,n \cdot \bp(\ev,\ev')\\
  \end{array}
\end{equation}

The base case is for two level-$1$ proto-WCFLOBDDs, which correspond to a pair of $2 \times 2$ matrices.
The left and right decision-edges of a level-$1$ proto-WCFLOBDD's level-$0$ groupings hold weights, which serve as coefficients of \emph{linear} functions in the entries of each of the $2 \times 2$ matrices.
The result of multiplying two level-$1$ proto-WCFLOBDDs introduces bilinear polynomials.\footnote{
\twrchanged{
  As explained shortly, more complicated polynomials do not arise at levels $2, 3, \ldots$
}
}
Operationally, a bilinear polynomial can be represented by a set of triples of the form $\{ \ldots, [(i, j), c], \ldots \}$, where $[(i, j), c]$ represents a term of the form $c \cdot \ev_i \cdot \ev’_j$.

\begin{example}
Returning to \exref{MatrixMultiplicationIllustration}, each entry in the matrix on the right-hand side of \eqref{TwoByTwoMatrixSymbolic} can be represented by a set of triples.
\[
  \begin{bmatrix}
    \{[(1,1),4]\} & \{[(1,3),3]\} \\
    \{[(2,1),2], [(3, 1), 12]\} & \{[(3,3),12]\} 
  \end{bmatrix}  
\]
The \emph{MatMultTuple} is the listing of exit vertices for interleaved-variable order:
\begin{equation*}
  [ \{[(1,1),4]\}, \{[(1,3),3]\},  \{[(2,1),2], [(3,1),12]\}, \{[(3,3),12]\} ].
\end{equation*}
\end{example}

More abstractly, the function $\bp$ holds the coefficients of the bilinear polynomial over variable pairs $\{\ev \in \EV\} \times \{\ev' \in \EV’\}$, and $\bp$ implicitly denotes the expression
\[
    \sum_{{\ev \in \EV} \times {\ev \in \EV’}} \bp(\ev,\ev’) \cdot \ev \cdot \ev’.
\]

As symbolic matrix multiplication is performed, a key operation is to ``evaluate'' a bilinear polynomial with respect to a binding of exit-vertices to (other) bilinear polynomials.
\twrchanged{
So why don’t we get quartic polynomials, and then even higher-degree polynomials?
Because all operations are performed on bilinear polynomials over the same variable sets (i.e., sets of exit vertices).
Consider again \figref{MatrixMultiplicationIllustration}(c) and the bilinear polynomials associated with each exit vertex.
Each such polynomial can be considered to be a bilinear polynomial over the middle vertices of $g$ and $g'$, and thus consists of terms of the form $c \cdot m_i \cdot m’_j$.
Each such term is treated as a directive to multiply $g.\textrm{BConnections}[i]$ and $g’.\textrm{BConnections}[j]$.
}

\twrchanged{
For instance, the third exit vertex of \figref{MatrixMultiplicationIllustration}(c) has the associated bilinear polynomial $\textit{pA} = 2\ev_2 \ev'_1 + 12\ev_3 \ev'_1$.
The first term, $2\ev_2 \ev'_1$, is evaluated with respect to the binding $[\ev_2 \mapsto B_2, \ev'_1 \mapsto B'_1]$, leading to $B_2 \times B'_1$;
the second term, $12\ev_3 \ev'_1$, is evaluated with respect to the binding $[\ev_3 \mapsto B_3, \ev'_1 \mapsto B'_1]$, leading to $B_3 \times B'_1$. 
A multiplication of $g.\textrm{BConnections}[i]$ and $g’.\textrm{BConnections}[j]$ produces a matrix $m^{i,j}$ whose $(k,l)^{\textit{th}}$ entry is a bilinear polynomial $\textit{pB}^{i,j}_{k,l}$ over the exit vertices of $g.\textrm{BConnections}[i]$ and $g’.\textrm{BConnections}[j]$.
Now---and this observation is the key reason why everything stays bilinear---$g.\textrm{BReturnTuples}[i]$ and $g’.\textrm{BReturnTuples}[j]$ are used to convert $\textit{pB}^{i,j}_{k,l}$ into a bilinear polynomial $p^{i,j}_{k,l}$ over the \emph{exit vertices of $g$ and $g’$}.
Consequently, the remaining steps of ``evaluating'' $\textit{pA}$---multiplying a bilinear polynomial such as $p^{i,j}_{k,l}$ by a constant and addition of bilinear polynomials---are all performed on bilinear polynomials with the same variable sets, namely, the exit vertices of $g$ and $g’$.
Bilinear polynomials are closed under these operations (\eqref{BilinearPolynomialOperations}).
}

At top level, we form a preliminary value tuple $w$ by evaluating each bilinear polynomial $\bp$ in the top-level \emph{MatMultTuple} as follows:
\begin{equation*}
  \langle v, v' \rangle (\bp)
    \eqdef \sum \left\{ \bp(\ev, \ev') \cdot v(\ev) \cdot v'(\ev') \mid \ev \in \EV, \ev' \in \EV' \right\}
\end{equation*}
By structural invariant (\ref{It:OutermostLevel:ValueTuple}), the value tuple $\tilde{w}$ of the answer WCFLOBDD must be one of $\{ [\bar{0}], [\bar{1}], [\bar{0}, \bar{1}], [\bar{1}, \bar{0}] \}$.
Thus, $\tilde{w}$ is constructed from $w$ by replacing all non-$\bar{0}$ entries of $w$ with $\bar{1}$ and removing duplicates.
To ensure that the final answer satisfies the WCFLOBDD structural invariants, \texttt{Reduce} is called (\algref{Reduce} in \sectref{canonicalness}), which is passed the pattern of repeated values in $w$ (along with $w$ itself).

Certain optimizations can be performed to avoid traversing all the levels of the argument groupings:
when at least one argument is $\texttt{ConstantZero}$, then all paths have weight $\bar{0}$, and the computation can be short-circuited.
Similarly, when one of the arguments is the identity matrix, we can return the other argument.


\subsubsection{Sampling.}
\label{Se:SamplingOverview}
A WCFLOBDD with no non-negative edge weights can be considered to represent a discrete distribution over the set of assignments to the Boolean variables.
An assignment---or equivalently, the corresponding matched path—is considered
to be an elementary event. The probability of a matched path $p$ is the weight of
$p$ divided by the sum of the weights of all matched paths of the WCFLOBDD.

\textit{Weight Computation.} To sample an assignment from a WCFLOBDD, we first compute the weight corresponding to every exit vertex of every grouping $g$:
the weight of the $i^{\textit{th}}$ exit vertex of grouping $g$ is the sum of weights of all paths leading to $i$ from the entry vertex of $g$. 
This information at every grouping $g$ is computed by recursively calling $g$'s A-connection and B-connection groupings and using the information to compute the information for $g$.

At level-$0$ (base case), if $g$ is a \texttt{DontCareGrouping}, the weight of the exit vertex is the sum of the weights of the two edges, $(lw + rw)$.
If $g$ is a \texttt{ForkGrouping},
the weight of exit vertex $1$ is $lw$, and the weight of exit vertex $2$ is $rw$.

\textit{Sampling.} To sample an assignment $a$ from the probability distribution represented by a WCFLOBDD,
half of the assignment, $a_A$, is sampled recursively from the A-connection, and the other half, $a_B$, is sampled recursively from one of the B-connections, where $a = a_A \, || \, a_B$.
Consider the outermost grouping $g$ and a given exit vertex $i$.
For each middle vertex $m$ of $g$, the sum of weights of the matched paths from the entry vertex of $g$ to $i$ that passes through $m$ forms a distribution $D$ on $g$'s middle vertices 
(see Eqn.~(\ref{Eq:sampling-equation}), \sectref{Sampling})
.
To sample a matched path that leads to $i$ we (i) first sample the index ($m_{\textit{index}}$) of a middle vertex of $g$ according to $D$, (ii) recursively sample from $g$.AConnection with respect to the exit vertex that leads to $m_{index}$, (iii) recursively sample from $g$.BConnection[$m_{\textit{index}}$] with respect to the exit vertex that leads to $i$, and (iv) concatenate the sampled paths to obtain the sampled path of $g$.
(Steps (ii) and (iii) can be done in either order.)

At level-$0$ (base case), if $g$ is a \texttt{DontCareGrouping}, the assignment is sampled from ``$0$'' and ``$1$'' in proportion to the edge weights $(\textit{lw}, \textit{rw})$.
If $g$ is a \texttt{ForkGrouping}, the assignment ``$0$'' or ``$1$'' is chosen according to index $i$.

\section{Evaluation}
\label{Se:eval}


\begin{wrapfigure}{r}{0.6\textwidth}
    \centering
    \vspace{-10.0ex}
\resizebox{\linewidth}{!}{
    \begin{tabular}{|c|c|c|c|c|c|c|c|}
    \hline
    Bench- & \multirow{2}{*}{$\log(n)$} & \multicolumn{2}{c|}{WCFLOBDD} & \multicolumn{2}{c|}{WBDD} & \multicolumn{2}{c|}{CFLOBDD} \\
    \cline{3-8}
    mark & & Time (s) & Size & Time (s) & Size & Time (s) & Size \\
    \hline
    \multirow{6}{*}{B1} & 16 & \textbf{0.001} & 442 & 7.32 & 65536 & \textbf{0.001} & \textbf{329}\\
    & 17 & \textbf{0.001} & 470 & 14.51 & 131072 & \textbf{0.001} & \textbf{350}\\
    & 18 & \textbf{0.001} & 498 & 29.23 & 262144 & \textbf{0.001} & \textbf{371}\\
    & 19 & \textbf{0.001} & 526 & 57.91 & 524288 & \textbf{0.001} & \textbf{392}\\
    & 20 & \textbf{0.001} & 554 & 117.44 & 1048576 & \textbf{0.001} & \textbf{413}\\
    \cline{5-6}
    & 21 & \textbf{0.001} & 582 & \multicolumn{2}{c|}{Timeout} & \textbf{0.001} & \textbf{434}\\
    \hline
    \multirow{6}{*}{B2} & 16 & \textbf{0.001} & 654 & 7.34 & 65538 & \textbf{0.001} & \textbf{645}\\
    & 17 & \textbf{0.001} & 696 & 14.56 & 131074 & 0.002 & \textbf{687}\\
    & 18 & \textbf{0.001} & 738 & 29.06 & 262146 & 0.003 & \textbf{729}\\
    & 19 & \textbf{0.001} & 780 & 58.43 & 524290 & 0.005 & \textbf{771}\\
    & 20 & \textbf{0.001} & 822 & 118.44 & 1048578 & 0.007 & \textbf{813}\\
    \cline{5-6}
    & 21 & \textbf{0.001} & 864 & \multicolumn{2}{c|}{Timeout} & 0.013 & \textbf{855}\\
    \hline
    \multirow{6}{*}{B3} & 16 & \textbf{0.001} & 199 & 7.27 & \textbf{1} & 0.005 & 197\\
    & 17 & \textbf{0.001} & 211 & 14.4 & \textbf{1} & 0.01 & 209\\
    & 18 & \textbf{0.001} & 223 & 28.63 & \textbf{1} & 0.018 & 221\\
    & 19 & \textbf{0.001} & 247 & 57.17 & \textbf{1} & 0.035 & 233\\
    & 20 & \textbf{0.001} & 247 & 116.52 & \textbf{1} & 0.071 & 245\\
    \cline{5-6}
    & 21 & \textbf{0.001} & 259 & \multicolumn{2}{c|}{Timeout} & 0.142 & 257\\
    \hline
    \multirow{6}{*}{B4} & 16 & 0.51	& \textbf{307} & 7.29 & 32769 & \textbf{0.003} & 647\\
    & 17 & 1.03	& \textbf{326} & 14.52 & 65537 & \textbf{0.004} & 690\\
    & 18 & 2.13	& \textbf{345} & 28.91 &	131073 & \textbf{0.008} & 733\\
    & 19 & 4.51	& \textbf{364} & 57.93 &	262145 & \textbf{0.017} & 776\\
    & 20 & 10 &	\textbf{383} & 117.67 & 524289 & \textbf{0.04} & 819\\
    \cline{5-6}
    & 21 & 23.84 & \textbf{402} & \multicolumn{2}{c|}{Timeout} & 0.1 & 862\\
    \hline
    \multirow{6}{*}{B5} & 16 & \textbf{0.001} & \textbf{83} & 7.25 & \textbf{1} & \textbf{0.001} & \textbf{83}\\
    & 17 & \textbf{0.001} & \textbf{88} & 14.38 & \textbf{1} & \textbf{0.001} & \textbf{88}\\
    & 18 & \textbf{0.001} & \textbf{93} & 28.71 & \textbf{1} & \textbf{0.001} & \twrchanged{\textbf{93}}\\
    & 19 & \textbf{0.001} & \textbf{98} & 57.15 & \textbf{1} & \textbf{0.001} & \textbf{98}\\
    & 20 & \textbf{0.001} & \textbf{103} & 116.77 & \textbf{1} & \textbf{0.001} & \textbf{103} \\
    \cline{5-6}
    & 21 & \textbf{0.001} & \textbf{108} & \multicolumn{2}{c|}{Timeout} & \textbf{0.001} & \textbf{108}\\
    \hline
    \end{tabular}
}
    \caption{Execution times and sizes of WCFLOBDDs and WBDDs on synthetic benchmarks.
    (For B3, the WBDD size is 1 because $H \times H = I$, which is handled as a special case in MQTDD.)
    }
    \label{Fi:syntheticBenchmarks}
    \vspace{-7ex}
\end{wrapfigure}

Our experiments were designed to answer two questions:
\begin{itemize}[align=left, leftmargin=1cm]
  \item [\textbf{RQ1:}]
    Can WCFLOBDDs represent substantially larger functions than WBDDs?
  \item [\textbf{RQ2:}]
    In terms of time and space, can WCFLOBDDs outperform WBDDs and CFLOBDDs in a practical domain?
\end{itemize}
We ran all experiments on an 
\twrchanged{
Intel\,\textsuperscript{\textregistered} Xeon\,\textsuperscript{\textregistered} Gold 5218 CPU machine running Ubuntu OS
}
with 31GB RAM, 1000MHz CPU frequency with 64 CPUs. We also set the stack size to ``unlimited.''
\twrchanged{
The implementation is single-threaded and consists of about 2,500 lines of C++.
It also uses Quasimodo~\cite{DBLP:journals/corr/abs-2302-04349}, written in Python.
}

\subsection{Experiments to Answer RQ1}

For RQ1, we created five synthetic benchmarks, each defining a family of functions whose members (i) have different numbers of variables (all powers of $2$), and (ii) could occur as part of a matrix-multiplication sequence in a quantum-circuit simulation:

\smallskip
\qquad\begin{tabular}{l@{\hspace{5.0ex}}l@{\hspace{5.0ex}}l}
    B1: $I_n + X_n$    & B3: $H_n \times H_n$   &   B5: $H_n - H_n$  \\
    B2: $\CNOT_n(0, n-1) \times \CNOT_n(\frac{n}{2}-1, \frac{n}{2})$ & B4: $H_n \times I_n + I_n \times X_n$  & 
\end{tabular}


\smallskip
\noindent
$I_n$, $X_n$, and $H_n$ are the Identity, NOT, and Hadamard matrices, respectively, of size $2^{n/2} \times 2^{n/2}$.
$\CNOT_n(a, b)$ is a Controlled-NOT matrix of size $2^{n/2} \times 2^{n/2}$.
$\CNOT(a, b)(\langle x_0, x_1, \ldots, x_{n/2-1} \rangle)$ $=$
$\langle x_0, \ldots, x_a, \ldots,$ $(x_a {\oplus} x_b),$ $\ldots, x_{n/2-1} \rangle$, for $x \in \{0,1\}^{n/2}$.
We tested our implementation of WCFLOBDDs against WBDDs and CFLOBDDs.
We used the MQTDD package \cite{zulehner2019package} (an implementation of WBDDs) and a publicly available implementation of CFLOBDDs \cite{Code:CFLOBDDImplementation}.

\figref{syntheticBenchmarks} shows the performance of WCFLOBDDs and WBDDs in terms of (i) execution time, and (ii) size.
Here ``size'' means the number of vertices and edges for WCFLOBDDs and CFLOBDDS, and the number of nodes for WBDDs.
We find that on the synthetic benchmarks, WCFLOBDDs perform better than WBDDs.
\figref{syntheticBenchmarks} also shows the performance of CFLOBDDs on the synthetic benchmarks.
We find that WCFLOBDDs are comparable in size to CFLOBDDs on benchmarks B1-B3 and B5, and are smaller than CFLOBDDs on B4.
In terms of execution time, WCFLOBDDs are comparable to CFLOBDDs on B1 and B5, and better on B2 and B3.
\twrchanged{
CFLOBDDs perform better than WCFLOBDDs on B4 because B4 involves matrix addition (pointwise addition).
In CFLOBDDs, pointwise addition is fast (based on PairProduct), whereas the weight manipulations in WeightedPairProduct
(\algref{WeightedPairProduct})
have a substantial overhead.
}

\subsection{Experiments to Answer RQ2}

\begin{wrapfigure}{R}{0.6\textwidth}
    \centering
    \vspace{-7.0ex}
\resizebox{\linewidth}{!}{
    \setlength{\tabcolsep}{1.0pt}
    \begin{tabular}{|c|c|c|c|c|c|c|c|}
    \hline
    \multirow{2}{*}{Circuit} & \multirow{2}{*}{\#Qubits} & \multicolumn{2}{c|}{WCFLOBDD} & \multicolumn{2}{c|}{WBDD} & \multicolumn{2}{c|}{CFLOBDD} \\
    \cline{3-8}
    & & Time(s) & Size & Time(s) & Size & Time(s) & Size\\
    \hline
    \multirow{6}{*}{\GHZ} & 16 & 0.03 & 138 & \textbf{0.01} & \textbf{32} & 0.02 & 136\\
    & 256 & \textbf{0.03} & 250 & 0.08 & 512 & \textbf{0.03} & \textbf{248} \\
    & 4096 & \textbf{0.14} & 362 & 8.45 & 8192 & 0.16 & \textbf{360}\\
    & 32768 & \textbf{1.05} & 446 & 714.97 & 65536 & 1.39 & \textbf{444}\\
    \cline{5-6}
    & 65536 & \textbf{2.28} & 474 & \multicolumn{2}{c|}{\multirow{2}{*}{Timeout}} & 3.02 & \textbf{472}\\
    & 2097152 & \textbf{687.7} & 614 & \multicolumn{2}{c|}{} & 760.48 & \textbf{612}\\
    \hline
    \multirow{6}{*}{\BV} & 16 & 0.03 & 157 & \textbf{0.01} & \textbf{18} & 0.03 & 156\\
    & 256 & \textbf{0.07} & 714 & 0.12 & \textbf{258} & 0.16 & 713 \\
    & 4096 & \textbf{0.92} & 5491 & 22.83 & \textbf{4098} & 2.8 & 5490\\
    & 16384 & \textbf{4.52} & 16446 & 447.19 & \textbf{16386} & 13 & 16445\\
    \cline{5-6}
    & 65536 & \textbf{27.08} & 58698 & \multicolumn{2}{c|}{\multirow{2}{*}{Timeout}} & 67.97 & \textbf{58697}\\
    & 262144 & \textbf{287.87} & 218441 & \multicolumn{2}{c|}{} & 515.07 & \textbf{218440}\\
    \hline
    \multirow{7}{*}{\textit{DJ}} & 16 & 0.03 & 140 & \textbf{0.01} & \textbf{18} & 0.03 & 157\\
    & 256 & \textbf{0.08} & \textbf{244} & 0.15 & 258 & 0.16 & 269 \\
    & 4096 & \textbf{1.01} & \textbf{348} & 29.35 & 4098 & 2.58 & 381\\
    & 16384 & \textbf{4.59} & \textbf{400} & 618.48 & 16386 & 11.52 & 437\\
    \cline{5-6}
    & 65536 & \textbf{24.71} & \textbf{452} & \multicolumn{2}{c|}{\multirow{3}{*}{Timeout}} & 56.64 & 493\\
    & 262144 & \textbf{171.9} & \textbf{504} & \multicolumn{2}{c|}{} & 350.23 & 549\\
    \cline{7-8}
    & 524288 & \textbf{562.21} & \textbf{530} & \multicolumn{2}{c|}{} & \multicolumn{2}{c|}{Timeout}\\
    \hline
    \multirow{4}{*}{Simon} & 16 & 0.07 & 242 & \textbf{0.01} & \textbf{83} & 0.07 & 327\\
    & 256 & 1.62 & \textbf{942} & \textbf{0.39} & 1524 & 1.68 & 1524\\
    & 4096 & 124.8 & \textbf{7042} & 93.96 & 24563 & \textbf{92.05} & 10117\\
    & 8192 & 453.98 & \textbf{11902} & 439.74 & 49141 & \textbf{293.73} & 17130\\
    \hline
    \multirow{4}{*}{\QFT} & 4 & 0.03 & 63 & \textbf{0.01} & \textbf{5} & 0.02 & 88\\
    & 16 & 0.03 & 226 & \textbf{0.01} & \textbf{17} & 0.34 & 69736\\
    \cline{7-8}
    & 256 & \textbf{1.39} & 3366 & 1.73 & \textbf{257} & \multicolumn{2}{c|}{\multirow{2}{*}{Timeout}}\\
    & 2048 & \textbf{123.28} & 26678 & 777.5 & \textbf{2049} & \multicolumn{2}{c|}{}\\
    \hline
    \multirow{3}{*}{Grover} & 4 & 0.04 & 120 & \textbf{0.01} & \textbf{11} & 0.04 & 188\\
    & 8 & 0.16 & 199 & \textbf{0.02} & \textbf{23} & 1.33 & 1741\\
    \cline{7-8}
    & 16 & 8.8 & 318 & \textbf{0.51} & \textbf{47} & \multicolumn{2}{c|}{Timeout}\\
    \hline
    \begin{tabular}{@{}c@{}}
     Shor \\
    (15, 2)
    \end{tabular} & 4 & \textbf{0.22} & 205 & 0.25 & \textbf{26} & 322.4 & 2903\\
    \hline
    \begin{tabular}{@{}c@{}}
    Shor \\
    (21, 2)
    \end{tabular} & 5 & 2.81 & 441 & \textbf{0.19} & \textbf{84} & \multicolumn{2}{c|}{\multirow{6}{*}{Timeout}}\\
    \cline{1-6}
    \begin{tabular}{@{}c@{}}
    Shor \\
    (39, 2)
    \end{tabular} & 6 & 5.29 & 442 & \textbf{0.6} & \textbf{143} & \multicolumn{2}{c|}{}\\
    \cline{1-6}
    \begin{tabular}{@{}c@{}}
    Shor \\
    (95, 8)
    \end{tabular} & 7 & \multicolumn{2}{c|}{Timeout} & \textbf{684.35} & \textbf{1440} & \multicolumn{2}{c|}{}\\
    \hline
    \end{tabular}
}
    \caption{Execution times and sizes
    for quantum-circuit benchmarks}
    \label{Fi:quantum-benchmarks}
    \vspace{-3ex}
\end{wrapfigure}

We applied WBDDs, CFLOBDDs, and WCFLOBDDs to quantum-circuit simulation, which involves exponentially large vectors and matrices.
We used MQTDD \cite{zulehner2019package}, CFLOBDDs \cite{Code:CFLOBDDImplementation}, and WCFLOBDDs as back ends to Quasimodo~\cite{DBLP:journals/corr/abs-2302-04349}, a quantum-simulation tool that provides a convenient interface for attaching different back-end decision-diagram packages.
We used seven standard quantum circuits:
Greenberger–Horne–Zeilinger
(\GHZ), Bernstein–Vazirani
(\BV), Deutsch–Jozsa
(\textit{DJ}), Simon's algorithm\Omit{ (one is a slightly reordered circuit---reordering is done in the placement of the qubits)}, Quantum Fourier Transform (\QFT), Grover's algorithm, and Shor's algorithm.

For each circuit other than \GHZ and \QFT, a ``hidden'' string was initially sampled, and the oracle for the circuit was constructed based on the sampled string.
For \QFT, one of the basis states was sampled and the algorithm changed the basis state into a Fourier-transformed state.
For Shor's algorithm, we used the standard $2n+3$ circuit~\cite{beauregard2002circuit} and report the results for values $(N, a)$, where $N$ is the number to be factored and $a$ is a number used in the phase-estimation procedure.
The table in \figref{CircuitWidth} shows the \emph{circuit width} of each benchmark: the total number of qubits in the circuit as a function of $n$, the number of qubits for, e.g., the hidden string (reported in column 2 of \figref{quantum-benchmarks}):



\begin{figure}[tb!]
\centering
{\small
\begin{tabular}{|c|c|c|c|c|c|c|c|}
    \hline
     Circuit & \GHZ & \BV & \textit{DJ} & Simon & QFT & Grover & Shor \\
    \hline
    Width & $n$ & $n + 1$ & $n + 1$ & $2n$ &  $n$ & $2n-1$ & $2n+3$\\
    \hline
\end{tabular}
}
  \caption{Circuit widths of the quantum-circuit benchmarks.}
  \label{Fi:CircuitWidth}
\end{figure}

We ran each benchmark 10 times with a 15-minute timeout, starting from a random initial state.
\figref{quantum-benchmarks} reports average times and sizes of the final state vector.
WCFLOBDDs perform better than WBDDs and CFLOBDDs in most cases. 
The number of qubits that WCFLOBDD can handle is
2,097,152 for GHZ (1$\times$ over CFLOBDDs, 64$\times$ over WBDDs); 
262,144 for BV (1$\times$ over CFLOBDDs, 16$\times$ over WBDDs);
524,288 for DJ (2$\times$ over CFLOBDDs,
32$\times$ over WBDDs);
8,192 for Simon's algorithm
(1$\times$ over CFLOBDDs, 1$\times$ over WBDDs)
2,048 for QFT (128$\times$ over CFLOBDDs, 1$\times$ over WBDDs);
and 16 for Grover's algorithm (2$\times$ over CFLOBDDs, 1$\times$ over WBDDs).

For Simon's algorithm, the number of qubits handled by WCFLOBDDs, WBDDs, and CFLOBDDs is up to 8,192 qubits; however, CFLOBDDs perform better in terms of execution time.
For QFT, which is a function whose image is exponentially large, WCFLOBDDs perform better in time than WBDDs (but worse in space), and have substantially better performance in both time and space than CFLOBDDs.
For Grover's algorithm and Shor's algorithm, WBDDs perform better than both WCFLOBDDs and CFLOBDDs in both time and space.

In cases when WCFLOBDDs, WBDDs, and CFLOBDDs are all successful, the sizes of the final state vector of WCFLOBDDs and CFLOBDDs are similar, and smaller than WBDDs for GHZ, BV, DJ, and Simon's algorithm.
For QFT, Grover's algorithm, and Shor's algorithm, WBDDs represent the final state vector more compactly than WCFLOBDDs and CFLOBDDs.

The evaluation results for CFLOBDD differ from those in~\cite{TOPLAS:SCR24} because of 
(i) differences in the machine on which the experiments were run,
(ii) differences in the version of the circuits used for some benchmarks.
However, in all results reported in \figrefs{syntheticBenchmarks}{quantum-benchmarks}, WCFLOBDDs, WBDDs, and CFLOBDDs are compared on the same machine, and with the same implementations of the benchmarks (except for the decision-diagram package used).
Thus, our studies represent a fair comparison of the three data structures.

\section{Related Work}
\label{Se:related}

BDDs~\cite{toc:Bryant86} have been used in many tools for checking properties of hardware and software.
There are many variants of BDDs, e.g., MTBDDs~\cite{DBLP:journals/fmsd/FujitaMY97}, ADDs~\cite{DBLP:journals/fmsd/BaharFGHMPS97}, Free BDDs \cite[\S6]{DBLP:books/siam/Wegener00}, BMDs \cite{dac:BC95}, Hybrid DDs \cite{iccad:CFZ95}, Differential BDDs \cite{LNCS:AMU95}, and Indexed BDDs \cite{toc:JBAAF97}.
None of these BDD variants use anything like the ``procedure-call'' mechanism of CFLOBDDs \cite{TOPLAS:SCR24}, which we adopted in WCFLOBDDs.

Other variants of BDDs include
\twrchanged{
Edge-Valued BDDs (EVBDDs) \cite{vrudhula1996edge},
Factored Edge-Valued BDDs (FEVBDDs) \cite{tafertshofer1997factored}, 
Quantum Information Decision Diagram (QuIDDs)
\cite{DBLP:journals/qip/ViamontesMH03},
Quantum Multiple-valued DDs (QMDDs) \cite{DBLP:journals/tcad/NiemannWMTD16},
MQTDDs \cite{zulehner2019package}, and
Tensor Decision Diagrams (TDDs) \cite{hong2020tensor},
}
in which the edges out of decision nodes have weights, and some ``accumulation'' (e.g., + or $\times$) of the weights encountered on the path for an assignment $a$ yields the function's value on $a$.
\twrchanged{
QMDDs, QuIDDs, MQTDDs, and TDDs are all variants what we have been calling WBDDs (with complex-valued weights).
They are based on a Shannon decomposition, and thus the only sharing is of sub-DAGs, whereas the call-return structure used in WCFLOBDDs allows sharing of the ``middle of a DAG.''
}

\twrchanged{
WCFLOBDDs are also based on the idea of accumulating weights along (matched) paths, but combine it with the ``procedure-call'' mechanism of CFLOBDDs to get the best of both worlds.
The ''procedure-call'' idea originated with CFLOBDDs, but has not been previously investigated for weighted decision diagrams.
}
The algorithms for CFLOBDDs and WCFLOBDDs are similar in that they are typically divide-and-conquer algorithms that recurse on the A-connection and then the B-connections (or vice versa), thereby splitting the variables in half for each subproblem.  However, the WCFLOBDD algorithms require additional manipulations of weights---typically the propagation of additional weight values that represent accumulated products, often to set the top-level factor weight correctly.

\twrchanged{
Considered as graphs, (W)CFLOBDDs can contain cycles, although the matched-path principle on which they are based excludes all cyclic paths from consideration.
A number of BDD variants have also allowed cycles \cite{gupta1993representation,Thesis:Gupta94,ACSC:R99}.
A discussion of the differences between those structures and CFLOBDDs can be found in \cite[\S11]{TOPLAS:SCR24};
the observations made there apply to WCFLOBDDs as well.
}

There has been previous research on using
a hierarchical organization in compressed representations of functions~\cite{gupta1993representation, thierry2009hierarchical, darwiche2011sdd,nakamura2020variable}.
However, the hierarchy used in WCFLOBDDs is different than the kinds of hierarchical organization used in those works.
A given level-$0$ grouping in WCFLOBDDs is not associated with a specific Boolean
variable---see the Contextual Interpretation Principle in \sectref{BasicStructureOfWCFLOBDDs};
consequently,
a \protoWCFLOBDD can be reused for different variables.

\twrchanged{
The closest similarity between (W)CFLOBDDs and previous structures is with Sentential Decision Diagrams (SDDs) \cite{darwiche2011sdd} and Variable-Shift SDDs (VS-SDDs) \cite{nakamura2020variable}).
Like (W)CFLOBDDs, these data structures generalize BDDs by imposing a tree-structured ordering on variables, and there are functions for which these data structures are exponentially more succinct than BDDs.
SDDs and VS-SDDs represent Boolean functions only ($\mathbb{B}^n \rightarrow \mathbb{B}$), but have also been extended to represent probability distributions \cite{kisa2014probabilistic}.
}
WCFLOBDDs can represent semi-field-valued functions ($\mathbb{B}^n \rightarrow \mathcal{D}$).

\twrchanged{
As presented, WCFLOBDDs give up the freedom---available in SDDs and VS-SDDs---of using a variable decomposition based on arbitrary trees:
the variable decompositions used for CFLOBDDs \cite{TOPLAS:SCR24} and in this paper are balanced binary trees.
That is, at each level $k > 0$, (W)CFLOBDDs are based on a grammar with binary productions (where nonterminals are indexed by level) of the form
\[
  X_k \rightarrow X_{k-1} ~~ X_{k-1},
\]
with half of the variables interpreted in each child.
However, it is easy to generalize (W)CFLOBDDs to use arbitrary trees \cite[\S12]{DBLP:journals/corr/abs-2211-06818}.
All of the principles (variable decomposition, left-folding, canonicity, etc.) carry over if one were to use other grammars, e.g., ternary productions
\[
      X_k \rightarrow X_{k-1} ~~ X_{k-1} ~~ X_{k-1}
\]
or Fibonacci-like productions
\[
     X_k \rightarrow X_{k-1} ~~ X_{k-2}.
\]
One could also use more kinds of indexed nonterminals (e.g., $Y_k$, $Z_k$, etc.) and/or grammar rules of different forms at different levels.
In unpublished work, X.\ Zhi has results that vanilla CFLOBDDs can simulate all such variants based on indexed grammars with polynomial-space overhead.
}

\twrchanged{
VS-SDDs support a sub-structure sharing property that goes beyond SDDs;
however, as with BDDs and WBDDs, the objects that are shared are always sub-DAGs.
In contrast, the (W)CFLOBDD ``procedure-call'' device is a more flexible mechanism for reusing substructures than the mechanism for reusing substructures in VS-SDDs.
In particular, the call-return structure used in (W)CFLOBDDs allows sharing of the ``middle of a DAG.''
}

\twrchanged{
Local Invertible Map DDs (LIMDDs) \cite{DBLP:journals/quantum/VinkhuijzenCEDL23} are a DAG-structured decision diagram (i.e., no ``procedure calls'') in which two nodes are merged if their respective decision trees (considered as vectors) are equal up to multiplication by a tensor product of $2 \times 2$ matrices that represent single-qubit gates.
It has been shown that they can be exponentially more succinct that WBDDs---results that should carry over to WCFLOBDDs (i.e., there are functions for which LIMDDs are efficient, but WCFLOBDDs are not).
One drawback of LIMDDs is that every node-creation operation requires performing a cubic-time operation akin to Gaussian elimination to maintain the invariants that ensure canonicity;
consequently, LIMDDs seem to have a built-in handicap.
Each of the steps required to ensure that WCFLOBDDs are canonical are very fast in comparison (e.g., hashed lookups).
Little has been published about the performance of LIMDDs, but one empirical evaluation \cite{DBLP:conf/spin/VinkhuijzenGHBWL23} reported that QFT on a random stabilizer state took about 250,000 seconds for 24 qubits.
}

WCFLOBDDs closely resemble (weighted) finite, non-recursive, visibly pushdown automata (VPAs) or their close cousins nested-word automata (NWAs). 
There has been a significant amount of work on unweighted VPAs~\cite{DBLP:conf/stoc/AlurM04} and unweighted NWAs  \cite{DBLP:journals/jacm/AlurM09} (see also \cite{URL:NestedWords:Alur24,URL:VPAs:Madhu24}).
Weighted VPAs have been studied by \citet{DBLP:conf/dlt/CaralpRT12} and \citet{DBLP:conf/lata/LabaiM16};
weighted NWAs have been studied by
\citet{DBLP:journals/corr/abs-1001-2175} and---via conversion to weighted Pushdown Systems \cite{DBLP:conf/popl/BouajjaniET03,DBLP:conf/sas/RepsSJ03}---\citet{DBLP:conf/cav/DriscollTR12}.
In general, there do not exist VPAs for a visibly pushdown language that are canonical and minimal, except in a few special cases~\cite{alur2005congruences,kumar2006minimization}.
WCFLOBDDs are a restricted class of weighted VPAs with additional structural invariants (\defref{StructuralInvariants}, invariants (\ref{Inv:1})-(\ref{It:OutermostLevel:ValueTuple})), which enable WCFLOBDDs to support a unique canonical representation of each function in $\mathbb{B}^n \rightarrow \mathcal{D}$.

We used WCFLOBDDs for simulating quantum circuits.
Quantum-circuit simulation can be exact or approximate;
our focus here is on exact simulation (modulo floating-point round-off).
Tensor networks are a quantum-simulation method that is not based on decision diagrams.
CFLOBDDs perform better than tensor networks on
\twrchanged{
some benchmarks
}
\cite[Tab.\ 5]{TOPLAS:SCR24}.
Although we did not compare them with WCFLOBDDs directly, the previous study with CFLOBDDs implies that WCFLOBDDs perform better than tensor networks on
\twrchanged{
many benchmarks.
}
However, for problems such as random circuit sampling, tensor networks perform better than methods based on decision diagrams.
\section{Conclusion}
\label{Se:conclusion}

In this paper, we presented a new data structure, called \emph{Weighted Context-Free-Language OBDDs (WCFLOBDDs)}, an extension to CFLOBDDs, a recent innovation. 
CFLOBDDs are akin to BDDs, but use a hierarchical structure to represent functions more efficiently than BDDs.
However, similar to BDDs, CFLOBDDs also suffer from structure explosion when the
image of a function has an exponential number of different values.
Weighted BDDs (WBDDs) are used in place of BDDs to overcome this issue.
We introduced WCFLOBDDs to avoid problems of size explosion of CFLOBDDs, with the potential to have the good properties of both WBDDs and CFLOBDDs.
For two classes of functions, we showed that WCFLOBDDs are exponentially more succinct than any WBDD (respectively, CFLOBDD) representation of the functions in that class.
Our experiments compared the performance of WCFLOBDDs, WBDDs (using the MQTDD package), and CFLOBDDs on (i) synthetic benchmarks, and (ii) quantum-simulation benchmarks.
\twrchanged{
We found that on most benchmarks, WCFLOBDDs outperform WBDDs and CFLOBDDs in terms of the problem sizes that can be handled successfully in a fixed time period.
}
Overall, our results support the conclusion that, for at least some applications, WCFLOBDDs provide the best of both worlds: 
\twrchanged{
the problem sizes that WCFLOBDDs can handle roughly match whichever of WBDDs and CFLOBDDs is better.
(From the standpoint of running time, the results are more nuanced.)
}

\section*{Data-Availability Statement}
The artifact used to generate the experimental results in~\sectref{eval}, along with the source code, can be found on Zenodo~\cite{wcflobdds-artifact}. The source code for WCFLOBDDs can also be found at \url{https://github.com/trishullab/cflobdd}, with the code for experiments at \url{https://github.com/trishullab/Quasimodo}.

\begin{acks}
This work was supported, in part,
by a gift from
\grantsponsor{00001}{Rajiv and Ritu Batra}{}, and by
\grantsponsor{00003}{NSF}{https://www.nsf.gov/}
under grants
\grantnum{00003}{CCF-2211968}
and
\grantnum{00003}{CCF-2212558}.
Any opinions, findings, and conclusions or recommendations
expressed in this publication are those of the authors,
and do not necessarily reflect the views of the sponsoring
entities.
\end{acks}

\bibliographystyle{ACM-Reference-Format}
\bibliography{refs}

\clearpage
\appendix
\section{Appendices}
\subsection{Constructing a Canonical WCFLOBDD from a Decision Tree}
\label{Se:canonicalness}

The construction of a WCFLOBDD for a function $f$ from the decision-tree representation of $f$ is a two-step process: (i) constructing a weighted decision tree (WDT) from the decision tree for $f$, and (ii) constructing a WCFLOBDD from the WDT for $f$.

The construction of a WDT from a decision tree has the same flavor as constructing a WBDD, as discussed in~\cite{vrudhula1996edge} and~\cite{Book:ZW2020}.

\begin{Constr}{\bf [Decision Tree to WDT]}
\,\newline

Given a decision tree $T$ that represents a function $f$, perform the following actions recursively, traversing $T$ in post-order (i.e., performing actions just before returning), thereby creating the WDT bottom-up.
\begin{enumerate}
  \item
    If the current node $n$ is a leaf node of $T$ with terminal value $v$, create a new leaf node $n'$ with terminal value $v'$, where
    \begin{equation*}
      v' = \begin{cases}
             \bar{1} & \text{if $v \neq \bar{0}$}\\
             \bar{0} &\text{otherwise }
           \end{cases}
    \end{equation*}
    Return the tuple $\langle v, n' \rangle$.
  \item
    If $n$ is an internal node of $T$, and tuples $\langle v_1, n_1 \rangle$ and $\langle v_2, n_2 \rangle$ are the tuples returned from the recursive calls to the left child and right child, respectively, create a new internal node $n_0$ with (i) left child $n_1$ and right child $n_2$, and (ii) edge weights $(lw, rw)$ as follows:
    \begin{equation*}
      (lw, rw) = \begin{cases}
                   (\bar{1}, v_1^{-1} \cdot v_2) & \text{when $v_1 \neq \bar{0}$}\\
                   (\bar{0}, \bar{1}) &\text{otherwise }
                 \end{cases}
    \end{equation*}
    Let $v_0$ be defined as follows:
    \begin{equation*}
        v_0 = \begin{cases}
                v_1 & \text{when $v_1 \neq \bar{0}$}\\
                v_2 &\text{otherwise }
              \end{cases}
    \end{equation*}
    Return the tuple $\langle v_0, n_0 \rangle$.
\end{enumerate}
The result is a WDT $\langle v', n' \rangle$, where $v'$ is the factor weight.
\end{Constr}

\begin{Constr}\label{Constr:DecisionTreeToCFLOBDD}{\bf [WDT to WCFLOBDD]\/}
\,\newline

\begin{definition}\label{De:ProtoWDT}
A \emph{proto-Weighted Decision Tree} (proto-WDT) is a WDT whose ``leaves'' are either
(i) all terminal values $\bar{0}$ or $\bar{1}$, or (ii) all proto-WDTs that are of the same height and obey the edge-weight conditions of a WDT.
\end{definition}

Just as inductive arguments about WCFLOBDDs have to be couched in terms of proto-WCFLOBDDs, the main part of the construction below focuses on proto-WDTs of a given WDT ($T$), and shows how to construct a \protoWCFLOBDD from a proto-WDT.
That construction provides a way to construct a WCFLOBDD from a WDT as a special case.

One property of WDTs that is different for proto-WDTs is the number of equivalence classes of their leaves.
Suppose that WDT $T$ has $2^{2^n}$ leaves.
$T$ has at most two kinds of leaves, $\bar{0}$ and $\bar{1}$, and thus just one or two equivalence classes of leaves.
In contrast, suppose that $T'$ is a proto-WDT of $T$ of height $2^k$ and $2^{2^k}$ leaves.
$T'$ can have between $1$ and $2^{2^k}$ equivalence classes of leaves.

For convenience, in the discussion below, when we refer to a proto-WCFLOBDD $w$, we mean a proto-WCFLOBDD proper, together with a value tuple $v$ where $|v|$ is equal to the number of exit vertices of the head-grouping of $w$.

The following recursive procedure describes how to convert a proto-WDT $T'$ of height $2^k$ and $2^{2^k}$ leaves into a level-$k$ \protoWCFLOBDD whose value tuple consists of an enumeration of the leaf equivalence classes $e'$ of $T'$.
The enumeration of the $|e'|$ equivalence-class representatives respects the relative ordering of their first occurrences in a left-to-right sweep over the leaves of $T'$.
\begin{enumerate}
  \item 
    \label{Construct:BaseCases}
    Base case (when $k = 0$):
    There are two cases, depending on whether the height-$1$ proto-WDT has one leaf equivalence class,
    $\{ v \}$, or two, $\{ v_1 \}$ and $\{ v_2 \}$.
    In the first case, create a \texttt{DontCareGrouping} (a level-$0$ proto-WCFLOBDD), and attach $[v]$ as the value tuple.
    In the second case,  create a \texttt{ForkGrouping} (a level-$0$ proto-WCFLOBDD), and then attach $[v_1, v_2]$ as the value tuple.
    In both cases, the left-edge and right-edge weights of the level-$0$ grouping are copied from the weights of the left and right edges of the proto-WDT node.
  \item
    \label{Construct:Recursive1}
    For each of the $2^{2^{k-1}}$ proto-WDTs of height $2^{k-1}$ in the lower half of $T'$,
    construct---via a recursive application of the construction---$2^{2^{k-1}}$ level-$(k{-}1)$ {\protoWCFLOBDD}s.
    The leaf values of the height-$2^{k-1}$ proto-WDTs are carried over from the leaf values of $T'$.
    
    \hspace{1.5ex}
    These {\protoWCFLOBDD}s are then partitioned into some number $e^{\#} \geq 1$ of equivalence classes of equal {\protoWCFLOBDD}s.
    A representative of each class is retained, and the others discarded.
    Each of the $2^{2^{k-1}}$ ``leaves'' of the upper half of proto-WDT $T'$
    is labeled with the appropriate equivalence-class representative (for the subtree of the lower half of $T'$ that begins there).
    These \protoWCFLOBDD-valued leaves serve as the leaf values of the upper half of proto-WDT $T'$ when the construction process is applied recursively to the upper half in
    step~\ref{Construct:Recursive2}.

    \hspace{1.5ex}
    The enumeration $1 \ldots |e^{\#}|$ of the equivalence-class representatives respects the relative ordering of their first occurrences in a left-to-right sweep over the leaves of the upper half of $T'$.
  \item
    \label{Construct:Recursive2}
    Construct---via a recursive application of the procedure---a level-$(k\textrm{--}1)$ \protoWCFLOBDD $A'$ for the proto-WDT consisting of the upper half of $T'$ (with the WCFLOBDDs constructed in step \ref{Construct:Recursive1} as the leaf values).
  \item
    \label{Construct:Grouping}
    Construct a level-$k$ proto-WCFLOBDD from the level-$(k\textrm{--}1)$ {\protoWCFLOBDD}s created in steps~\ref{Construct:Recursive1} and~\ref{Construct:Recursive2}.
    The level-$k$ grouping is constructed as follows:
    \begin{enumerate}
      \item
        \label{Construct:NewAConnection}
        The $A$-connection points to the level-$(k\textrm{--}1)$ proto-WCFLOBDD of \protoWCFLOBDD $A'$ constructed in step~\ref{Construct:Recursive2}.
      \item
        \label{Construct:NewMiddleVertices}
        The $|e^{\#}|$ middle vertices correspond to the equivalence classes
        formed in step~\ref{Construct:Recursive1} (in the enumeration order
        $1 \ldots |e^{\#}|$ of step~\ref{Construct:Recursive1}).
      \item
        \label{Construct:NewAReturnTuple}
        The $A$-connection return tuple is the identity map back to the middle vertices (i.e., the tuple $[1..|e^{\#}|]$).
      \item
        \label{Construct:NewBConnections}
        The $B$-connections point to the level-$(k\textrm{--}1)$ proto-WCFLOBDDs of the $|e^{\#}|$ equivalence-class representatives constructed in step~\ref{Construct:Recursive1}, in the enumeration order $1 \ldots |e^{\#}|$.
      \item
        \label{Construct:NewExitVertices}
        The exit vertices of the proto-WCFLOBDD correspond to the equivalence classes $e'$ of the leaves of $T'$, in the enumeration order $1 \ldots |e'|$.
      \item
        \label{Construct:NewBReturnTuples}
        The $B$-connection return tuples connect (i) the exit vertices of the level-$(k{-}1)$ groupings of $B$-connections to (ii) the level-$k$ grouping's exit vertices that were created in step~\ref{Construct:NewExitVertices}.
        The connections are made according to matching leaf equivalence classes from $T'$.
      \item
        \label{Construct:CheckForDuplicates}
        Consult a table of all previously constructed level-$k$ groupings to determine whether the grouping constructed by steps~\ref{Construct:NewAConnection}--\ref{Construct:NewBReturnTuples}
        duplicates a previously constructed grouping.
        If so, discard the present grouping and switch to the previously constructed one;
        if not, enter the present grouping into the table.
    \end{enumerate}
  \item
    \label{Construct:AttachValueTuple}
    To create a \protoWCFLOBDD with value tuples from the proto-WCFLOBDD (without value tuples) constructed in step \ref{Construct:Grouping}, we use a value tuple consisting of the leaf equivalence classes $e'$ of $T'$, listed in enumeration order (i.e., in the ordering of first occurrences in a left-to-right sweep over the leaves of $T'$).
\end{enumerate}

To construct a WCFLOBDD for a WDT $T$, we merely consider $T$ to be a proto-WDT with each leaf labeled by $\bar{0}$ or $\bar{1}$.
Note that the leaf equivalence classes will be one of $[\bar{0}]$, $[\bar{1}]$, $[\bar{0}, \bar{1}]$, or $[\bar{1}, \bar{0}]$ (which becomes the value tuple of the constructed WCFLOBDD).
\end{Constr}

\subsection{Algorithms Used in Pointwise Multiplication}
\label{Se:AlgorithmsUsedInPointwiseMultiplication}

\begin{algorithm}[tb!]
\caption{Reduce\label{Fi:Reduce}}
\SetKwFunction{Reduce}{Reduce}
\SetKwProg{myalg}{Algorithm}{}{end}
\small{
\Input{Grouping g, ReductionTuple reductionTuple, ValueTuple valueTuple}
\Output{Grouping g' that is ``reduced,'' weight $w$ (factor weight for g')}
\Begin{
\label{Li:ReduceNoProcessingStart}
\tcp{Test whether any reduction actually needs to be carried out}
\label{Li:ReduceNoProcessingEnd}
\If{reductionTuple == $[$1..$\mid$reductionTuple$\mid]$ and each element in valueTuple equals $v \neq \bar{0}$}
{\Return [g, $v$]\;}
\lIf{every element in valueTuple == $\bar{0}$}
{\Return [ConstantZeroProtoCFLOBDD(g.level), $\bar{0}$]}

\If{g is fork grouping, reductionTuple = $[$1,2$]$ and let valueTuple = $[v_1, v_2]$}
{
\eIf{$v_1 == \bar{0}$}
{
ForkGrouping g' = new ForkGrouping($\bar{0}$, $\bar{1}$); \Return [g', $v_2$]\;
}
{
ForkGrouping g' = new ForkGrouping($\bar{1}$, $v_1^{-1} v_2$); \Return [g', $v_1$]\;
}
}

\If{g is fork grouping, reductionTuple = $[$1,1$]$ and let valueTuple = $[v_1, v_2]$}
{
\eIf{$v_1 == \bar{0}$}
{
DontCareGrouping g' = new DontCareGrouping($\bar{0}$, $\bar{1}$); \Return [g', $v_2$]\;
}
{
DontCareGrouping g' = new DontCareGrouping($\bar{1}$, $v_1^{-1} v_2$); \Return [g', $v_1$]\;
}
}

InternalGrouping g' = new InternalGrouping(g.level)\;
g'.numberOfExits = $|\{ x : x \in \textit{reductionTuple} \}|$\;
Tuple reductionTupleA = []; Tuple valueTupleA = []\;
\For{$i \leftarrow 1$ \KwTo $g.\textit{numberOfBConnections}$}{
\label{Li:BConnectionCollapseStart}
Tuple deducedReturnClasses = [reductionTuple(v) : v $\in$ g.BReturnTuples[i]]\;
\label{Li:AConnectionProcessingStart}
Tuple$\times$Tuple [inducedReturnTuple, inducedReductionTuple] = CollapseClassesLeftmost(deducedReturnClasses)\;
Tuple inducedValueTuple = [deducedReturnClasses(i) : i $\in$ [1..$|$deducedReturnClasses$|$]]\;
\label{Li:BConnectionCollapseEnd}
Grouping$\times$Weight [h, $w_B$] = Reduce(g.BConnection[i], inducedReductionTuple, inducedValueTuple)\;
int position = InsertBConnection(g', h, inducedReturnTuple)\;  \label{Li:CallInsertBConnection}
reductionTupleA = reductionTupleA $||$ position\;  \label{Li:ExtendReductionTupleA}
valueTupleA = valueTupleA $||$ $w_B$\; \label{Li:ExtendValueTupleA}
}
Tuple$\times$Tuple [inducedReturnTuple, inducedReductionTuple] = CollapseClassesLeftmost(reductionTupleA)\;  \label{Li:CallCollapseClassesLeftmost}
Tuple inducedValueTuple = [reductionTupleA(i) : i $\in$ [1..$|$reductionTupleA$|$]]\;
Grouping$\times$Weight [h', $w$] = Reduce(g.AConnection, inducedReductionTuple, inducedValueTuple)\;
g'.AConnection = h'\;
g'.AReturnTuple = inducedReturnTuple\;  \label{Li:AConnectionProcessingEnd}
\Return [RepresentativeGrouping(g'), $w$]\;
}
}
\end{algorithm}

\begin{algorithm}[tb!]
\caption{InsertBConnection\label{Fi:InsertBConnection}}
\small{
\Input{InternalGrouping g, Grouping h, ReturnTuple returnTuple}
\Output{int --  Insert (h, ReturnTuple) as the next B-connection of g, if they are a new combination; otherwise return the index of the existing occurrence of (h, ReturnTuple)}

\Begin{
\label{Li:InsertBConnectionStart}
\lIf{there exists $i \in [1..g.\mathrm{numberOfBConnections}]$ such that g.BConnections[i] == h $\&\&$ g.BReturnTuples[i] == returnTuple}{\Return i}
g.numberOfBConnections = g.numberOfBConnections + 1\;
g.BConnections[g.numberOfBConnections] = h\;
g.BReturnTuples[g.numberOfBConnections] = returnTuple\;
\Return g.numberOfBConnections\;
\label{Li:InsertBConnectionEnd}
}
}
\end{algorithm}

This section gives two of the algorithms used in pointwise multiplication for which there was insufficient space to present in the body of the paper.
\begin{itemize}
  \item
    \algref{Reduce} reduces a grouping based on return tuples or the value tuple, and as a by-product ensures canonicity of the resulting WCFLOBDD.
  \item
    \algref{InsertBConnection} determines the position for a B-connection in a grouping being constructed, reusing one if it is already present in the grouping.
\end{itemize}
\subsection{Pointwise Addition}
\label{Se:PointwiseAddition}

The pseudo-code for pointwise addition of two WCFLOBDDs is given in~\algrefs{ABO}{WeightedPairProduct}.

\begin{algorithm}[tb!]
\caption{Pointwise Addition\label{Fi:ABO}}
\small{
\Input{WCFLOBDDs n1 = $\langle \textit{fw1}, \textit{h1}, \textit{vt1} \rangle$, n2 = $\langle \textit{fw2}, \textit{h2}, \textit{vt2} \rangle$}
\Output{CFLOBDD n = n1 $+$ n2}
\Begin{
\tcp{Perform ``weighted'' cross product}
    Grouping$\times$Tuple [g,pt] = WeightedPairProduct($\textit{h1}$,$\textit{h2}$, $\textit{fw1}$, $\textit{fw2}$)\;  \label{Li:ABO:CallPairProduct}
    ValueTuple deducedValueTuple = [ c1 $\cdot$ vt1[i1] + c2 $\cdot$ vt2[i2]~:~[$\langle (c1, i1),(c2, i2) \rangle$] $\in$ pt ]\;  \label{Li:ABO:LeafValues}
     \tcp{Collapse duplicate leaf values, folding to the left}
    Tuple$\times$Tuple [inducedReturnTuple,inducedReductionTuple] = CollapseClassesLeftmost(deducedValueTuple) \;  \label{Li:ABO:CollapseLeafValues}
    Tuple inducedValueTuple = one of $[\bar{1}, \bar{0}], [\bar{0}, \bar{1}], [\bar{1}], [\bar{0}]$ based on inducedReturnTuple \;
    Grouping$\times$Weight [g', $\textit{fw}$] = Reduce(g, inducedReductionTuple, deducedValueTuple) \;  \label{Li:ABO:CallReduce}
    WCFLOBDD n = RepresentativeCFLOBDD($\textit{fw}$, g', inducedValueTuple) \;
    \Return n\;
}
}
\end{algorithm}

\begin{algorithm}[tb!]
\caption{WeightedPairProduct\label{Fi:WeightedPairProduct}}
\small{
\Input{Groupings g1, g2; Weights p1, p2}
\Output{Grouping g: g1 $+$ g2; Tuple ptAns: tuple of pairs of exit vertices with corresponding weights}
\Begin{
    \lIf {g1 is ConstantZeroCFLOBDD }  
    { \Return [ g2, [$\langle (\bar{1}, 1), (\bar{0}, k) \rangle$~:~k $\in$ [1..g2.numberOfExits]] }   
    \lIf {g2 is ConstantZeroCFLOBDD }  
    { \Return [ g1, [$\langle (\bar{0}, k), (\bar{1}, 1) \rangle$~:~k $\in$ [1..g1.numberOfExits]] }   
    \lIf{g1 == g2}
    { \Return [ g1, [$\langle (\bar{1}, k), (\bar{1}, k) \rangle$~:~k $\in$ [1..g2.numberOfExits]] }   
\tcp{Similar for other base cases, with appropriate weights and exit vertices pairings}
\If {g1 and g2 are fork groupings  \label{Li:WPPBothForkGroupings}} {
    ForkGrouping g = new ForkGrouping(1,1)\;
\Return [ g, [$\langle (p1 \cdot g1.lw, 1), (p2 \cdot g2.lw, 1) \rangle$, $\langle (p1 \cdot g1.rw, 2), (p2 \cdot g2.rw, 2) \rangle$] ]\;
} \label{Li:WPPBothForkGroupingsEnd}
\tcp{ Pair the A-connections}
    Grouping$\times$Tuple [gA,ptA] = WeightedPairProduct(g1.AConnection, g2.AConnection, p1, p2)\;
    InternalGrouping g = new InternalGrouping(g1.level)\;
    g.AConnection = gA \;
    g.AReturnTuple = [1..$|$ptA$|$]\tcp*[r]{Represents the middle vertices}
    g.numberOfBConnections = $|$ptA$|$ \;  \label{Li:WPPNumberOfBConnections}
    \tcp{Pair the B-connections, but only for pairs in ptA }
    \tcp{Descriptor of pairings of exit vertices}
    Tuple ptAns = []\;  \label{Li:WPPptAnsDeclaration}

    \tcp{Create a B-connection for each middle vertex}
    \For {$j \leftarrow 1$ \KwTo $|ptA|$}{
        Grouping$\times$Tuple [gB,ptB] = WeightedPairProduct(g1.BConnections[ptA(j)(1)(2)], g2.BConnections[ptA(j)(2)(2)], ptA(j)(1)(1), ptA(j)(2)(1))\;
        g.BConnections[j] = gB\;
        \tcp{Now create g.BReturnTuples[j], and augment ptAns as necessary}
        g.BReturnTuples[j] = []\;        \label{Li:WPPExitVertexLoopStart}
        \For{$i \leftarrow 1$ \KwTo $|ptB|$} {
            c1 = g1.BReturnTuples[ptA(j)(1)(2)](ptB(i)(1)(2))\tcp*[r]{an exit vertex of g1}
            c2 = g2.BReturnTuples[ptA(j)(2)(2)](ptB(i)(2)(2))\tcp*[r]{an exit vertex of g2}
            f1 = g1.BReturnTuples[ptA(j)(1)(2)](ptB(i)(1)(1))\tcp*[r]{an associated weight of g1}
            f2 = g2.BReturnTuples[ptA(j)(2)(2)](ptB(i)(2)(1))\tcp*[r]{an associated weight of g2}
            \eIf(\tcp*[f]{Not a new exit vertex of g}){$\langle (f1, c1), (f2, c2) \rangle$ $\in$ ptAns}
                {
                index = the k such that ptAns(k) == $\langle (f1, c1), (f2, c2) \rangle$\;
                g.BReturnTuples[j] = g.BReturnTuples[j] $||$ index\;
                }
            (\tcp*[f]{Identified a new exit vertex of g}){
                g.numberOfExits = g.numberOfExits + 1\;
                g.BReturnTuples[j] = g.BReturnTuples[j] $||$ g.numberOfExits\;
                ptAns = ptAns $||$ $\langle (f1, c1), (f2, c2) \rangle$\; \label{Li:WPPConcatenateExitPair}
            }
        }  \label{Li:WPPExitVertexLoopEnd}
    }
    \Return [RepresentativeGrouping(g), ptAns]\; \label{Li:WPPTabulateAnswer}
}
}
\end{algorithm}
\subsection{Kronecker Product}
\label{Se:KroneckerProduct}



The pseudo-code for Kronecker Product on WCFLOBDDs is given as \algrefs{KroneckerProductAlgo}{KroneckerProductOnGrouping}.
(In \algref{KroneckerProductAlgo}, $\bowtie$ denotes the operation to interleave two variable orderings.)

\begin{algorithm}[tb!]
\caption{Kronecker Product \label{Fi:KroneckerProductAlgo}}
\Input{WCFLOBDDs n1 = $\langle \textit{fw1}, g1, \textit{vt1} \rangle$, n2 = $\langle \textit{fw2}, g2, \textit{vt2} \rangle$ with variable ordering of n1: $x \bowtie y$ and n2: $w \bowtie z$}
\Output{WCFLOBDD n = n1 $\otimes$ n2 with variable ordering of n: $(x || w) \bowtie (y || z)$}
\Begin{
    \lIf{\textit{fw1} == $\bar{0}$ or \textit{fw2} == $\bar{0}$}
    {\Return ConstantZeroProtoCFLOBDD(n1.level)}
    e = index of $\bar{0}$ in $\textit{vt1}$ (-1 if no such occurrence)\;
    e' = index of $\bar{0}$ in $\textit{vt2}$ (-1 if no such occurrence)\;
    Grouping g = KroneckerProductOnGrouping(g1, g2, e, e')\;
    ValueTuple $vt$\;
    \lIf{e == -1}{$\textit{vt}$ = $\textit{vt2}$}
    \lIf{e == 1}{$\textit{vt}$ = $[\bar{0}, \bar{1}]$}
    \eIf{e == 2 and \textit{vt2} == [$\bar{1}$]}{
      $\textit{vt}$ = $\textit{vt2}$ $||$ [$\bar{0}$]\;
    }{
      $\textit{vt}$ = $\textit{vt2}$\;
    }
    \Return RepresentativeCFLOBDD($\textit{fw1} \cdot \textit{fw2}$, g, $\textit{vt}$)\;
}
\end{algorithm}

\begin{algorithm}[tb!]
\caption{KroneckerProductOnGrouping\label{Fi:KroneckerProductOnGrouping}}
\Input{Groupings g1, g2; e (exit of $g1$ that leads to $\bar{0}$), e' (exit of $g2$ that leads to $\bar{0}$)}
\Output{Grouping g' such that g' = g1 $\otimes$ g2}
\Begin{
    InternalGrouping g' = new InternalGrouping(g.level + 1)\;
    g'.AConnection = g1\;
    g'.AReturnTuple = [1..$|$g1.numberOfExits$|$]\;
    g'.numberOfBConnections = $|$g1.numberOfExits$|$\;
    \If{e == -1}{
        g'.BConnection[1] = g2\; 
        g'.BReturnTuple[1] = [1..$|$g2.numberOfExits$|$]\;
        g'.numberOfExits = $|$g2.numberOfExits$|$\;
    }
    \If{e == 1}{
        g'.BConnection[1] = ConstantZeroProtoCFLOBDD(g.level)\;
        g'.BReturnTuple[1] = [1]\;
        g'.BConnection[2] = g2\;
        \lIf{e' == -1} { g'.BReturnTuple[2] = [2] }
        \lIf{e' == 1}  { g'.BReturnTuple[2] = [1,2] }
        \lElse         { g'.BReturnTuple[2] = [2,1] }
        g'.numberOfExits = 2\;
    }
    \If{e == 2}{
        g'.BConnection[1] = g2\; 
        g'.BReturnTuple[1] = [1..$|$g2.numberOfExits$|$]\;
        g'.BConnection[2] = ConstantZeroProtoCFLOBDD(g.level)\;
        \lIf{e' == -1} { g'.BReturnTuple[2] = [2] }
        \lIf{e' == 1}  { g'.BReturnTuple[2] = [1] }
        \lElse         { g'.BReturnTuple[2] = [2] }
        g'.numberOfExits = 2\;
    }
    \Return RepresentativeGrouping(g')\;
}
\end{algorithm}

\subsection{Matrix Multiplication}

Pseudo-code for the matrix-multiplication algorithm for WCFLOBDDs is shown in~\algrefsp{MatrixMult}{MatrixMultGrouping}{MatrixMultOnGroupingContd}.

\begin{algorithm}[tb!]
\caption{Matrix Multiplication \label{Fi:MatrixMult}}
\Input{WCFLOBDDs n1 = $\langle \textit{fw1}, g1, \textit{vt1} \rangle$, n2 = $\langle \textit{fw2}, g2, \textit{vt2} \rangle$}
\Output{WCFLOBDD n = n1 $\times$ n2}
\Begin{
Grouping$\times$MatMultTuple$\times$Weight [g,m,w] = MatrixMultOnGrouping(g1, g2)\;
ValueTuple v\_Tuple = []\;   \label{Li:CallMatrixMultOnGrouping}
\For{$i \leftarrow 1$ \KwTo $|m|$ }{   \label{Li:ValueComputationStart}
    Value $v = \langle \textit{vt1}, \textit{vt2} \rangle(m(i))$\;  \label{Li:ConcreteWeightedDotProduct}
    v\_Tuple = v\_Tuple $||$ $v$;
}  \label{Li:ValueComputationEnd}
Tuple$\times$Tuple [inducedValueTuple, inducedReductionTuple] = CollapseClassesLeftmost(v\_Tuple)\;
Grouping$\times$Weight [g, $\textit{fw}$] = Reduce(g, inducedReductionTuple, v\_Tuple)\;
ValueTuple valueTuple = one of $[\bar{0}, \bar{1}], [\bar{1}, \bar{0}], [\bar{0}], [\bar{1}]$ based on inducedValueTuple \;
WCFLOBDD n = RepresentativeCFLOBDD($w \cdot \textit{fw} \cdot \textit{fw1} \cdot \textit{fw2}$, g, valueTuple)\;
\Return n\;
}
\end{algorithm}

\begin{algorithm}[tb!]
\small {
\caption{MatrixMultOnGrouping \label{Fi:MatrixMultGrouping}}
\Input{Groupings g1, g2}
\Output{Grouping$\times$MatMultTuple$\times$Weight [g,m,w] such that g = g1 $\times$ g2}
\Begin{
\If(\tcp*[f]{Base Case: matrices of size $2\times 2$}){g1.level == 1}{
\tcp{Construct a level $1$ Grouping that reflects which cells of the product hold equal entries in the output MatMultTuple}
}
\If {g1 or g2 is ConstantZeroProtoCFLOBDD}
{
\Return [ \twrchanged{ConstantZeroProtoCFLOBDD(g1.level), $[\{[(0,0), \bar{0}]\}]$}, $\bar{0}$ ]\;
}
\If {g1 is Identity-Proto-CFLOBDD}
{
m = [ \twrchanged{ $\{[(0, k), \bar{1}]\}$ : $k \in$} [1..g2.numberOfExits] ]\;
\Return [g2, m, $\bar{1}$]\;
}
\If {g2 is Identity-Proto-CFLOBDD}
{
m = [ \twrchanged{ $\{[(k, 0), \bar{1}]\}$ : $k \in$} [1..g1.numberOfExits] ]\;
\Return [g1, m, $\bar{1}$]\;
}
InternalGrouping g = new InternalGrouping(g1.level)\;
Grouping$\times$MatMultTuple$\times$Weight [aa,ma,wa] = MatixMultOnGrouping(g1.AConnection, g2.AConnection)\;  \label{Li:AConnectionRecursion}
g.AConnection = aa; g.AReturnTuple = [1..$|$ma$|$]\; 
g.numberOfBConnections = $|$ma$|$\;
\tcp{Continued in \algref{MatrixMultOnGroupingContd}}
}
}
\end{algorithm}

\begin{algorithm}[tb!]
\small {
\caption{MatrixMultOnGrouping (cont.) \label{Fi:MatrixMultOnGroupingContd}}
\setcounter{AlgoLine}{18}
\Input{Groupings g1, g2}
\Output{Grouping$\times$MatMultTuple$\times$Weight [g,m,w] such that g = g1 $\times$ g2}
\Begin{
\tcp{Interpret ma to (symbolically) multiply and add BConnections}
MatMultTuple m = []\;
Tuple valueTuple = []\;
\For(\tcp*[f]{Interpret $i^{\textit{th}}~\BP$ in ma to create g.BConnections[$i$]}){$i \leftarrow 1$ \KwTo $|ma|$ }{    \label{Li:InterpretationLoopStart}
    \tcp{Set g.BConnections[$i$] to the (symbolic) weighted dot product $\sum_{((k_1,k_2),v) \in \textrm{ma}(i)} v * \textrm{g1.BConnections}[k_1] * \textrm{g2.BConnections}[k_2]$}
    CFLOBDD curr\_cflobdd = ConstantCFLOBDD(g1.level, $[\ZeroBP]$)\;  \label{Li:WeightedDotProductStart}
    \For{$((k_1,k_2),v) \in$ ma($i$)}{   
        Grouping$\times$MatMultTuple$\times$Weight [bb,mb,wb] = MatrixMultOnGrouping(g1.BConnections[$k_1$], g2.BConnections[$k_2$])\;  \label{Li:BConnectionRecursion}
        MatMultTuple mc = []\;
        Tuple inducedValueTuple = []\;
        \For{$j \leftarrow 1$ \KwTo $|mb|$ }{
            $\BP$ $\bp = \langle \mathrm{g1.BReturnTuples}[k_1],$
            $\mathrm{g2.BReturnTuples}[k_2] \rangle(\mathrm{mb}(j))$\;
            mc = mc $||$ $\bp$;
            valueTuple\_B = valueTuple\_B $||$ wb\;
        }
        Tuple$\times$Tuple [inducedMatMultTuple, inducedReductionTuple] = CollapseClassesLeftmost(mc)\;
        [bb, $fw$] = Reduce(bb, inducedReductionTuple, valueTuple\_B)\;
        CFLOBDD n = RepresentativeCFLOBDD($\textit{fw}$, bb, inducedMatMultTuple)\;
        curr\_cflobdd = curr\_cflobdd + $v$ $\ast$ n \tcp*[r]{Accumulate symbolic sum}
    }  
    g.BConnection[$i$] = curr\_cflobdd.grouping\;
    g.BReturnTuples[$i$] = curr\_cflobdd.valueTuple\;
    m = m $||$ curr\_cflobdd.valueTuple\;  
    valueTuple = valueTuple $||$ curr\_cflobdd.factor\_weight \label{Li:WeightedDotProductEnd}
} \label{Li:InterpretationLoopEnd}
g.numberOfExits = $|$m$|$\;
Tuple$\times$Tuple [inducedMatMultTuple, inducedReductionTuple] = CollapseClassesLeftmost(m)\;           
[g, $fw$] = Reduce(g, inducedReductionTuple, valueTuple)\;
\Return [RepresentativeGrouping(g), m, $fw$]\;
}
}
\end{algorithm}
\subsection{Sampling}
\label{Se:Sampling}

A WCFLOBDD with no non-negative edge weights can be considered to represent a discrete distribution over the set of assignments to the Boolean variables.
An assignment---or equivalently, the corresponding matched path---is considered to be an elementary event.
The probability of a matched path $p$ is the weight of $p$ divided by the sum of the weights of all matched paths of the WCFLOBDD.

\subsubsection{Weight Computation.}
To sample an assignment directly from the WCFLOBDD representation of the function, we first need to compute the weight corresponding to every exit vertex.
The weight of an exit vertex $e$ of a grouping $g$ is the sum of the weights of all matched paths through the proto-WCFLOBDD headed by $g$ that lead to $e$.
This information can be computed recursively by (i) computing the weight of every middle vertex of $g$ for all matched paths from the entry vertex to middle vertices, and then (ii) computing the weight of every exit vertex of $g$ for all matched paths from the middle vertices to the exit vertices, (iii) combining this information to obtain the weight of every exit vertex of $g$ for matched paths from entry vertex to exit vertices of $g$.

Consider a grouping $g$ at level $l$ with $e$ exit vertices.
Suppose that $g.\textit{AConnection}$ has $p$ exit vertices;
suppose that $g.\textit{BConnections}[j]$ (where $1 \leq j \leq p$) has $k_j$ exit vertices;
and let $g.\textit{BReturnTuples}[j]$ be the return edges from $g.\textit{BConnections}[j]$'s exit vertices to $g$'s exit vertices.
To compute Step (i), we recursively call the weight-computation procedure for $g.\textit{AConnection}$, which yields a vector of weights $v_A$ of size $1 \times p$.
For Step (ii), the vectors obtained from recursive calls on the weight-computation procedure for the $p$ B-connections of $g$ are used
to create a matrix $M_B$ of size $p\times e$, in which the $j^{\textit{th}}$ row is the vector of weights from the $j^{\textit{th}}$ middle vertex of $g$ to $g$'s exit vertices.
(The details are given in the next paragraph.)
For Step (iii),
the vector-matrix product $v_A \times M_B$ yields $g$'s weight vector, of size $1\times e$.
At level-$0$, if the grouping is a \emph{ForkGrouping} with left-edge and right-edge weights $(\textit{lw}, \textit{rw})$, then the weight vector is $[\textit{lw}, \textit{rw}]$.
If the grouping is \emph{DontCareGrouping}, then the weight vector is $[\textit{lw} + \textit{rw}]$.

Because the exit vertices of $g.\textit{BConnections}[j]$ are connected to $g$'s exit vertices via $g.{\textit{BReturnTuples}}[j]$, the $j^{\textit{th}}$ row of $M_B$ is the product of the weight vector for $g.\textit{BConnections}[j]$ (of size $1\times k_j$) and a ``permutation matrix'' $\textit{PM\,}^{g.\textit{BReturnTuples}[j]}$ (of size $k_j \times e$).
Each entry of $\textit{PM}$ is either $0$ or $1$;
each row must have exactly one $1$;
and each column must have at most one $1$.

This definition can be stated equationally, where the expression in large brackets represents $M_B$.

\[
\small{
\begin{array}{@{\hspace{0ex}}l@{\hspace{0ex}}}
  \textit{weightOfExit\,}^g_{1 \times e} = 
         \begin{cases}
           [lw, rw]_{1\times2} & \text{if g = \textit{Fork}} \\
           [lw + rw]_{1\times1} & \text{if g = \textit{DontCare}} \\
           \begin{array}{@{\hspace{0ex}}l@{\hspace{0ex}}}
             \textit{weightOfExit\,}^{g.\textit{AConnection}}_{1\times p} \, \times \\
             \quad \begin{bmatrix}
                      \vdots\\
                      \textit{weightOfExit\,}^{g.{\textit{BConnections}[j]}}_{1\times k_j}   \\
                      \quad \times \textit{PM\,}^{g.{\textit{BReturnTuples}}[j]}_{k_j \times e}\\
                      \vdots
                    \end{bmatrix}_{\scriptsize {\begin{array}{@{\hspace{0ex}}l@{\hspace{0ex}}} {p \times e} \\ j \in \{1..p\} \end{array}}}
           \end{array} & \text{otherwise}
         \end{cases}
\end{array}
}
\]

Pseudo-code for the algorithm is given as \algref{WeightComp}.

\begin{algorithm}[tb!]
\caption{ComputeWeights \label{Fi:WeightComp}}
\Input{Grouping g}
\Begin{
\eIf{g.level == 0}{
    \eIf{g == DontCareGrouping}{
        g.weightsOfExits = [g.lw + g.rw]\;
    }(\tcp*[h]{g == $\textit{ForkGrouping}$}){
        g.weightsOfExits = [g.lw, g.rw]\;
    }
}
{
    ComputeWeights(g.AConnection)\;
    \For{$i \leftarrow 1$ \KwTo $g.\textit{numberOfBConnections}$}{
        ComputeWeights(g.BConnection[i]);
    }
    g.weightsOfExits = a $\bar{0}$-initialized array of length $\mid$g.numberOfExits$\mid$\;  \label{Li:path-counting-multiplication-start}
    \For{$i \leftarrow 1$ \KwTo $g.\textit{numberOfBConnections}$}{
        \For{$j \leftarrow 1$ \KwTo $g.\textrm{BConnection}[i].\textit{numberOfExits}$}{
            k = BReturnTuples[i](j)\;
            g.weightsOfExits[k] += \\
            \quad g.AConnection.weightsOfExits[i] * g.BConnection[i].weightsOfExits[j]\;
        }
    }   \label{Li:path-counting-multiplication-end}
}
}
\end{algorithm}

\subsubsection{Sampling.}
To sample an assignment from the probability distribution
efficiently,
we need to perform the operation directly on the WCFLOBDD that represents the probability distribution.
Suppose that the WCFLOBDD has $l$ levels.
If the distribution were given as a vector of weights, $W = [w_1,\cdots,w_{2^{2^l}}]$, then the probability of selecting the $p^{th}$ matched path would be
\[
  \textit{Prob}(p) = \dfrac{w_p}{\Sigma_{i=1}^{2^{2^l}}w_i}
\]

Because we do not have this information directly, instead of sampling one matched path, we will sample a set of matched paths that lead to an exit vertex. At top level, we will consider only those paths that lead to the terminal value $\bar{1}$. 
At every other grouping $g$, given an exit-vertex $e$, we will sample a path from all the matched paths that lead to $e$.

We take advantage of the structure of matched paths to break the assignment/path-sampling problem down into a sequence of smaller assignment/path-sampling problems that can be performed recursively.
At each grouping $g$ visited by the algorithm, the goal is to sample a matched path based on the weight of the matched path from the set of matched paths $P_{g,i}$
(in the proto-WCFLOBDD headed by $g$) that lead from $g$'s entry vertex to a specific exit vertex $i$ of $g$.

Consider a grouping $g$ and a given exit vertex $i$. For each middle vertex $m$ of $g$, the sum of weights of the matched paths from the entry vertex of $g$ to $i$ that passes through $m$ forms a distribution $D$ on $g$'s middle vertices. To sample a matched path from $P_{g,i}$, we (i) first sample the index ($m_{index}$) of a middle vertex of $g$ according to $D$, (ii) recursively sample on $g$.AConnection with respect to the exit vertex that leads to $m_{index}$, (iii) recursively sample on $g$.BConnection[$m_{index}$] with respect to the exit vertex that leads to $i$, (iv) concatenate the sampled paths to obtain the sampled path of $g$.

Only those B-connections of $g$ whose exit vertices are connected to $i$ contribute to the paths leading to $i$.
Therefore, to sample a middle vertex, we need to consider only those B-connection groupings that lead to $i$.
For such an $i$-connected B-connection grouping $k$, let $(g.\textit{BReturnTuples}[k])^{-1}[i]$ denote the exit vertex of $g.\textit{BConnections}[k]$ that leads to $i$;
i.e., $\langle j, i \rangle \in g.\textit{BReturnTuples}[k] \Leftrightarrow (g.\textit{BReturnTuples}[k])^{-1}[i] = j$.

Given the sum of weights of all matched paths leading to exit vertex $i$ as $\textit{weightsOfExits}[i]$, we can sample $m_{index}$ based on the following probability espression (where $g.A$ denotes $g.\textit{AConnection}$, $g.B[k]$ denotes $g.\textit{BConnections}[k]$, and $g.\textit{BRT}$ denotes $g.\textit{BReturnTuples}$):
\begin{equation}
\label{Eq:sampling-equation}
\small{
  \textit{Prob}(m_{\textit{index}})
  =
  \dfrac{\begin{array}{@{\hspace{0ex}}c@{\hspace{0.25ex}}l@{\hspace{0ex}}}
                  & \textit{weightsOfExit\,}^{g.A}[m_{\textit{index}}] \\
           \times & \textit{weightsOfExit\,}^{g.B[m_{\textit{index}}]}[(g.\textit{BRT}[m_{\textit{index}}])^{-1}[i]]
          \end{array}
        }{g.\textit{weightsOfExit}[i]}
}
\end{equation}
Using this process, $m_{\textit{index}}$ is selected.
The sampling procedure is called recursively on $g$'s A-connection, which returns an assignment $a_A$, and then on B-connection[$m_{\textit{index}}$] of $g$, which returns an assignment $a_B$.
The sampled path/assignment of grouping $g$ is $a = a_A||a_B$. 

At level-$0$ (the base case), if $g$ is a \texttt{DontCareGrouping}, the assignment is sampled from ``$0$'' and ``$1$'' in proportion to the edge weights $(\textit{lw}, \textit{rw})$.
If $g$ is a \texttt{ForkGrouping}, the assignment ``$0$'' or ``$1$'' is chosen according to the specified index $i$.
Pseudo-code for the sampling algorithm is given as~\algref{Sampling}.

\begin{algorithm}[tb!]
\caption{Sample an Assignment from a WCFLOBDD\label{Fi:Sampling}}
\SetKwFunction{SampleAssignment}{SampleAssignment}
\SetKwFunction{SampleOnGroupings}{SampleOnGroupings}
\SetKwProg{myalg}{Algorithm}{}{end}
\myalg{\SampleAssignment{n}}{
\Input{WCFLOBDD n = $\langle \textit{fw}, g, \textit{vt} \rangle$}
\Output{
Assignment sampled from n according to $vt$.
}
\Begin{
$i$ = the index of $\bar{1}$ in $vt$\;
Assignment a = SampleOnGroupings(g, $i$)\;
\Return a\;
}
}{}
\setcounter{AlgoLine}{0}
\SetKwProg{myproc}{SubRoutine}{}{end}
\myproc{\SampleOnGroupings{g, i}}{
\Input{Grouping g, Exit index $i$}
\Output{
Assignment for a path leading to exit $i$ of $g$, sampled in proportion to path weights 
}
\Begin{
\If{g.level == 0}{
    \eIf{$g$==\texttt{DontCareGrouping}}{
        \tcp{random(w1,w2) returns 1 if w1 is chosen, else w2}
        \Return random(g.lw, g.rw) ? ``0'' : ``1''\;
    }(\tcp*[h]{$g$==$\textit{ForkGrouping}, \textrm{so}\:i\,{\in}\,[1,2]$}){
        \Return ($i$ == 1) ? ``0'' : ``1''    
    }
}
{
Tuple WeightsOfPathsLeadingToI = []\;
\For(\tcp*[f]{Build weight info tuple from which to sample}){$j \leftarrow 1$ \KwTo $g.\textit{numberOfBConnections}$}{
    \If(\tcp*[f]{if $j^{th}$ B-connection leads to $i$})
    {$i \in g.\textit{BReturnTuples}[j]$}{
    WeightsOfPathsLeadingToI = WeightsOfPathsLeadingToI 
        $||$ (g.AConnection.WeightsOfExits[$j$] * g.BConnections[$j$].weightsOfExits[$k$]), where $i$ = BReturnTuples[$j$]($k$)\;
    }
}
$m_{\textit{index}} \leftarrow$ Sample(WeightsOfPathsLeadingToI)\tcp*[r]{Sample middle-vertex index $m_{\textit{index}}$}
Assignment a = SampleOnGroupings(g.AConnection, $m_{\textit{index}}$) $||$ SampleOnGroupings(g.BConnection[$m_{\textit{index}}$], $k$), where $i$ = BReturnTuples[$m_{\textit{index}}$]($k$)\;
\Return a\;
}
}
}
\end{algorithm}

\end{document}